\documentclass[11pt]{article}
\usepackage{amsmath,amssymb}
\usepackage{hyperref}
 \usepackage[left=1in,right=1in,top=1.1in,bottom=1in]{geometry}
\usepackage{bm}
\usepackage{float}
\usepackage{expl3}
\usepackage{makecell}
\usepackage{tikz}
\usepackage{appendix}
\usepackage[english]{babel}
\usepackage{bbm}
\usepackage{graphicx}
\usepackage[justification=centering]{caption}
\usepackage{amsfonts,amssymb,amsmath,latexsym,amsthm}
\usepackage{multirow}
\usepackage{enumerate}
\usepackage{geometry}
\usepackage{diagbox}
\usepackage[usenames,dvipsnames]{pstricks}
\usepackage{epsfig}
\usepackage{pst-grad} 
\usepackage{pst-plot} 
\usepackage[space]{grffile} 
\usepackage{enumitem}
\usepackage{etoolbox} 
\makeatletter 
\patchcmd\Gread@eps{\@inputcheck#1 }{\@inputcheck"#1"\relax}{}{}

\newtheorem{thm}{Theorem}[section]
\newtheorem{cor}[thm]{Corollary}
\newtheorem{lem}[thm]{Lemma}

\newtheorem{quest}[thm]{Question}

\newtheorem{claim}[thm]{Claim}

\newtheorem{defn}[thm]{Definition}

\newcommand{\genlang}{\Lambda}
\newcommand{\termlang}{\Lambda'}

\newcommand{\mX}{{\mathcal{X}}}

\newcommand{\mC}{{\mathfrak{C}}}

\newcommand{\mG}{{\mathcal{G}}}
\newcommand{\mA}{{\mathcal{A}}}
\newcommand{\mU}{{\mathcal{U}}}
\newcommand{\mR}{{\mathcal{R}}}

\newcommand{\mT}{{\mathfrak{T}}}

\newcommand{\Acc}{\mathcal{A}_{\text{acc}}}

\newcommand{\mS}{{\mathcal{S}}}

\newcommand{\disc}{\text{Disc}}
\newcommand{\Vol}{\text{Vol}}

\newcommand{\lan}{\mathsf{L}}

\newcommand{\lang}[1]{\mathsf{L}_{#1}}

\newcommand{\res}[2]{{#1}_{#2}}

\newtheoremstyle{plainupright}
  {\topsep}   
  {\topsep}   
  {\upshape}  
  {}          
  {\bfseries} 
  {.}         
  { }         
  {}          
\theoremstyle{plainupright}
\newtheorem{example}{Example}
\newtheorem{dfn}{Definition}[section]

\usepackage{changepage} 

\newtheoremstyle{indentedstyle}
  {\topsep}{\topsep}{\upshape}{0pt}{\scshape}{.}{.2em}{}

\theoremstyle{indentedstyle}
\newtheorem{innerremark}{Remark}

\newenvironment{remark}
  {\begin{adjustwidth}{1.5em}{0pt} 
   \begin{innerremark}}
  {\end{innerremark}
   \end{adjustwidth}
   \addvspace{1.5em}} 

\let\svthefootnote\thefootnote
\newcommand\blankfootnote[1]{%
	\let\thefootnote\relax\footnotetext{#1}%
	\let\thefootnote\svthefootnote%
}

\newcommand{\omt}[1]{}

\newcommand\supproof[1]{}

\def\coll{{\mathcal{X}}}
\def\seen{S}
\def\fint{t^*}
\def\zt{t^+}

\def\trueL{{K}}

\def\out{o}

\def\lbanach{\delta_B}

\DeclareMathOperator{\LCA}{LCA}

\ExplSyntaxOn

\ExplSyntaxOff

\begin{document}

\pagenumbering{gobble}

\title{Validity, Sparse Holes, and Breadth in Language Generation: Banach Density, Topology, and Geometry}

 \author{Jon Kleinberg\thanks{Department of Computer Science and Information Science, Cornell University, Ithaca NY 14853 USA.  Supported in part by a Vannevar Bush Faculty Fellowship, AFOSR grant FA9550-23-1-0410, a Simons Collaboration grant, and a grant from the MacArthur Foundation.} \and Fan Wei\thanks{Department of Mathemaics, Duke University, 120 Science Drive, Durham, NC 27710, USA. Research supported by NSF grant DMS-2401414. } }

\maketitle

\begin{abstract}
Language generation in the limit, rooted in the classical work of Gold
\cite{gold1967language} and Angluin
\cite{angluin1979finding, angluin1980inductive, angluin1983inductive}
and recently revived by Kleinberg and Mullainathan
\cite{kleinberg2024limit}, gives a clean model for studying language
generation with minimal assumptions. An adversary reveals strings from an
unknown target language \(K\), and the algorithm must eventually generate
new unseen strings from \(K\). A central open question raised by Kleinberg
and Mullainathan is the tension between validity and breadth: can a
generator avoid hallucination in the limit while still covering the hidden
language broadly, rather than collapsing onto a small part of it?

In our previous work \cite{kleinberg2025density, kleinberg2026density}, we proved that if breadth is measured by lower asymptotic density, which looks at prefixes of a fixed ordering of the target language, then every countable language
class admits the optimal \(1/2\) guarantee. In this sense, validity and
breadth are always compatible for an averaged, prefix-based measure of
density.

This paper asks what happens when breadth is measured locally. In many
settings, strings are represented in an embedding space, and a broad
generator should not miss an entire large region of related strings. We
capture this using lower Banach density, which checks all long intervals
or boxes rather than only prefixes of one ordering. We show that the
positive averaged picture does not survive this local test: for some
countable language collections, every eventually valid generator must
leave arbitrarily large sparse holes in the true language. Thus a
worst-case local form of mode collapse is unavoidable: sparse holes are
not always just a flaw of a poorly designed generator; in the worst case,
they are forced by the requirement of valid generation.

The main result is that lower Banach density is not just a stricter
version of lower asymptotic density. It changes what the breadth problem
can see: prefix asymptotic density averages away topological ambiguity,
while Banach density can turn that ambiguity into sparse holes. 
This density notion changes the nature of the problem: it turns breadth into topology and combinatorics in one dimension. In dimensions \(d\ge 2\), the
theory also becomes geometric: Ramsey/discrepancy phenomena can force
zero ordinary Banach density even for a singleton language class. We
introduce a filtered lower Banach density that removes this geometric
obstruction and prove a \(1/2-\epsilon\) guarantee for finite-rank
classes. We also extend the one-dimensional theory to \(f\)-window
densities, interpolating between lower asymptotic and lower Banach
density.

\end{abstract}

\newpage

\tableofcontents

\newpage

\pagenumbering{arabic}
\setcounter{page}{1}

\section{Introduction}\label{sec:intro}

The growing power of large language models (LLMs) has
led to a range of theoretical investigations of their
properties and performance.
These formalisms operate at different levels of
abstraction, ranging from concrete models of the
transformer architectures that power current LLMs
\cite{chen2025transformer,peng-transformer,sanford-transformer,strobl2024formal}
to more abstract, architecture-independent models that
analyze properties of the language generation problem itself
\cite{kalai2023calibrated,kleinberg2024limit}.

In this paper, we work with a model in the latter style
called {\em language generation in the limit}. It takes a \textbf{basic, assumption-free} perspective on language generation, rooted in the foundational work of Gold and Angluin \cite{angluin1979finding,angluin1980inductive,gold1967language}.
In this model, we consider language generation as a game
played between an adversary and an algorithm as follows.
First, the adversary chooses a secret language $K$ known only to
come from a countable list of candidate languages
$L_1, L_2, L_3, \ldots$, where each language $L_i$
is an infinite set of finite strings.
The adversary then enumerates the strings of the language $K$,
one string per time step $t = 1, 2, 3, \ldots$.
In each time step $t$, the algorithm can see the set of strings
$S_t \subseteq K$ that the adversary has enumerated so far
(in all the steps up to $t$), and the algorithm must output
a string $o_t$ with the goal that $o_t \in K - S_t$: 
that is, $o_t$ should be a string
that is in the language $K$ but that has not been seen before.
The algorithm wins the game, achieving generation in the limit, 
if there is some time step $t^*$ such that $o_t \in K - S_t$ for all
$t \geq t^*$.

This model was proposed in \cite{kleinberg2024limit} as an 
adaptation of a classical model of language learning due to
Gold and Angluin
\cite{angluin1979finding,angluin1980inductive,gold1967language};
in their earlier model the algorithm's goal was {\em identification
in the limit}: rather than outputting a string $a_t$, it produced
an index $i_t$ in each time step $t$, with the goal that
for all $t \geq t^*$, we have $L_{i_t} = K$.
That is, the algorithm's goal in this case was to exactly identify
the language.
The results for identification in the limit are largely negative,
whereas the surprising result of \cite{kleinberg2024limit} is
that there is an algorithm that achieves language generation in the limit
for {\em all} collections $L_1, L_2, L_3, \ldots$.
In other words, generation in the limit is a much more tractable
problem than identification in the limit.
Since this result, a number of papers have considered different
aspects of the problem 
\cite{kalavasis-stoc25, kalavasis2024breadth, charikar-pabbarju, bai2025noise, raman2025noisy, li2025learning, hanneke2025union},
several of which we'll discuss in more detail below.

\subsection{Central Question: Tension between Breadth (Density) and Validity}

A language generator is useful only if it balances  two basic goals:
\begin{itemize}
    \item {\em validity}:
that it should eventually only generate strings in the true language $K$ with no hallucination, and
\item {\em breadth}: that it should generate many well-representative strings from $K$.
\end{itemize}

These two goals pull in opposite directions. If the algorithm is too
aggressive, it risks outputting strings outside \(K\), which is the limiting
analogue of hallucination. If it is too conservative, it may remain valid but
cover only a narrow part of \(K\), which is the limiting analogue of mode
collapse.

This paper studies this  \textbf{tension between validity and breadth} in the model of
language generation in the limit, a central question raised by Kleinberg and Mullainathan\cite{kleinberg2024limit}. Explicitly, 
\begin{quote} 
$(\star) \qquad$
 How broadly
can a valid generator cover the hidden target language? Does guarenteeing validity means breadth has to be sacrificed?  
\end{quote}

In our previous work \cite{kleinberg2026density}, we gave a strong positive answer
 for a prefix averaged
notion of breadth: if breadth is measured by lower asymptotic density, then
every countable language class admits the optimal \(1/2\) guarantee. In that
sense, validity and essentially perfect breadth are always compatible when breadth is measured on
prefixes of a fixed ordering.

Our question is what happens when breadth is local rather than averaged. If
strings are represented in an embedding space, then a generator should not miss
an entire large region of related strings. Such a missing region is a {\it sparse
hole}. We show that, for this local notion of breadth, the positive averaged
picture breaks: for some countable language collections, every eventually valid
generator must leave arbitrarily large sparse holes in the true language. \textbf{In the worst case, a local form of mode collapse is forced by the
requirement of valid generation. } Thus the validity--breadth tension $(\star)$ raised
by Kleinberg and Mullainathan \cite{kleinberg2024limit} is realized mathematically for local breadth.

This is a special case of our more general result. The \textbf{mathematical message} is that lower Banach density is not just a stricter
density notion. It changes what breadth can see. Prefix asymptotic density
averages away topological ambiguity; Banach density can turn that ambiguity
into sparse holes. In one dimension, the nature of the problem is governed by
 \textbf{the topology of the language collection}. In higher dimensions,
 \textbf{Ramsey/discrepancy geometry} also enters the problem.

\

We now explain this in more details.
The notion of breadth has been formulated in several different ways
in recent work \cite{charikar-pabbarju,kalavasis-stoc25,kleinberg2025density}, and our starting point is the density-based formulation of breadth introduced in our previous work \cite{kleinberg2025density}.

Informally, if $A \subseteq B$ are both countable sets
(with the elements of $B$ listed in some order as
$B = \{b_1, b_2, b_3, \ldots\}$), then the
density of $A$ in $B$ is determined by looking at the fraction
of the first $n$ elements of $B$ that are contained in $A$, and 
taking the limit of this fraction as $n \rightarrow \infty$.
These fractions might not converge to a limit as $n$ grows, and so 
we must look at the lim inf separately from the lim sup.
Formally, we say that the {\em lower asymptotic density} of $A$ in $B$ is
\begin{equation}
  \underline{d}(A,B) = \liminf_{n \to \infty} 
  \frac{|A \cap \{b_1, b_2, b_3, \ldots, b_n\}}{n}. 
\label{def:intro-asymptotic}
\end{equation}
(There is an analogous notion of {\em upper asymptotic density}
based on using the $\limsup_{n \to \infty}$ in place of
$\liminf_{n \to \infty}$ in Equation (\ref{def:intro-asymptotic}),
In this paper, we focus primarily on lower density. Because the upper density is naturally bounded below by the lower density, this approach yields a more robust mathematical guarantee.) As a notational point, when we write $\underline{d}(A)$ without a second parameter, we will take this to be the lower asymptotic density of $A$ in the natural numbers, $\underline{d}(A,\mathbb{N})$.

In \cite{kleinberg2025density,kleinberg2026density}, 
we used lower asymptotic density as a notion of breadth for a language generation algorithm:
if $O$ is the set of all strings in the true language $K$ that
the algorithm ever generates, then 
the lower asymptotic density $\underline{d}(O,K)$
of $O$ in $K$ serves as a kind of breadth measure for how much of
$K$ is covered by the generated set $O$.
The original algorithm for language generation in the limit 
from \cite{kleinberg2024limit} achieves lower asymptotic density zero
on some instances (implying a kind of ``vanishing breadth''), 
we proved in
\cite{kleinberg2025density,kleinberg2026density}
that an algorithm can achieve lower asymptotic density $1/2$
on all instances, which is the best constant possible.

This is the baseline for our work. Lower asymptotic density gives a universal
averaged law: every countable language class admits the optimal \(1/2\)
guarantee. The word ``averaged'' is important. Since the density is measured
on prefixes of one ordering, a sparse region far from the origin can be
compensated for by generated strings elsewhere.

The sharper question is

\begin{quote}What happens when breadth is measured on all large windows?
\end{quote}

\paragraph{Main Message:} This is the point at which the choice of density becomes decisive.
As we observed in \cite{kleinberg2025density}, identification in the
limit is intrinsically topological: finite evidence determines which
hypotheses remain possible. A guiding message of this paper is that the
 breadth question can become topological as well, but only for
the right density notion. 
We show that
the local requirement on all  intervals changes the problem:  it makes the universal breadth guarantee remember the topology of the language collection. In higher
dimensions, the same viewpoint also reveals a separate geometric barrier.

\subsection{Local breadth-- sparse holes and Banach density; the hidden structure of breadth}\label{subsec:introbanach}

There are two different reasons to use density as a measure of breadth. The
first is ranking-based: if the strings of \(K\) are ordered by importance or
frequency, then it is natural to ask how many of the first \(n\) strings are
generated. This leads to lower asymptotic density.

The second is local and embedding-based. A generator may have high density on
prefixes and still miss an entire large region of related strings. If strings
are represented in a coordinate space, this failure has a concrete witness: a
large interval or box \(R\) such that $R$ contains many target strings
but  contains few generated outputs. Such an \(R\) is a sparse
hole. Lower Banach density is the density notion that controls sparse holes,
because it tests all large windows rather than only prefixes.

The point of this paper is that the second rationale is not only
practically motivated by embeddings; it is also where the mathematical
structure of breadth becomes visible.

\paragraph{1. The logic of ranking: asymptotic lower density and the universal averaged law.}
The lower asymptotic density is based on a distinguished ordering of the
elements of \(K\), say \(K=\{w_1,w_2,w_3,\ldots\}\). It asks how many of
the first \(n\) elements of this ordering are contained in the generated
set \(O\). This is a natural question when the ordering represents
priority, importance, or frequency. For example, in the case of human languages,
the ordering might correspond to a ranking of the words by frequency
\cite{newman2005power}.

This is the logic of ranking. It is very important how many of
\(w_1,\ldots,w_{100}\) are generated. It is less important, under this
measure, whether the generator misses a block such as
\(w_{10001},\ldots,w_{10100}\), provided this loss is compensated by
enough generated strings earlier in the ordering.
Viewed this way, the lower asymptotic density is essentially asking:
what fraction of the $n$ most important strings in $K$ are
in the generated set $O$?

In \cite{kleinberg2026density},  we proved that this averaged notion of breadth has a
striking universal law: every countable language class admits the optimal
lower asymptotic density guarantee \(1/2\). This theorem is the baseline
for our work. It shows that averaged breadth is always possible.

\paragraph{2. The logic of spatial embedding: Lower Banach Density and the hidden structure of breadth.}

This second rationale has both practical and mathematical force, and it is the one developed in
this paper. A genuine failure of breadth need not be merely that the
generator misses many strings in aggregate in an interval from the beginning. It can have a local
certificate: an arbitrarily large finite region of the target language
containing few or no generated outputs.
The lower asymptotic density, which is a prefix-based density, can average away such a certificate. If a long
interval far from the origin is missed, lower asymptotic density may
still be large because the generator produced many strings earlier in
the ordering. But if breadth is meant to rule out large sparse regions
wherever they occur, then density must be measured over all large
windows, not only over prefixes, as all large boxes are comparably important, because
they simply represent different parts of the embedding space that 
need to be covered. 

This is precisely lower Banach density.

The mathematical question is:
for every large coordinate window \(R\), what fraction of \(K\cap R\)
is generated?

In the discrete coordinate model justified below, this is precisely a
question about axis-parallel rectangles in \(\mathbb Z^d\). We therefore
define breadth using lower Banach density.
\begin{defn}\label{def:banach-intro}
   Given two infinite subsets $A \subseteq B \subseteq \mathbb{Z}^d$,
we define the {\em lower Banach density} of $A$ in $B$ as
\[\lbanach(A,B) = \liminf_{|R \cap B| \to \infty}   \frac{|A \cap B \cap R|}{|B \cap R|}.
\] 
where $R$ ranges over all axis-parallel boxes in $\mathbb{Z}^d$. 
\end{defn}
 (Here, we use the term ``box" in a loose sense, allowing side lengths to differ. 
 We also address the case where one restricts to boxes with equal side lengths.)
When specialized to dimension $d = 1$,
i.e. the one-dimensional case where the underlying ground set is 
$\mathbb{Z}$, the axis-parallel boxes just become intervals.
 Thus a guarantee
\(\delta_B(O,K)\ge \alpha\) says that every sufficiently large window of
the target language contains an \(\alpha-o(1)\) fraction of generated
outputs.

This is why lower Banach density is the correct notion for the local question. The
question is not only whether the generated set is large on average; it is
whether every large local region of the target language is represented. Once
the density test is local in this sense, finite evidence and topological
ambiguity can be amplified into sparse holes. This is the mechanism behind
our one-dimensional dichotomy. In higher dimensions, the same local test also
sees a second obstruction: rectangular Ramsey/discrepancy geometry.

\paragraph{Interpretation: embedded underrepresentation.}
The same window-based rationale also has an interpretation in the context of current research in large language models.
This interpretation begins from the
observation that language models often represent strings by points in a
\(d\)-dimensional embedding space, with related strings mapped to nearby
or similarly located points
\cite{mikolov2013word2vec,pennington2014glove,pereira1993distributional}. 
When LLMs are evaluated, this evaluation proceeds not just at the level of the embedding of individual words, but through {\em sentence embeddings} that map arbitrarily long pieces of output text to a fixed-dimensional space, again preserving the principle that similar pieces of text are mapped to nearby points \cite{cer2018universal,pillutla2021mauve,reimders2019sentence}.  In this way, it becomes possible to view all the outputs of an LLM in the same fixed-dimensional space, with regions of this space representing similar pieces of text, as we do in our model here.

When an LLM's output is viewed in terms of these sentence embeddings, an important goal is to understand 
whether large coordinate regions of the target
language are well-represented among the outputs or poorly represented.
Indeed, this goal has also been a central issue in the development of deployed language models since at least the era of GPT-2 and continuing through to the present \cite{holtzman2020degeneration,pillutla2021mauve,shypula2025diversity,zhang2025verbalized}.
The issue in this line of work, in part, is to determine whether a large language model is systematically underrepresenting some region of its embedding space, and to try mitigating this effect, thereby restoring diversity of outputs and working against the possibility of mode collapse.
These considerations thus form an important part of the underlying motivation for our study of large sparse regions in a generated set of strings given by their embedded representations in a fixed dimension.

The LLM community tends to abstract sentence embeddings as mappings into \(\mathbb R^d\), but given that these models operate in quantized or fixed-precision form, they are in fact mapping into a discrete lattice that is well-modeled as the lattice \(\mathbb Z^d\).
This use of \(\mathbb Z^d\) also accords with our goals in studying breadth via density.
Since the universe of strings is countable, the
set of embedded strings that can ever appear in the candidate languages
is countable. Our density notion is also a counting density: inside a
window, we count what fraction of the target strings have been generated.
Thus the relevant question is not Lebesgue measure in the ambient
\(\mathbb R^d\), but counting strings inside coordinate windows.

For axis-parallel windows, the information used by the density is exactly
coordinate-wise order information. A window is determined by choosing a
lower and upper cutoff in each coordinate, and then counting which
strings lie between those cutoffs. Therefore, for this question, the
absolute real-valued coordinates are not the primitive object. The
primitive object is the relative coordinate order of the embedded
strings.

Thus, for discrete embedded strings, measuring counting density over
axis-parallel windows in \(\mathbb R^d\) is exactly the same as measuring
counting density over axis-parallel rectangles in \(\mathbb Z^d\). This
is why we formulate the model on \(\mathbb Z^d\). The lattice is thus a natural modeling choice even beyond the principle from applications that embeddings use fixed precision; more generally, 
it is the
normal form of the discrete coordinate-order information used by
window density.
This is also why we use axis-parallel boxes rather than balls: boxes are
the coordinate windows determined by cutoff values in each coordinate,
whereas balls depend on additional metric scale information that is not
used by the counting-density question studied here.

The mathematical force of the lower Banach density, as will be clear later in the paper, is that this density notion changes the nature of the problem: it turns breadth into topology and combinatorics in one dimension, and into topology plus Ramsey/discrepancy geometry in higher dimensions.

\paragraph{Lower Banach density is less forgiving than lower asymptotic density.}

In one dimension, this transition is already visible. We always have
\[
\delta_B(A,B)\leq \underline d(A,B),
\]
and there are sets with \(\delta_B(A,B)=0\) but \(\underline d(A,B)=1\).
For example, take \(B=\mathbb N\), and let \(A\) be obtained from
\(\mathbb N\) by deleting the intervals \([2^k,2^k+k]\) for all \(k\).
Prefix density averages these holes away; Banach density records them
as certificates of failure.

More qualitatively, 
lower Banach density $\lbanach(A,B)$ provides a translation-invariant notion of uniform largeness that is unavailable for asymptotic density and imposes strictly stronger structural constraints. Unlike asymptotic density, which evaluates frequency relative to a distinguished origin, lower Banach density captures a quantitative lower bound on the density of a set across all arbitrarily long intervals. The distinction between these metrics is substantial in additive combinatorics (see  e.g. \cite{Ruzsa2009}). Within the broader study of amenable groups, Banach densities have become the standard framework for formulating quantitative sumset estimates, such as Plünnecke-type inequalities \cite{Bjorklund2013PlunneckeIF}.

As in the original language-generation-in-the-limit model, our results
are information-theoretic. The algorithms may use the countable
description of the language class and are not intended as efficient
finite-sample or polynomial-time procedures. The goal is to understand
when validity and window-uniform breadth are compatible in principle;
computational and statistical efficiency are left for future work.

\section{Main Results Overview: Lower Banach Density, Topology, and Geometry}

Thus, after this order-index normalization discussed above, our candidate languages
\(L_1,L_2,L_3,\ldots\) are countable subsets of \(\mathbb Z^d\). A
natural breadth goal is then the following:

\begin{quest}
    Can we achieve such a representativeness
guarantee (i.e., no sparse holes in the generation) for all countable collections of languages
$L_1, L_2, L_3, \ldots$, each equipped with an embedding into $\mathbb{Z}^d$?
\end{quest}
A sparse window means that the generator systematically underrepresents some
coordinate region of the target language. The practical stake is that validity
alone may force this failure: our lower bounds construct language classes for
which every eventually valid generator has arbitrarily large windows of the
true language with vanishing output density. In this worst-case local sense,
mode collapse is unavoidable. The positive results then shows a large class of language collections 
that avoid this failure and recover the optimal \(1/2\) benchmark.

The \textbf{sharp contrast} is:
\begin{enumerate}
\item {Lower asymptotic density gives a \textbf{topology-independent} universal \(1/2\) law. \cite{kleinberg2026density}} 
\item  Lower Banach density is the first place where \textbf{topology} enters the breadth guarantee.
\item {in dimensions \(d\ge2\), \textbf{ geometry / Ramsey} enters as well for lower Banach density.}
\end{enumerate}

\subsection{Main results at a glance}

The results can be summarized as follows. The table separates three
issues: the windows tested by the density notion, the universal guarantee
we prove, and the source that governs failure.

For definitions, see Definition~\ref{def:banach-intro} for lower Banach density,
Section~\ref{sec:topo} for Cantor--Bendixson rank, Definition ~\ref{def:filBanach}
for filtered lower Banach density, and Definition~\ref{def:fdensity} for
\(f\)-window density.

\begin{center}
\begin{tabular}{|p{0.18\linewidth}|p{0.18\linewidth}|p{0.31\linewidth}|p{0.23\linewidth}|}
\hline
Breadth notion & Windows tested & Universal guarantee & What can force failure \\
\hline
Lower asymptotic density
& Intervals starting from beginning (prefixes), or axis-parallel squares centered at zero in $\mathbb{Z}^d$
& Universal optimal \(1/2\) guarantee for all countable classes \cite{kleinberg2026density}
& No topology or geometry obstruction at the level of countable classes \\
\hline
Lower Banach density, \(d=1\)
& All long intervals
& Finite Cantor--Bendixson rank gives optimal \(1/2\); infinite-rank examples can force zero density
& Topology: infinite Cantor--Bendixson rank can force failure; finite rank rules it out \\
\hline
Ordinary lower Banach density, \(d\ge2\)
& All large axis-parallel rectangles
& Zero density can be forced even for a singleton language class
& Geometry: rectangular Ramsey/discrepancy barrier, independent of learning the target \\
\hline
Filtered lower Banach density, \(d\ge2\)
& Nondegenerate axis-parallel rectangles
& After removing the purely geometric barrier, finite rank gives \(1/2-\varepsilon\); infinite-rank examples can still force zero
& Topology after geometric degeneracy is filtered out \\
\hline
\(f\)-window density (Def \ref{def:fdensity}), \(d=1\)
& Intervals \([m,m+k-1]\) with \(1\le m\le f(k)\)
& Bounded \(f\): universal \(1/2\); unbounded \(f\): finite rank gives \(1/2\), while infinite-rank examples can force zero
& Topology appears exactly in the unbounded-window regime \\
\hline
\end{tabular}
\end{center}

Technically, the finite-rank one-dimensional theorem replaces the dynamic critical-chain analysis of prior lower-asymptotic-density work with a static tree extracted from the Cantor–Bendixson topology. And after some careful set-up,  an interval combinatorial game will imply the optimal density bound. In higher dimensions, the interval game is replaced by a size-sensitive discrepancy lemma for infinite point sets, which supplies the analogue of one-dimensional alternation on nondegenerate rectangles.

\subsection{Main Results in More Details}
\subsubsection{First set of results: Lower Banach density of generated sets in one dimension}
It is useful to start with the one-dimensional case before moving
to higher dimensions, since it is cleaner and contains much 
(though far from all) of the complexity.
In the one-dimensional case, we have a countable collection
of languages $\mX$ where each candidate language $L_i \in \mX$
is a subset of $\mathbb{Z}$ (Without loss of generality, we could assume the countable list of underlying strings is $\mathbb{Z}$).

Our \textbf{first result} is a negative one:

\begin{thm}\label{thm:intro-lowerdensity0fBanach}
There exists a countable collection of languages $\mX$ such that any algorithm that generates in the limit will produce a generated set $O$
with lower Banach density $\lbanach(O,K) = 0$. 
\end{thm}

This is the first main answer to the validity--breadth tension question for local
breadth. For lower asymptotic density, prior work shows that every countable
language class admits the optimal \(1/2\) guarantee. As a \textbf{strong contrast}, for lower Banach density,
this is no longer true: there are countable language collections for which
validity forces sparse holes. Equivalently, for every algorithm that eventually
generates only valid strings, there are arbitrarily large intervals of the
true language \(K\) that the algorithm intersects arbitrarily sparsely.

Theorem~\ref{thm:intro-lowerdensity0fBanach} already reveals the rich structure that arises when working with lower Banach density compared to lower asymptotic density.

Our \textbf{second result} shows that finite Cantor--Bendixson rank rules out the
topology-driven zero-density phenomenon. In every finite-rank class, the
optimal \(1/2\) lower Banach-density guarantee can be recovered.
Conversely, Theorem 1.1 shows that infinite rank is a real source of
hardness: there are infinite-rank classes for which every valid generator
has lower Banach density zero.

To define the complexity precisely,  we first need some topological definitions.
In \cite{kleinberg2025density,kleinberg2026density}, we defined a topological 
space on a collection of languages $\mX$
in which the points of the space are the languages $L_i \in \mX$
and the open sets are generated by the basis of sets
\[
U_{\lan,F} \;=\; \{\, \lan' \in \mX \mid F \subseteq \lan' \subseteq \lan \,\}.
\]
where $\lan$ ranges over all languages in $\mX$ and 
$F$ ranges over all finite subsets of $\mathbb{Z}$.
Convergent sequences in this topology have a natural meaning
(they correspond to the {\em infinite perfect towers} from
\cite{kleinberg2025density}), and from the limit points of 
this space we can define the {\em Cantor-Bendixson rank}
in the standard way:
if $d(\mX)$ denotes the set of all limit points of $\mX$, and
we define $\mX^{(0)}=\mX$, and $\mX^{(1)}=d(\mX)$,
and in general 
$\mX^{(i+1)} = d\bigl(\mX^{(i)}\bigr)$,
then the Cantor-Bendixson rank is the least ordinal $\alpha$ such that
$\mX^{(\alpha)}=\emptyset$.
If no such $\alpha$ exists, its Cantor-Bendixson rank is $\infty$.

In our previous work~\cite{kleinberg2025density, kleinberg2026density}, we established that topology plays a central role in characterizing language identification in both the the original model of language generation in the limit and its partial-enumeration variant (where the adversary need only reveal an infinite subset of the target language), as well as in determining when the algorithm can guess the correct
language index infinitely often. While this topology appears to be the appropriate one for these identification tasks, it is largely irrelevant for proving element-based breadth when breadth is measured via lower asymptotic density, since optimal breadth can always be attained under that measure. 

Our next result identifies the connection between generation breadth and topology. 
The Banach-density notion is aligned with the finite-window nature of
the difficulty. The topology records how candidate languages can remain
indistinguishable under finite evidence, while Banach density asks whether this ambiguity can be amplified into large low-density intervals far from the prefix. This
is the sense in which Banach density reveals structure that prefix
density averages away.

We can now state the positive result precisely. Finite Cantor–Bendixson rank rules out the zero-density failure exhibited above. Equivalently, any language class that forces every valid generator to have lower Banach density smaller than $1/2$ (the optimal benchmark) must lie beyond finite Cantor–Bendixson rank. Precisely, we prove:

\begin{thm}\label{thm:intro-finite-cb-rank}
There is an algorithm with the property that for any collection
of languages $\mX$ with finite Cantor-Bendixson rank, the
algorithm achieves generation in the limit with an output set $O$
that has lower Banach density at least $1/2$ in $K$, which is the optimal. 
\end{thm}

Together, Theorems \ref{thm:intro-lowerdensity0fBanach} and \ref{thm:intro-finite-cb-rank} show why lower Banach density is
mathematically different from the universal lower-asymptotic-density
theory. Lower asymptotic density gives the same optimal \(1/2\) guarantee
for all countable classes. Lower Banach density makes the topology of the
language collection relevant: finite rank recovers the optimal guarantee,
and infinite rank is rich enough to support genuine zero-density examples. 
It remains an interesting open question to determine  which infinite-rank classes nevertheless admit positive Banach-density generation, and what is the density generated. 

In summary, the lower Banach density thus reveals a deeper topological structure in the language generation problem, one that is invisible to formulations based on lower asymptotic density, and highlights an {\textbf{intrinsic connection between lower Banach density and the underlying topology.}}

\subsubsection{Second set of results: Lower Banach density of generated sets in  higher dimensions}
We now return to the higher-dimensional case, in which we have
a collection of candidate languages $\mX$, where each $L_i \in \mX$
corresponds to an infinite set of points in $\mathbb{Z}^d$,
and where the lower Banach density is defined in terms
of axis-aligned $d$-dimensional boxes.

First, the negative result in
Theorem \ref{thm:intro-lowerdensity0fBanach} carries
over directly to any dimension $d \geq 1$, since we can 
place the one-dimensional construction used for this theorem
along a coordinate axis of $\mathbb{Z}^d$ for any $d$.

We can also prove a positive result in higher dimensions
for collections with finite Cantor-Bendixson rank, but we
need to refine our definitions slightly to avoid a 
genuine combinatorial obstacle that only appears in dimensions $d > 1$.

In particular, we show how to construct a collection $\mX$
consisting of a single language $L_1$ in $\mathbb{Z}^d$ 
(for any $d \geq 2$)
such that the best lower Banach density achievable by any algorithm is zero.

The higher-dimensional case shows that \textbf{Banach lower density exposes not only
topology but also geometry}. In one dimension, if the target language is
known, the generator can alternate with the adversary and obtain density
\(1/2\) on every long interval. In dimensions \(d\ge2\), no such
alternation balances all axis-parallel rectangles in the most adversarial set of languages. Thus ordinary \textbf{Banach
density detects a second intrinsic barrier: the geometry itself}.

\begin{thm}\label{intro:strictbanachHD}
Suppose \(d\ge2\). There exists an infinite language
\(K\subseteq \mathbb Z^d\) such that, for the singleton class
\(\mathcal X=\{K\}\), every generator can be forced by some adversarial
enumeration of \(K\) to produce an output set \(O\) satisfying
\[
\liminf_{\substack{|R\cap K|\to\infty\\ R\text{ axis-parallel rectangle}}}
\frac{|O\cap K\cap R|}{|K\cap R|}=0.
\]
\end{thm}

It is clear that the identity of the language is not in doubt
for this collection,
since $L_1$ is the only option, but as the adversary and the algorithm
alternate turns generating string from $L_1$, we show how
to design an adversary that uses a subtle strategy based on
a Ramsey-theoretic result of
Chen, Pach, Szegedy, and Tardos
\cite{Chen2008DelaunayGO} in order to create larger and larger
homogeneous regions that an algorithm will miss completely.

We will show that this pathology can be handled by imposing 
the following mild nondegeneracy condition on the
$d$-dimensional boxes we consider.
Ordinarily we'd define the lower Banach density by
consider the $\liminf$ over all sequences of axis-aligned boxes
$R_1, R_2, R_3, \ldots$ where $|R_n \cap K| \rightarrow \infty$.
But our Theorem \ref{intro:strictbanachHD} tells us that isn't
going to work as is; on the other hand, it is sufficient to 
rule out the phenomena they identify by focusing on 
axis-aligned boxes that have a lightweight guarantee on the number of 
elements of $K$ they contain.
Specifically, we say that a 
sequence of axis-aligned boxes 
$R_1, R_2, R_3, \ldots$ 
is {\em nondegenerate} with respect to the true language $K$ 
if there exists a function $f(x)$ going to infinity with $x$
such that 
$$\frac{|R_n \cap K|}{f(\Vol(R_n)) \log(\Vol(R_n))} \rightarrow \infty.$$
For example, if $f(x) = \log(x)^\epsilon$, we retain only those boxes $R_n$ satisfying
$|R_n \cap K| \gg \log^{1+\epsilon}(\Vol(R_n)).$
This is a very weak requirement: a box $R_n$ contains at most $\Vol(R_n)$ elements, and its maximum side length is at most $\Vol(R_n)^{1/d}$, yet we only require the intersection size to exceed $\log \Vol(R_n)$, which is even smaller than the logarithm of its largest side length.

We then define the 
{\em filtered lower Banach density} of the generated set $O$
in the true language $K$ to be the $\liminf$ over all
nondegenerate sequences:
\begin{equation}
    \liminf_{nondegenerate ~ R_1, R_2, R_3, ...} \frac{|R_n \cap K \cap O|}{|R_n \cap K|}. \label{def:filterBD}
\end{equation}

We can then prove the following.

\begin{thm}\label{thm:intro-finite-cb-rank-d}
For any $\epsilon > 0$, there is an algorithm with the property that for any collection $\mX$
of languages in $\mathbb{Z}^d$ with finite Cantor-Bendixson rank, the
algorithm achieves generation in the limit with an output set $O$
that has filtered lower Banach density (defined in (\ref{def:filterBD})) at least $1/2 - \epsilon$ in $K$, which is essentially the optimal bound.
\end{thm}
This bound $1/2-\epsilon$ matches, up to $\epsilon$, the learning-free singleton benchmark after the geometric obstruction in Theorem \ref{intro:strictbanachHD} has been filtered out.

Thus the higher-dimensional theorem separates two obstructions. Ordinary Banach density mixes the learning/topological difficulty with a rectangular Ramsey/discrepancy obstruction that already appears when the target language is known. Filtered Banach density removes this singleton obstruction, and finite rank then essentially recovers the $1/2$ benchmark.

Viewed in light of the work on underrepresentation and mode collapse in language models \cite{holtzman2020degeneration,shypula2025diversity,zhang2025verbalized}, our results here provide a way of articulating these issues in the framework of language generation in the limit, and raising the question of whether this might provide any insights into the underrepresentation phenomena seen in practice. 

\subsection{Proof ideas}

The mathematical guiding point of the proof is that the Banach-density breadth
problem is not merely a harder density problem; it is intrinsically topological
and combinatorial. In one dimension, the right object is a
carefully designed containment-refined topology on the language
collection: finite Cantor--Bendixson rank guides the positive algorithm,
while infinite rank supplies the obstruction. The proof then turns this
topology into density guarantees using combinatorial tools: static trees,
local pods, and an interval game. In dimensions \(d\ge2\), a second mathematical layer appears:  Ramsey/discrepancy phenomena create an obstruction that is independent of learning the target language. This obstruction motivates the filtered  notion, and the positive theorem requires a size-sensitive discrepancy lemma for infinite point sets.

The lower Banach density is a local
validity--breadth problem. If the algorithm tries to fill every possible
window too aggressively, it risks outputting strings that are invalid for some
still-consistent target language. If it remains valid, the adversary may be
able to force a sparse window.

We now give a proof overview. The main results have four parts:
a one-dimensional infinite-rank density 0, a one-dimensional finite-rank
positive theorem, a higher-dimensional geometric obstruction, and a
higher-dimensional filtered positive theorem.

The starting key ingredient of all these proofs is a carefully designed containment-restricted topology on the set of languages. This topology is finer than the topology we used in \cite{kleinberg2026density} for language identification. Its Cantor–Bendixson rank gives a hierarchy of the candidate languages, which the positive algorithm turns into a static pullback tree.

\paragraph{Infinite-rank can force 0 density}
We show how infinite Cantor--Bendixson rank can force
lower Banach density zero. The adversary's goal is not merely to make the
algorithm miss infinitely many isolated strings. Banach density requires
a stronger failure: arbitrarily long intervals of the true language must
contain essentially no algorithmic outputs.

The construction builds a countable language family with an infinite
perfect-tower structure. Roughly speaking, during a mega step, the adversary tries to
occupy one long interval level by level. At level \(q\), the adversary
attempts to extend an already occupied block by one more point. If the
algorithm blocks this by outputting too far ahead, the adversary moves
to a higher-rank approximating language. If such blocking continued
forever, these approximating languages would converge to a limit language
on which the far-ahead outputs would eventually be invalid, contradicting
generation in the limit. Hence, at some level, the adversary must succeed
in extending the occupied interval. Repeating this at larger and larger
scales produces arbitrarily long missed intervals, forcing lower Banach
density zero.

The perfect-tower structure is what keeps this diagonal process inside a
fixed countable family. A fully adaptive decision tree would naturally
create uncountably many possible limit languages; the towers encode the
long gaps while preserving countability.

\paragraph{Finite-rank positive theorem in one dimension.}
For finite Cantor--Bendixson rank, we prove that the optimal
\(1/2\) lower Banach-density guarantee can be recovered. The algorithm starts from the accurate-guessing algorithm we introduced in 
\cite{kleinberg2025density}: at each time \(t\), it guesses a candidate
language, and this guess is correct infinitely often
(see Theorem~\ref{thm:acc}). This is also the starting point of our
lower asymptotic density algorithms in \cite{kleinberg2026density}.

The new difficulty is that accurate guesses alone do not control arbitrary
intervals. In the prefix-density setting, missed late strings can be
charged to generated strings earlier in the prefix. For Banach density,
that charge is invalid: a sparse interval must be repaired inside that
interval. The algorithm therefore needs both (1) a global structure that
controls how hypotheses change and (2) a local mechanism that protects the
windows where the adversary is acting.

The global structure is a carefully defined static tree extracted from the
Cantor--Bendixson topology, and a delicate rule of how the guesses move. Roughly, when the current hypothesis moves
laterally in this tree, the algorithm pulls back to a common ancestor.
Each such pullback moves upward in Cantor--Bendixson rank, and finite
rank can bound how many such moves can affect any fixed interval.

The local mechanism is a pod construction. Whenever the adversary reveals
a string, the algorithm adds a growing neighborhood around that string
inside the appropriate pullback language. These pods are local rather
than prefix-based: they protect the region where the adversary is acting,
instead of only prioritizing early strings in a fixed ordering. We can show that, because of our careful design, after a finite time, these pullbacks and
pods remain compatible with the true language. The algorithm always outputs a string near the adversarial input inside
the relevant pod, and the validity analysis shows that these outputs
eventually lie in the true language.

After proper but involved set-up, the density analysis of the finite-rank algorithm is reduced to a finite
game on an interval \(I\). Adversarial inputs in \(I\) are black stones,
and algorithmic outputs in \(I\) are white stones. The main way the
algorithm loses density inside \(I\) is through leaks: times when an
adversarial input in \(I\) causes the algorithm to output outside \(I\).

At each time, the algorithm can only place a white stone in the currently
visible part of the interval, which is determined by the current pullback
language and the pods. The static Cantor--Bendixson tree implies that
this visibility structure changes only a bounded number of times. Between
two such changes, repeated leaks must move geometrically farther away in
the relative order of the current visible set. Hence each phase has only
\(O(\log |I|)\) leaks. Since the number of phases is bounded by the
finite rank, the total leakage is \(o(|I|)\), giving density
\(1/2-o(1)\) on every long interval.

The actual language-generation process is more complicated than the game:
many events can occur between two adversary moves involving \(I\). 
The point of the setup is that the game abstraction is robust to these
intervening events; they can change the visibility structure only through
the bounded number of pullbacks controlled by the Cantor--Bendixson rank.

\paragraph{Higher-dimensional geometric obstruction.}
In dimensions \(d\ge2\), ordinary lower Banach density has a new
obstruction that is not caused by learning the target language. We
construct a singleton language class for which every generator can be
forced to have ordinary lower Banach density zero. Since the language is
known, the obstruction cannot come from identification. It comes from
rectangular geometry: axis-parallel rectangles in dimensions at least two
cannot always be balanced by a one-dimensional alternation strategy. A
Ramsey/discrepancy construction produces large rectangles that are
homogeneous with respect to the adversary/algorithm coloring, forcing
zero ordinary Banach density.

\paragraph{Filtered positive theorem in higher dimensions.}
The higher-dimensional obstruction shows that ordinary Banach density is
too sensitive to geometrically degenerate boxes. We therefore introduce
filtered lower Banach density, which ignores boxes whose intersection
with the target language is too sparse relative to the discrete volume
of the box.

After filtering, the finite-rank positive theorem can be recovered in
higher dimensions. The proof again uses the static Cantor--Bendixson-tree
strategy, but the one-dimensional interval game is replaced by a
discrepancy argument for rectangles. The key new ingredient is a
size-sensitive discrepancy lemma for infinite point sets. Classical
discrepancy bounds are usually finite and depend on a global parameter
\(n\). Here the target language is infinite, and the density guarantee
must hold locally over boxes of many different sizes and locations. The
lemma gives a coloring whose error in a box is negligible compared with
\(|K\cap R|\) whenever \(R\) is nondegenerate. This supplies the
higher-dimensional replacement for one-dimensional alternation and gives
the \(1/2-\varepsilon\) filtered Banach-density guarantee.

\subsection{Proof novelty and comparison with prior work}

A key is that here mathematics changes compared to prior work \cite{kleinberg2026density} studying asymptotic lower density. The one-dimensional theory is topology plus combinatorics:
the containment-refined topology organizes the learning obstruction, and
the interval game turns it into a density bound. The higher-dimensional
theory adds Ramsey/discrepancy geometry, which both creates the ordinary
Banach-density obstruction and motivates the filtered positive theorem.

We highlight the main ways in which the present paper departs from our previous works on lower asymptotic density
\cite{kleinberg2025density,kleinberg2026density}.

As in \cite{kleinberg2025density,kleinberg2026density}, the positive algorithms in this paper start from the index-guessing algorithm of we introduced in \cite{kleinberg2025density}. This algorithm maintains, at each time step,
a candidate language that it treats as the current guess for the true
language, and this guess is correct infinitely often; see
Theorem~\ref{thm:acc} for the precise statement.

The rest of the design and analysis is different. The reason is that
lower Banach density is not an amortized prefix quantity. In the
lower-asymptotic-density setting, a missed late string can be charged to
generated strings earlier in the prefix
\cite{kleinberg2025density,kleinberg2026density}. Such charging is no
longer valid for Banach density: if the adversary creates a sparse
interval or box, outputs elsewhere do not help. The algorithm must
control every large window locally. In particular, it must sometimes
generate far from the origin, near the region where the adversary is
revealing points, while still maintaining validity in the limit. This is
where the present algorithms begin to depart from the earlier ones.
\paragraph{Topology.}
A first major novelty is that the topology needed for Banach-density
breadth is not simply the finite-evidence topology used for
identification in \cite{kleinberg2026density}. In identification in the limit, finite evidence records
which candidate languages remain possible: after the adversary has shown
a finite set \(F\), every language containing \(F\) is still consistent.
This finite-evidence topology is valuable for identification, but it is
too coarse for generation density analysis. In particular, it does not
separate a language from its supersets, and therefore it does not record
the containment refinements along which the algorithm must pull back.

This issue does not arise for lower asymptotic density. Our universal
\(1/2\) lower-asymptotic-density algorithm of
\cite{kleinberg2026density} does not need to use the full topological
structure of the language collection; its analysis is based on prefix
charging and inclusion-chain bookkeeping. For lower Banach density, this
amortized charging is no longer valid. A sparse interval cannot be
repaired by outputs elsewhere, so the algorithm must control how
hypotheses move inside a fixed window.

For this reason, the present paper uses a containment-refined
finite-evidence topology. Its basic neighborhoods are
\[
U_{L,F}=\{L'\in\mathcal X:F\subseteq L'\subseteq L\}.
\]
Here \(F\) records the finite evidence already seen, while the condition
\(L'\subseteq L\) records which refinements remain inside the current
candidate language \(L\). This containment refinement is one of the main
new ingredients of the paper. It is what allows us to build the static
Cantor--Bendixson tree and to prove that, when the rank is finite, only
boundedly many pullbacks can affect any fixed interval.

\paragraph{Static topological tree and guessing rule.}
The positive finite-rank theorem uses a static tree extracted from the
Cantor--Bendixson topology. This is a major departure from the
lower-asymptotic-density algorithms, which rely on dynamic chains of
critical languages. For Banach density, a dynamic chain is not enough:
to bound the damage inside an arbitrary interval, we need a fixed
hierarchy in which lateral movement forces a pullback, and finite rank
bounds the number of such pullbacks.

The guessing rule also changes. Rather than following
only the current dynamic critical chain, the algorithm uses the tree to carefully 
decide when to pull back to a common ancestor. This is what gives the
analysis a uniform bound on the number of times the algorithm can move
upward across the structure relevant to a fixed interval.

\paragraph{Local pods rather than prefixes.}
Both the previous papers and the present paper use a notion of ``pod'':
a set of strings added to a priority list so that they have priority in
future generations. However, the role of pods is different here. In the
prefix-density setting, pods can preferentially support early strings in
the ordering. For Banach density, this is insufficient, since distant
intervals matter just as much as early ones.

Our algorithm therefore builds local pods around adversarially revealed
strings inside the appropriate pullback language. This
locality is essential: the algorithm must place mass near the windows
where sparse holes could otherwise form. The parameter controlling the number of new elements to be absorbed into a pod grows with
both time and location due to some technical subtlety, rather than only with time as in previous works.

\paragraph{The interval game.}
In the current paper, the density analysis reduces to a novel finite game on an interval. The
adversary places black stones, the algorithm places white stones, and
the visibility set changes only a bounded number of times, controlled
by the Cantor--Bendixson rank. Within each phase, leaking moves force
distances to grow geometrically, so there are only \(O(\log \ell)\)
leaks on an interval of length \(\ell\). This yields \(1/2-o(1)\)
density on every long interval.

\paragraph{Higher-dimensional geometry and size-sensitive discrepancy lemma.}
In dimensions \(d\ge2\), the topological story remains, but Banach
density also exposes a second barrier. Even for a singleton language
class, rectangular ranges cannot always be balanced. A Ramsey/discrepancy
construction forces large monochromatic rectangles, giving zero ordinary
Banach density. The filtered density removes this purely geometric
barrier, and the positive theorem then combines the static topological
tree with a new size-sensitive discrepancy lemma for infinite point sets. This ingredient has no analogue in our earlier one-dimensional algorithms \cite{kleinberg2025density, kleinberg2026density},
where intervals can be balanced by alternation and no rectangular
discrepancy issue arises.

\subsection{Extensions: Generalized notion of density and generation speed}

We also study a one-dimensional family of densities interpolating between
the asymptotic lower density (universal prefix regime) and the Banach regime. 
\begin{defn}\label{def:fwindow}
   For a nondecreasing
function \(f:\mathbb N\to\mathbb N\cup\{\infty\}\), define \emph{\(f\)-window lower density} as
\[
\delta_f(A)=\liminf_{k\to\infty}
\min_{1\le m\le f(k)}
\frac{|A\cap [m,m+k-1]|}{k}.
\] 
\end{defn}
Here \(f(k)\) is the maximum allowed starting position of a window of
length \(k\). Thus \(f(k)=1\) recovers lower asymptotic density, while
\(f(k)=\infty\) recovers lower Banach density. Another example is the lower P\'olya density, which is the limit of the lower $f$-window densities
for $f(k) = ck$ as $c \rightarrow \infty$.

From a structural perspective, the distinctions among these notions
are governed by the extent to which the viewing window may be
translated relative to its length.  The function \(f(k)\) determines the maximum allowed starting position of a window of length \(k\), thereby restricting its translational freedom.

The $f$-window densities form a hierarchy due to the following natural
monotonicity property that is easily verified:
if $f$ and $g$ are both non-decreasing functions from 
$\mathbb{N}$ to $\mathbb{N} \cup \{\infty\}$,
and if $f \leq g$, then 
$\delta_f(A) \geq \delta_g(A)$
for any set $A \subseteq \mathbb{N}$.

\subsubsection{Third set of results: dichotomy for lower $f$-window densities in one dimension}

We show the same transition appears in this family of densities. If \(f\) is uniformly
bounded, the density remains in the prefix regime and the universal
optimal \(1/2\) guarantee applies to all countable classes. If
\(f(k)\to\infty\), the density becomes window-uniform enough to see the
topological layer: finite Cantor--Bendixson rank gives \(1/2\), while
our infinite-rank construction forces density zero.

We summarize this pair of dichotomies in Table \ref{tab:main_table},
where ``rank of $\mX$'' refers to the Cantor-Bendixson rank of
the collection of languages $\mX$.

\begin{table}[h]
    \centering
    \caption{Dichotomy of the output density depending on the Cantor-Bendixson Rank and the different $f$-window density}
    \label{tab:main_table}
    \begin{tabular}{|p{4cm}|p{5cm}|p{5cm}|}
        \hline
        \diagbox{$f$}{Rank of $\mX$} & $\mX$ infinite Cantor-Bendixson Rank & $\mX$ finite Cantor-Bendixson Rank \\
        \hline
       \makecell{$f(k) \to \infty$ \\
       (including Banach lower \\ density, P\'olya lower \\ density)} & Exists $\mX$ where the lower $f$-window density is always zero  (Theorem \ref{thm:lowerdensity0f}, Corollary \ref{thm:lowerdensity0fBanach})   & The algorithm can always output the optimal lower $f$-window density  $1/2$ (Theorem \ref{thm:finiterank}, Corollary \ref{thm:finiterankf}) \\
        \hline
        \makecell{$f(k)$ uniformly  \\ bounded  above  \\ (include asymptotic \\ lower density)} & The algorithm can always output the optimal lower $f$-window density  $1/2$ (Corollary of the results in \cite{kleinberg2026density}) & The algorithm can always output the optimal lower $f$-window density  $1/2$ (Corollary of the results in \cite{kleinberg2026density}) \\
        \hline
    \end{tabular}
\end{table}

This sharp phase transition phenomenon above on $f$ indicates that the
asymptotic lower density (the main subject in our previous papers studying
breadth of generation \cite{kleinberg2026density,
kleinberg2025density}) is indeed fundamentally different
and more tractable to achieve than the other
densities. And as noted earlier, when $f$ is not
constant (i.e, not asymptotic lower density), we have to utilize the
topology much more carefully, in contrast to the 
asymptotic lower density, where the arguments ultimately
only require the partially ordered set of the languages
with respect to set inclusion.

Finally, the same arguments extend to the model where generators  may output up to $C \geq 1$ strings per adversarial input. In the bounded-window/asymptotic regime, the optimal benchmark becomes $C/(C+1)$; in the unbounded-window regime, the infinite-rank zero-density examples persist. We state the formal version later.

\section{Review of basic definitions for density and topology}

Before proceeding to the results, we review and formalize two sets of basic definitions: those concerning the density of integer sequences and the topological notions underlying the basic dichotomy. Although introduced in Section~\ref{sec:intro}, we present them here more formally and record some basic properties, preparing for the first main results in the next section.

\subsection{Banach density in one dimension}

In the study of integer sequences, various notions of density are utilized to capture the structural sparseness or richness of a set $A \subseteq \mathbb{N}$. The  most prominent lower densities in the literature are the lower asymptotic density and the lower Banach density. 

As in the introduction, for one-dimension, we can either consider the density of
a set of natural numbers $A$ in another  set $B \subset \mathbb{N}$
or we can simply consider the density of $A$ in the natural numbers
themselves.
From a definitionally point of view, the considerations are
essentially equivalent, by havng the indices of
the elements of $B = \{b_1, b_2, b_3, \ldots\}$ play
the role of the natural numbers, and so we will go back and forth
between these as needed.

We recall several definitions from the Introduction. In our previous work \cite{kleinberg2025density, kleinberg2026density}, we used \textit{lower asymptotic density} to measure breadth.  It measures the sparseness of the set anchored to the origin. To be more precise, for a subset $A \subset \mathbb{N}$,  the lower asymptotic density in $\mathbb{N}$ is defined as
\begin{equation}
    \underline{d}(A, \mathbb{N}) = \liminf_{n \to \infty} \frac{|A \cap [1, n]|}{n}. \label{def:asymptotic}
\end{equation}

However, because the interval is anchored at $1$, lower asymptotic density does not account for the behavior of $A$ when the measuring window begins at later points. 

\subsubsection{Lower Banach density}

The \emph{lower Banach density} captures the absolute worst-case global sparseness by allowing an interval of length $k$ to translate arbitrarily along the natural numbers. 
\begin{defn}
  The \emph{lower Banach density} of a set $A \subseteq \mathbb{N}$ is the limit of the minimum relative density of $A$ in any interval of length $n$, as $n \to \infty$, 
\[\delta_B(A, \mathbb{N}) = \liminf_{|I| \to \infty} \frac{|A \cap I|}{|I|}.\] 
Or equivalently, 
\begin{equation}
    \delta_B(A, \mathbb{N}) = \liminf_{k \to \infty} \inf_{m \geq1} \frac{|A \cap [m, m+k-1]|}{k}. \label{def:banach}
\end{equation}
\end{defn}

One can similarly define the upper Banach density $\delta^B(A, \mathbb{N})$. 
We can also write  $\underline{d}(A, \mathbb{N})$  and $\bar{d}(A, \mathbb{N})$ for the lower and upper asymptotic densities, respectively. 

In the one-dimensional case, we could assume the true language $K$ is a subset of  $\mathbb{N}$, so the set of algorithm output $A$ in $K$ always satisfy 
$$\delta_B(A, K) \leq \underline{d}(A,K ) \leq \bar{d}(A, K ) \leq \delta^B(A, K). $$

Thus, obtaining lower bounds on the lower Banach density is the most difficult, as the averaging window can slide arbitrarily; it is therefore the most stringent notion of density (also see Claim \ref{claim:banachsmall}). It is possible to have a set with a lower asymptotic density of 1 (and thus upper Banach density of 1) but a lower Banach density of 0. For example: Consider a set $A$ constructed by removing arbitrarily long intervals that appear very sparsely (e.g., removing an interval of length $k$ at position $2^k$). The sparseness ensures the asymptotic density remains 1, but the existence of arbitrarily large gaps forces the lower Banach density to 0.

\subsection{More general densities:  $f$-window density.}
One may define a hierarchy of lower densities by varying the restrictions imposed on the left endpoint of the sliding window. Between the lower asymptotic density and the lower Banach density, several intermediate notions have been introduced to capture the macroscopic tail behavior of a set \(A\); a prominent example is the \emph{lower P\'olya density}. By restricting the observation window to intervals of the form \([\theta n, n]\) for a fixed proportion \(\theta \in (0,1)\), this density is defined by
\begin{equation}
    \underline{d}_P(A, \mathbb{N})
    \;=\;
    \lim_{\theta \to 1^-}
    \liminf_{n \to \infty}
    \frac{|A \cap [\theta n, n]|}{(1-\theta)n}.
\end{equation}
These classical densities satisfy the well-known hierarchy
$    \delta_B(A, \mathbb{N}) \leq \underline{d}_P(A, \mathbb{N}) \leq  \underline{d}(A, \mathbb{N}).$

From a structural perspective, the distinctions among these notions are governed by the extent to which the viewing window may be translated relative to its length. To formalize and interpolate this relationship, we introduce a unifying framework, which we term the \emph{\(f\)-window lower density}.

\begin{defn}\label{def:fdensity}
Let $f: \mathbb{N} \to \mathbb{N} \cup \{\infty\}$ be a non-decreasing function. For any set $A \subseteq \mathbb{N}$, the lower $f$-window density of $A$ is defined as
\begin{equation}
    \delta_f(A, \mathbb{N}) = \liminf_{k \to \infty} \min_{1 \leq m \leq f(k)} \frac{|A \cap [m, m+k-1]|}{k}.
\end{equation}
\end{defn}

The function \(f(k)\) determines the maximum allowed starting position of a window of length \(k\), thereby restricting its translational freedom. Larger values of \(f(k)\) impose stricter constraints on admissible windows, making the density \(\delta_f(A)\) less forgiving and typically smaller. The following monotonicity property is immediate from the definition.

\begin{claim}[Monotonicity]\label{claim:monotone}
 Suppose $f \geq g$ and both $f, g : \mathbb{N} \to \mathbb{N} \cup \{\infty\}$ are non-decreasing functions. Then $\delta_f(A, \mathbb{N}) \leq \delta_g(A, \mathbb{N})$.
\end{claim}
Under this framework, $\delta_f(A, \mathbb{N})$ reveals a continuous spectrum of intermediate densities. Because a faster-growing $f(k)$ strictly expands the search domain for sparse intervals, $\delta_f(A, \mathbb{N})$ is monotonically decreasing with respect to the growth rate of $f$.

By choosing the growth rate of $f$, we strictly recover the three classical lower densities:

\begin{itemize}
    \item {Asymptotic Density ($f(k) = 1$):} The window is permanently pinned to $m=1$, granting it zero translational freedom. 
    
    \item {Banach Density ($f(k) = \infty$):} The window possesses infinite translational freedom, allowing it to seek out the globally sparsest interval of length $k$. 
    
    \item {P\'olya Density ($f(k) = ck$):} The standard P\'olya interval $[\theta n, n]$ has length $k = (1-\theta)n$ and starting position $m = \theta n$. Substituting $c = \frac{\theta}{1-\theta}$ forces $m = ck$, granting the window \textit{linear} translational freedom. 

    \begin{lem}
Let $\delta_f(A, \mathbb{N})$ denote the $f$-window lower density of a set $A \subseteq \mathbb{N}$. For linear growth $f(k) = ck$ with $c > 0$, we have
\begin{equation}
    \lim_{c \to \infty} \delta_{f(k) = ck}(A, \mathbb{N}) = \underline{d}_P(A, \mathbb{N}).
\end{equation}
\end{lem}
The proof of this lemma is in the Appendix. 
\end{itemize}

Since $f(k) = \infty$ is the ``largest" possible function, we have the following:
\begin{claim}\label{claim:banachsmall}
   The lower Banach density is the smallest among all lower $f$-window densities. 
\end{claim}

\subsection{Underlying topology on language collections}\label{sec:topo}

This paper shows that, for stringent notions of lower density such as the lower Banach density and the generalized densities introduced later, an analysis of the topological structure of $(\mX,\mT)$, in particular its Cantor--Bendixson rank, is essential. In this regime, the breadth of language generation is tightly coupled to topological complexity, in contrast to asymptotic lower density, which depends primarily on the inclusion poset and is largely insensitive to finer topological structure.

We will use the topological space we defined in 
\cite{kleinberg2025density,kleinberg2026density}, and we will
see that it is exactly what is needed to express the dichotomy
in achievable lower Banach density.
We will describe it here in self-contained form for the sake of completeness.

The definition of the topology is motivated by a natural 
type of combinatorial limit point in the collection of languages $\mX$,
and so we define this first.

\begin{dfn}[infinite perfect tower]\cite{kleinberg2025density}\label{def:perfecttower-intro}
We say that a sequence of languages 
$\genlang_1, \genlang_2, \genlang_3, \ldots $, each from $\coll$,
forms an {\em infinite perfect tower} with respect to a language
$\termlang \in \coll$ if 
\begin{itemize}
\item[(i)] $\genlang_j \subsetneq \termlang$ for all $j \geq 1$.
\item[(ii)] $\genlang_j$ fixes at least one string of $\termlang$, 
for all $j \geq 1$.
\item[(iii)] Every string of $\termlang$ is fixed by some 
language $\genlang_j$ in the sequence.
\end{itemize}
\end{dfn}
We call $\termlang$ the {\it terminal language} of this infinite perfect tower $(\genlang_1, \genlang_2, \dots)$. 
Intuitively, as we iterate through the languages 
$\genlang_1, \genlang_2, \genlang_3, \ldots $, every
string of $\termlang$ eventually appears and after some point stays forever;
and the sequence is non-redundant in that each language $\genlang_j$
represents the moment of permanent appearance for some string of $\termlang$.

We may assume without loss of generality that every language in the countable collection $\coll$ is a subset of $\mathbb{N}$, since any countable discrete set of strings admits an effective enumeration by natural numbers. Let $\mX$ denote this collection of languages. We define a topology $\mathcal{T}$ on $\mX$ by specifying a basis as follows. For each language $\lan \in \mX$ and each finite set $F \subset \mathbb{N}$, define
\[
U_{\lan,F} \;=\; \{\, \lan' \in \mX \mid F \subseteq \lan' \subseteq \lan \,\}.
\]
The family of all such sets $U_{\lan,F}$, ranging over $\lan \in \mX$ and finite $F \subset \mathbb{N}$, forms a basis and hence induces a topology $\mathcal{T}$ on $\mX$.

The resulting space $(\mX,\mathcal{T})$ is Hausdorff and first-countable. Its central role for us stems from the fact that its topological notion of accumulation precisely captures the combinatorial structure of infinite perfect towers.

\begin{claim}\label{claim:iptequiv_rephrased} \cite{kleinberg2025density}
A language $\lan' \in \mX$ is a limit point of $(\mX,\mathcal{T})$ if and only if $\lan'$ arises as the terminal language of an infinite perfect tower.
\end{claim}

To study the structure of the topological space $(\mX,\mathcal{T})$, we use the \emph{Cantor--Bendixson rank} to be defined below. The \emph{derived set} $d(\mX)$ is defined as the collection of all limit points of $\mX$. Using this operator, one defines the \emph{Cantor--Bendixson hierarchy} (see, e.g., \cite{settheory}) by transfinite recursion. Set $\mX^{(0)}=\mX$, let $\mX^{(1)}=d(\mX)$, and for each integer $i\geq 1$ define
\[
\mX^{(i+1)} \;=\; d\bigl(\mX^{(i)}\bigr),
\]
the set of limit points of $\mX^{(i)}$ relative to $\mX$.

\begin{defn}\label{def:CBrank}
    The sequence $\mX^{(i)}$ is monotone decreasing. The least ordinal $\alpha$ such that $\mX^{(\alpha)}=\emptyset$ is called the \emph{Cantor--Bendixson rank} of $(\mX,\mathcal{T})$. If no such $\alpha$ exists, its Cantor-Bendixson rank is $\infty$. 
\end{defn}

In \cite{kleinberg2025density}, we used the Cantor–Bendixson rank to show that when the rank is small, the algorithm can always generate outputs in the true language with large asymptotic lower density. However, it turns out that for the study of asymptotic lower density alone, the Cantor–Bendixson rank, and more generally the topology on $(\mX, \mathcal{T})$, is not required (see \cite{kleinberg2026density}). Instead, the only structure needed is the poset structure on the set of languages, where the ordering is given by set inclusion.

\begin{remark}
The topology is chosen carefully to meet the requirements of our analysis. In \cite{kleinberg2026density}, we considered an  alternative topology, in which a basis of open sets is given by
\[
U_{\lan,F} \;=\; \{\, \lan' \in \mX \mid F \subseteq \lan' \},
\]
without imposing the condition that \(\lan' \subset \lan\). In the present work, however, we require the topology on \((\mX,\mathcal{T})\) to satisfy standard separation properties, in particular the Hausdorff condition. This necessitates the additional restriction \(\lan' \subset \lan\) in the definition of basic open sets. Absent this restriction, every language \(\lan'\) with \(\lan' \subset \lan\) would belong to the closure of \(\lan\), a degeneracy that is incompatible with our intended analysis.
\end{remark}

\section{Density results in one dimension}

In this section, we cover the results in one dimension, i.e., the embedding space has dimension one.  First, in Section \ref{subsec:inf-cb-rank}, we 
give the construction of a language collection
with infinite Cantor-Bendixson rank in which 
no algorithm can achieve a positive lower Banach 
density for its generated set.
The construction will necessarily use a sequence 
of infinite perfect towers, and
the requirement that the algorithm 
must achieve generation in the limit,
to iteratively build long 
intervals that the algorithm must completely miss.

Next, in Section \ref{subsec:finite-cb-rank}, 
we give the algorithm that achieves
positive lower Banach density when
the language collection has
finite Cantor-Bendixson rank.

A natural baseline for this problem is the index-based algorithm $\Acc$ we introduced in \cite{kleinberg2025density}, which is guaranteed to be accurate infinitely often. More explicitly, 
at every time $t$,
$\Acc$ maintains a current guess at an index $i_t$
in $\mX$ for the true language,
and in infinitely many time steps this
index is {\em accurate} in the sense that $L_{i_t} = K$.
This was the starting point for the
positive results on lower asymptotic density in
\cite{kleinberg2025density,kleinberg2026density}
as well.
But it is difficult to use this algorithm
directly to get a lower bound on the lower Banach density, since the key benefit of
an algorithm that is accurate infinitely often
is that it can ``pull back'' in a way that
ensures dense coverage of the early elements
in the language.
This works well for the asymptotic lower density,
where early elements are prioritized by definition.
But how does it help when we need to achieve
positive lower Banach density, requiring 
that we have high coverage in {\em all} long intervals?

In some sense, it can't help directly on its own,
as know from the construction in
Section \ref{subsec:inf-cb-rank} of a set
with infinite Cantor-Bendixson rank where
positive lower Banach density in unachievable.
So instead, we must use the finite Cantor-Bendixson
rank in an essential way, and we do this 
by extracting a hierarchical structure on the
languages where the finite rank 
helps us avoid the kind of ``infinite regress''
that leads to zero density.

After defining the algorithm and proving that
it achieves generation in the limit in 
Sections \ref{subsubsec:define-alg} and
\ref{subsubsec:validity}, we turn to the
analysis of the lower Banach density in 
Section \ref{subsubsec:alg-properties}.
After a rather involved setup, the analysis can be abstracted as an interaction between the adversary and the algorithm. The adversary implicitly attempts to enumerate most elements of a long interval, while the algorithm seeks to ensure that every interval is populated sufficiently densely.
It turns out to be useful to abstract this 
interaction in the form of a combinatorial 
game played between two players
(representing the adversary and algorithm),
and we define this game in 
Section \ref{subsubsec:comb-game} in a way
that is independent of the rest of the
language-generation setting.

\subsection{Infinite rank: constructions with zero lower Banach density }
\label{subsec:inf-cb-rank}

We recall our result from \cite{kleinberg2026density}, which shows that for asymptotic lower density, an algorithm can always generate outputs in the true language with the optimal density \(1/2\). In this section, we demonstrate a sharp contrast for the lower Banach density: there exists a countable collection of languages \(\mX\) for which no algorithm can achieve a positive lower Banach density. Equivalently, for any algorithm, there are arbitrarily large intervals of the true language \(K\) that the algorithm intersects only sparsely.

The lower bound is diagonal in spirit, but the Banach-density objective
imposes a stronger requirement. It
is not enough to force the generator to miss infinitely many individual
strings. To force lower Banach
density zero, the adversary must create arbitrarily long intervals
containing essentially no algorithmic outputs at all, while still ensuring that the target remains
a member of a fixed countable family. A fully adaptive decision tree would naturally create uncountably many
possible limit languages; the construction below instead encodes the
required long gaps inside a fixed countable infinite-rank family.

\begin{thm}\label{thm:lowerdensity0f}
Let $f: \mathbb{N} \to \mathbb{N} \cup \{\infty\}$ be a non-decreasing function satisfies that $\limsup_{n \to \infty} f(n) = \infty$. Then
    there exists a countable collection of languages $\mathcal{C}$ such that any algorithm that generates in the limit will only have $f$-window  lower density $0$. 
\end{thm}

\begin{cor}\label{thm:lowerdensity0fBanach}
There exists a countable collection of languages $\mathcal{C}$ such that any algorithm that generates in the limit only has lower Banach density $0$.
\end{cor}

We say that an algorithm has \emph{speed} $C$ if, at each time $t$, whenever the adversary reveals one string, the algorithm may output up to $C$ strings. 
A similar construction yields the following result, showing that increased speed does not provide any advantage.

\begin{cor}\label{cor:lowerdensity0f}
    There exists a countable collection of languages $\mathcal{C}$ such that any algorithm that generates in the limit with speed $C$ will only have Banach   lower density $0$. 
\end{cor}

Our construction relies on the following property of a function $f$. For clarity, we state this property explicitly before presenting the construction.

Let $f: \mathbb{N} \to \mathbb{N} \cup \{\infty\}$ be an increasing function. We say $f$ is {\it nice} if the following condition holds. For any $T > 0$, there exists an $m \geq 2$  and sequences $(k_i)_{i=1}^m  \in \mathbb{N}^m$ and $(a_i)_{i=1}^m \in \mathbb{N}^m$ satisfying:
\begin{enumerate}
    \item $k_1 \leq k_2 \leq \dots \leq k_m$
    \item $T < a_1 < a_2 < \dots < a_m$
    \item $a_{i+1} \geq a_i + k_i + 1$
    \item $a_i \leq f(k_i)$
\end{enumerate}

The following claim is a simple fact. 
\begin{claim}\label{claim:nicef}
An increasing function $f: \mathbb{N} \to \mathbb{N}$ is nice if and only if it is not universally bounded above.
\end{claim}
The proof of the Claim is in Appendix.

\begin{proof}[Proof of Theorem \ref{thm:lowerdensity0f}]
By Claim \ref{claim:nicef}, $f(x)$ is nice. 

We first define a sequence $N_1, N_2, \dots$ with $N_i \to \infty$, to be specified recursively below.

Let $N_1 = 1$. By Claim~\ref{claim:nicef}, there exist integers $m_1$ and sequences $\{k_i^{(1)}\}_{i=1}^{m_1}$ and $\{a_i^{(1)}\}_{i=1}^{m_1}$ satisfying the required conditions, since $f$ is nice.
Define the intervals
\[
[a_1^{(1)}, a_1^{(1)} + k_1^{(1)}],\; [a_2^{(1)}, a_2^{(1)} + k_2^{(1)}],\; \dots,\; [a_{m_1}^{(1)}, a_{m_1}^{(1)} + k_{m_1}^{(1)}].
\]
These intervals are pairwise disjoint by item~3. Let $\mathcal{I}_1$ denote this collection of intervals.

Let $N_2$ be sufficiently large such that it exceeds the largest endpoint of the intervals in $\mathcal{I}_1$, and moreover satisfies $f^{-1}(N_2 - 1) > k_{m_1}^{(1)}$. 
Applying Claim~\ref{claim:nicef} again, there exist integers $m_2$ and sequences $\{k_i^{(2)}\}_{i=1}^{m_2}$ and $\{a_i^{(2)}\}_{i=1}^{m_2}$ satisfying the required conditions, since $f$ is nice.
Define the collection $\mathcal{I}_2$ to consist of the intervals
\[
[a_1^{(2)}, a_1^{(2)} + k_1^{(2)}],\; [a_2^{(2)}, a_2^{(2)} + k_2^{(2)}],\; \dots,\; [a_{m_2}^{(2)}, a_{m_2}^{(2)} + k_{m_2}^{(2)}].
\]
These intervals are pairwise disjoint by item~3.
Furthermore, we have $k_1^{(2)} > k_{m_1}^{(1)}$. Indeed, since
$f(k_1^{(2)}) \geq a_1^{(2)} \geq N_2 > f\big(k_{m_1}^{(1)}\big),$
and $f$ is non-decreasing, the claim follows.

Repeating this process, we construct an increasing sequence $N_1 < N_2 < \cdots$ and a collection of pairwise disjoint intervals $\mathcal{I}_1 \cup \mathcal{I}_2 \cup \cdots$ such that, for each $s \ge 1$:

\begin{enumerate}
    \item $\mathcal{I}_s$ consists of $m_s$ disjoint intervals
    \[
    [a_1^{(s)}, a_1^{(s)} + k_1^{(s)}],\; [a_2^{(s)}, a_2^{(s)} + k_2^{(s)}],\; \dots,\; [a_{m_s}^{(s)}, a_{m_s}^{(s)} + k_{m_s}^{(s)}];
    \]
    \item the interval lengths satisfy
$    k_1^{(s)} \le k_2^{(s)} \le \cdots \le k_{m_s}^{(s)};$  
    \item the endpoints satisfy   
$    a_i^{(s)} \le f\big(k_i^{(s)}\big) < N_{s+1} \le a_1^{(s+1)} \quad \text{for all } i \in [m_s];$
    
    \item every interval in $\mathcal{I}_s$ is strictly shorter than every interval in $\mathcal{I}_{s+1}$.
\end{enumerate}

Let $P$ denote the set of integers that appear as left endpoints of intervals in $\mathcal{I}_1 \cup \mathcal{I}_2 \cup \cdots$, ordered in increasing order. For each nonnegative integer $k$, let $P[k]$ denote the set consisting of the $k$ smallest elements of $P$, with the convention that $P[0] = \emptyset$.

We now define a countable collection of languages as follows:

For each tuple $(\ell, a, d, (h_1, h_2, \dots, h_d))$, where $\ell, a, d$ range over all nonnegative integers and $(h_1, \dots, h_d)$ is any non-increasing $d$-tuple of nonnegative integers (i.e., $h_1 \ge h_2 \ge \cdots \ge h_d$), we define
\begin{align*}L_{l, a, d, (h_1, h_2, \dots, h_d)} & = [l]  \cup (P   + [0,a])  \\
   &  \cup \left(P[h_1] + [1,a+1]\right) \cup (P[h_2] + [1,a+2]) \cup \dots \cup (P[h_d] + [0,a+d]).\end{align*} 
Here given sets $P$ and $Y$, we define their sumset by
\[
P + Y \;=\; \{\, s + y \mid s \in P,\; y \in Y \}.
\]
In particular, $\emptyset + Y = \emptyset$. We also include $\mathbb{N}$ as a language in the collection.

The key property we will repeatedly use is the structure of the Cantor--Bendixson derivatives of this countable family of languages:
\begin{enumerate}
\item If $a = 0$, then this language is an isolated point. 
\item If $a \ge 1$, then $L_{\ell, a, d, (h_1, h_2, \dots, h_d)}$ is a limit point of the sequence of languages $L_{\ell, a-1, d+1, (k, h_1, \dots, h_d)}$ as $k \to \infty$.

\item For any sequence of languages whose $\ell$-indices tend to infinity, the sequence has $\mathbb{N}$ as a limit point.
\end{enumerate}
Consequently, this collection of languages has infinite Cantor-Bendixson rank, with the ranking informally governed by the parameter $a$. Indeed, after removing the languages with $a=0$, those with $a=1$ become isolated points, and the process continues similarly for larger values of $a$.

We define a \emph{mega step} as follows. The adversary aims to generate $k_i^{(s)}$ consecutive integers in $\mathbb{N}$ that have not been occupied by the algorithm. Concretely, the goal is for the adversary to occupy all $k_i^{(s)}$ elements in an interval $[a_i^{(s)}, a_i^{(s)} + k_i^{(s)}] \in \mathcal{I}_s$, for some $i$ and $s$.

A mega step consists of multiple rounds across different levels, defined below. Roughly speaking, the key elements in level $i$ are the elements in $P + \{i-1\}$ (again this denotes sum set).

At any time $t$, we say that an element $a_i^{(s)} + q \in \mathcal{I}_s$ is \emph{edgy} at time $t$ if none of the elements
\[
a_i^{(s)} + q + 1, \;\dots,\; a_i^{(s)} + k_i^{(s)}
\]
have been used by either the adversary or the algorithm up to time $t$.

In an ideal scenario, suppose there exist time stamps
$t_1 < t_2 < \cdots < t_{k_i^{(s)}}$ such that, at each time $t_j$,
the adversary is able to input $a_i^{(s)} + j - 1$, and this element
$a_i^{(s)} + j - 1$ is \emph{edgy} at time $t_j$.
Then, by time $t_{k_i^{(s)}}$, all elements in the interval
$\mathcal{I}_s$ will be occupied exclusively by the adversary.
However, ensuring that every element in a given interval is edgy
is highly nontrivial. In fact, for any fixed interval, we cannot
guarantee this property, since the behavior of the algorithm is
beyond our control. The key idea is instead to exploit the
topological structure of the collection of languages, together with
the requirement that the algorithm must generate in the limit, in
order to guarantee the existence of such a favorable interval.

We define it recursively. 
Suppose we are about to start a new Mega Round, and all the previously used strings (either by the adversary or by the algorithm) is at most $l'$. From now on, if the adversary pretends that the true language is of the form $L_{l, a, d, (h_1, h_2, \dots, h_d)}$ for some $l \geq l'$, then any of these languages will still be consistent with the adversary input so far. This technique is crucial because, across successive mega-steps, we must gradually increase the interval length $k_i^{(s)}$. Consequently, the adversary is forced infinitely often to pretend that the true language has a sparser  density, all while remaining consistent with previous adversarial enumerations. The parameter $l$ facilitates this trick by ensuring that an even sparser language can still maintain  consistent.

\

\noindent{Round level 1}: 
The adversary first identify the smallest $N_i$ that is larger than $l$. Set $l = N_i$.

Initially, the adversary pretends that the true language is $L_{l, 0, 0, \emptyset}$. Because the algorithm generates the language in the limit, after a finite amount of time, it will exclusively output elements from $L_{l, 0, 0, \emptyset}$. Consequently, from that point onward, all elements generated by either the adversary or the algorithm will belong to $P$.

Suppose when this happens, all the elements used so far are smaller than $N_s$ for some $s$. 
Recall there are $m_s$ intervals in $\mathcal{I}_s$: $[a_1^{(s)}, a_1^{(s)} + k_1^{(s)}], [a_2^{(s)}, a_2^{(s)}+k_2^{(s)}], \dots, [a_{m_2}^{(s)}, a_{m_2}^{(s)} + k_{m_2}^{(s)}]$, and \[a_1^{(s)} \geq N_s.\] 
Recall $a_1^{(s)}, a_2^{(s)}, \dots, a_{m_s}^{(s)} \in P$ and have not been used. 
As the adversary and the algorithm reveals a string in turn with the same speed, at least half of the elements $a_1^{(s)}, a_2^{(s)}, \dots, a_{m_s}^{(s)} \in P$ can be generated by the adversary first before the algorithm touches them in the next $m_s$ time stamps. Call the end of Round level 1.

Let $S_{t_1}$ be the set of these elements above which are occupied by the adversary. So \[|S_{t_1}| \geq m_s/2 \geq 1.\] 

In addition, we have the property that for each $a_i^{(s)} \in S_{t_1}$, none of the next $k_i^{(s)}$ elements were used by the adversary or the algorithm up to now, because the algorithm can only generate subsets of  $L_{l, 0, 0, \emptyset}$. Notice that there is no conflict with $a_{i+1}^{(s)}$ as we have assumed $a_i^{(s)}+1 + k_i^{(s)} < a_{i+1}^{(s)}$. In other words, at the end of Round level 1, all the elements in $S_{t_1}$ are edgy. 

\

\noindent{Round level 2}: 

The adversary pretends that the true language is $L_{l, 0, 1, (t_1')}$ such that $P[t_1'] + \{1\}$ is larger than any element the adversary and algorithm have used so far. 
To do this, the adversary first outputs integer $x + 1$ for the smallest $x$ in $S_{t_1}$, and greedily pick the next smallest available $a_i^{(s)} +1$ that is available and edgy with  $a_i^{(s)} \in S_{t_1}$, i.e., that $a_i^{(s)}+2, \dots, a_i^{(s)}+k_i^{(s)}$ have not been used by the algorithm at the moment. The adversary  occupies this $a_i^{(s)}+1$.
This round ends when all the elements $x+1$ are used by either the adversary or the algorithm for all $x \in S_{t_1}$. 

If in this round, the algorithm never generates a string in $P + [2, \infty]$, then by the same argument as before, at least half of the integers $x$ in $S_{t_1}$ will satisfy that both $x, x+1$ are occupied by the adversary, and $x+1$ is edgy at the end of Round level 2. 
Let this set of $x \in S_{t_1}$ be $S_{t_1'}$. Clearly 
$|S_{t_1'}| \geq |S_{t_1}|/2.$
In fact, we could improve this bound to be 
\[ |S_{t_1'}| \geq 1 + (|S_{t_1}|-1)/2.\]
The reason is that the adversary will first output an element of the form $P + \{1\}$, which contributes to the $+1$ term, and then if the algorithm never outputs in $P + [2, \infty]$, we get that the adversary can occupy an additional $(|S_{t_1}|-1)/2$ strings that are desirable. 
Although this new bound may appear to be a minor improvement, it is highly significant because it guarantees that $|S_{t_1'}| \geq 1$. Consequently, we avoid the issue of an exponentially decaying set size, provided we only require the set to be non-empty.

If in this round, the algorithm  at some point generates an element in $P + [2, \infty]$, we call this situation {\it bad at level 1}. 

If the bad situation occurs, the adversary first acts as though the true language were
$L_{l,0,1,(t_2)},$
where $t_2$ is the smallest value for which this language is consistent. The adversary continues this strategy until the algorithm generates only a subset of
$L_{l,0,1,(t_2)}$
(at round level~1). This ensures that the set $P$ contains many edgy elements that are all larger than some sufficiently large threshold $N_i$, chosen so that $N_i$ exceeds all elements used previously by the end of round level~1. As in the previous argument, we denote by $S_{t_2}$ the set of edgy elements output by the adversary.

Similar to the previous stage, the adversary then repeats the argument at round level~2. It does so by greedily generating available elements of the form $x+1$, where $x \in S_{t_2}$, under the additional requirement that $x+1$ itself be edgy. In this way, the adversary pretends that the true language is
$L_{l,0,1,(t_2')},
$ for some sufficiently large $t_2'$.
Again, round level 2 is good if in this round the algorithm never generates an element in $P+ [2, \infty]$. If this round is good, then again at least half of the integers $x$ in $S_{t_2}$ will satisfy that both $x, x+1$ are occupied by the adversary, and $x+1$ is edgy. 
Let this set of $x$ be $S_{t_2'}$. Clearly by the same reasoning as before
\[ |S_{t_2'}| \geq 1+ (|S_{t_2}|-1)/2 \geq 1.\]

Round level 2 is bad if in this round the algorithm ever generates an element in $P+ [2, \infty]$.

If round level~2 is bad again, we repeat the process: we return to round level~1 by pretending that the true language is
$L_{l,0,1,(t_i)},$
and then move  to round level~2 by pretending that the true language is
$L_{l,0,1,(t_i')},$
where the parameters are chosen to form a sufficiently well-separated sequence
$t_1 < t_1' < t_2 < t_2' < \cdots.$

We now show that the bad situation at level~2 can occur only finitely many times. This is a crucial step in making the rest of the proof work.

Suppose, for contradiction, that the process continues to alternate between Round level~1 and Round level~2 indefinitely, and that the situation remains ``bad'' throughout. We will show that this contradicts the assumption that the algorithm converges in the limit.

To show this, we need to find a valid enumeration plan of {\it all the} strings in $L_{l, 1, 0, \emptyset}$ by the adversary. This is achieved by the adversary gradually pretending the true languages is $L_{l, 0, 0,\emptyset}$, then $L_{l,0,1,(t_1')}, L_{l, 0, 1, (t_2)}, L_{l, 0, 1, (t_2')}, \dots$ as above. Note that, between two rounds, some earlier strings may have been skipped. Therefore, the adversary must ensure that these strings are eventually enumerated. One way to do this is to reserve stages of the form $2^{2^n}$ and, at each such stage, go back and enumerate the first missing string in the language that the adversary is simulating. Since this refill process is sufficiently sparse, all of our bounds on $|S_{t_i}|$ and $|S_{t_i'}|$ continue to hold, up to a negligible error term that vanishes in the limit, i.e., 
\[ |S_{t_2'}| \geq 1+ (|S_{t_2}|-1 - o(|S_{t_2}|))/2 \geq 1.\]

We now use the property that the infinite sequence $L_{l,0,1,(t_1')}, L_{l, 0, 1, (t_2)}, L_{l, 0, 1, (t_2')}, \dots$ is a perfect tower with limit being $L_{l,1,0,\emptyset}$. What we argued above shows that the adversary indeed has an enumeration strategy for $L_{l,1,0,\emptyset}$. Therefore, since the algorithm needs to generate in the limit, it means that after some finite time, the adversary should never output any string in $P + [2, \infty]$, this contradicts with our assumption that the algorithm generates a string in $P + [2, \infty]$ infinitely often (which are the definition of the bad situation).

Therefore there must be some point where it is good at level 2. It means that there is some $N_s$ for some $s$, with the $m_s$ intervals in $\mathcal{I}_s$: $[a_1^{(s)}, a_1^{(s)} + k_1^{(s)}], [a_2^{(s)}, a_2^{(s)}+k_2^{(s)}], \dots, [a_{m_2}^{(s)}, a_{m_2}^{(s)} + k_{m_2}^{(s)}]$, and \[a_1^{(s)} \geq N_s.\] 
Recall $a_1^{(s)}, a_2^{(s)}, \dots, a_{m_s}^{(s)} \in P$ and have not been used. 
For at least one element $x$ in $P \cap \mathcal{I}_s$, 
the adversary occupied both $x, x+1$, and $x+1$ is edgy at the end of Round 2. 
Call this set of $x$ to be $S_{T_2}$.
When this happens, we say it is {\it good at level 2}.

We will now recursively apply the previous argument to show that the adversary is able to occupy a whole interval. 
 When it is good at Level 2, it means $x+1$ is edgy at the beginning of Level 3. The adversary then proceeds with Round level 3.

\

\noindent{Round level 3}

The adversary pretends that the true language is $L_{l, 1, 1, (\star)}$ for some sufficiently large $\star$ so that it is consistent.  

If in this round, the algorithm never outputs a string in $P + [3, \infty]$, then at least half of the integers $x$ in $S_{T_2}$ will satisfy that both $x,x+1,  x+2$ are occupied by the adversary, $x+2$ is edgy at the end of this Round.  
Let this set of $x$ be $S_{t_1^2}$. Clearly 
\[ |S_{t_1^2}| \geq 1.\]

If in this round, the algorithm ever produces a string in $P + [3, \infty]$, we again say it is {\it bad at Level 3}. Once it happens, similar to before, the algorithm will repeat alternating between Level 1 and Level 2 until at some point it is good at level 2 (we will specify it  soon). Once it is good at level 3, we will be able to three consecutive strings occupied by the adversary and satisfying the edgy condition. 

If it is never good at level 3 whenever it is good at level 2, we show it is impossible by a similar argument as before. 

We again prove by contradiction. The adversary first go through repeatedly round level 1 and round level 2, until it is good at level 2, and we have shown it will happen. In this part, the adversary is always pretending the true language is $L_{l,0,1, (\star)}$ with $\star$ grows. Once it becomes good at level 2, the algorithm starts to pretend the true language is $L_{l,1,1, (\star')}$ for some sufficiently large $\star'$ depending on the strings used so far. And then it goes back to alternate between level 1 and level 2 again, then go to level 3, and repeat. 
A crucial observation is that all these languages that the adversary pretends to be a subsets of $L_{l, 2, 0, \emptyset}$, and these pretending languages form an infinite perfect tower with limit being $L_{l, 2, 0, \emptyset}$ since those $\star, \star'$s grow to  $\infty$. By filling in the missing string infinitely often but arbitrarily slowly, the adversary has established a plan  to eventually enumerate all possible strings in $L_{l, 2, 0, \emptyset}$. If we are never good at level 3, since the algorithm will generate in the limit, so after some time it will never generate a string of the form $P + [3, \infty]$. This contradicts with the assumption that we are never good at Level 3. 

Note that in this part of the proof it is crucial that the limit language is 
$L_{l,2,0,\emptyset}$, since we aim to forbid elements of the form 
$P + [3,\infty)$. To achieve this, it is essential that the approximating 
sequence of languages does not have increasing values of the subscript index $l$. 
Indeed, if $l$ were allowed to increase along the sequence, the limit would 
collapse to $\mathbb{N}$, but we would need the limit to be  $L_{*,2,0,\emptyset}$.
This is the key point in the construction where the Cantor--Bendixson rank of 
the example must necessarily be large.

Therefore we have shown that eventually we should be good at all levels 1,2,3. 

By an easy extension of our argument, eventually we should be good at all levels 1,2,\dots, $k_i^{(s)}$ for some $i$. More precisely, there exists an $s$ such that, at the end of Round $1$, we have $\mathcal{I}_s$: $[a_1^{(s)}, a_1^{(s)} + k_1^{(s)}], [a_2^{(s)}, a_2^{(s)}+k_2^{(s)}], \dots, [a_{m_s}^{(s)}, a_{m_s}^{(s)} + k_{m_s}^{(s)}]$, and \[a_1^{(s)} \geq N_s,\] such that 
  at least half of the elements $a_1^{(s)}, a_2^{(s)}, \dots, a_{m_s}^{(s)} \in P$ can be generated by the adversary first before the algorithm touches them in the next $m_s$ time stamps, and they are edgy. By a slight abuse of notation call this set $S_1$. In Round level 2, again at least half of these $x$ are such that $x, x+1$ are both occupied by the adversary and $x+1$ is edgy at the end of Round level 2. Call this set $S_2$; Eventually, there is some $1 \leq i \leq k_{m_s}^{(s)}$ such that in Round Level $k_i^{(s)}$, there is a set $S_{k_i^{(s)}}$  with 
 \[ |S_{k_i^{(s)}}| \geq 1 + (|S_{k_i^{(s)}-1}|-1)/2 \geq 1\]
 for some $1 \leq i \leq k_{m_s}^{(s)}$ such that $a_i^{(s)} \in S_{k_i^{(s)}}$, and all the elements $a_i^{(s)}, a_i^{(s)}+1, \dots, a_i^{(s)} + k_i^{(s)}$ are occupied by the adversary. 
This is the end of a Mega Step. 

\

To conclude the proof, we will repeatedly perform Mega Steps.

First we run a mega step with parameter $l_1 = 0$. The previous argument guarantees that there is an interval in $\mathcal{I}_{s_1}$ for some $s_1$ that is fully occupied by the adversary.

We then run the second mega stop with parameter  and $l = l_2$ to be sufficiently large such that it is larger than any string that have been used by the adversary or the algorithm and is larger than $N_{s_1}$. Then the previous argument guarantees that there is an interval  in $\mathcal{I}_{s_2}$ for some $s_2 > s_1$ that is fully occupied by the adversary.

Repeat this process. Eventually for each $i$, we choose $l_i$  sufficiently large, and we are guaranteed that there is an interval  in $\mathcal{I}_{s_i}$ for some $s_i > s_{i-1}$ that is fully occupied by the adversary.

We can always continue this process, because as $l_i \to \infty$, all the languages that the adversary pretends to be along the way form an infinite perfect tower whose limit is $\mathbb{N}$. Therefore regardless of the strategy of the algorithm, as long as it generates in the limit, the adversary is able to create an infinite sequence of intervals, each in different $\mathcal{I}_{s_i}$, such that the whole interval is occupied by the adversary. 

Let the lengths of these intervals play the role of $k$, and their left endpoints the role of $m$, in  Definition~\ref{def:fdensity} for $\delta_f(A)$. This shows that the $f$-window density of the algorithm output is forced to be $0$.
\end{proof}

\subsection{Finite Rank-- always optimal lower Banach density}
\label{subsec:finite-cb-rank}

The construction forcing density 0 (Theorem~\ref{thm:lowerdensity0f} and Corollary~\ref{thm:lowerdensity0fBanach}) has infinite Cantor--Bendixson rank. We
now show that finite rank rules out this topology-driven zero-density
phenomenon: if the language class has finite Cantor--Bendixson rank,
then the optimal \(1/2\) lower Banach-density guarantee can always be
recovered.

This theorem should be viewed as the Banach-density analogue of the
universal lower-asymptotic-density theorem, but with a crucial
difference. Lower asymptotic density gives the same optimal guarantee
for every countable class. Lower Banach density asks for uniform control
over all long intervals, and in this setting finite Cantor--Bendixson
rank becomes the structural hypothesis that makes the guarantee possible.

We do not claim that every infinite-rank class forces zero density. The
point is that zero density cannot be forced at finite rank, and that
infinite rank is rich enough to support genuine zero-density examples.

\begin{thm}\label{thm:finiterank}
    Suppose the countable collection of languages $\mX$ has finite Cantor-Bendixson rank. There is an algorithm which generates in the limit and whose lower Banach density in the true language is always at least $1/2$. 
\end{thm}

By the monotonicity Claim \ref{claim:monotone} of the lower $f$-window density, since lower Banach density is the smallest, we automatically have the following. 
\begin{cor}\label{thm:finiterankf}
 Suppose a countable collection of languages $\mX$ has finite Cantor-Bendixson rank.  
   There is an algorithm which generates in the limit whose lower $f$-window density in the true language is at least $1/2$, for any $f: \mathbb{N} \to \mathbb{N}^+$ that is not uniformly bounded above.
\end{cor}

Since the Cantor-Bendixson is a topological invariant which is preserved under reordering of elements (as the topology is defined regardless of the ordering of the elements). It turns out that under any reordering of the elements, there is always an algorithm that generates in the limit and achieves the optimal lower Banach density, which is $1/2$. This fact is not obvious from the algorithm description, and in particular, after reordering of the elements, the intervals in the language under the original ordering will be scattered under the new ordering. This is a bit surprising given that usually lower density is very sensitive to reordering of the underlying strings: for the same subset $A \subset \mathbb{N}$, reordering elements in $\mathbb{N}$ could make the Banach lower density (or even asymptotic lower density) of $A$ go from $0$ to non-zero and vice versa. 

\begin{cor}
    Let $\rho: \mathbb{N} \to \mathbb{N}$ be a bijection. Suppose the countable collection of languages $\mX$ has finite Cantor-Bendixson rank. Then there is an algorithm which generates in the limit whose lower Banach density in the true language is at least $1/2$ under the reordering of strings.  
\end{cor}

\begin{remark}
    Given that the lower Banach density could achieve zero (Theorem \ref{thm:lowerdensity0f}), it might somewhat sound surprising that one could actually prove that the optimal $1/2$ density could be achieved in this case. Even though the set of languages $\mX$ is countable, it could very well be that the completion of the space $\mX$ under the topology defined earlier is uncountable. For example, let the countable collection $L$ be some finite subsets of $2\mathbb{N}$ plus all the odd numbers. Then the completion of $\mX$ would be all subsets of $2\mathbb{N}$ plus all the odd numbers, which is clearly uncountable. In this example, $\mX$ has Cantor-Bendixson rank 1, and the completion $\bar{\mX}$ has Cantor-Bendixson rank infinity.  One might suspect that a counter-example to show that the Banach lower density could be tiny again, by taking advantage of the infinity Cantor-Bendixson rank of $\bar{\mX}$ (such as the previous construction) or the fact that $\bar{\mX}$ has uncountable cardinality (such as diagonalization method). However it turns out it is not the case. The reason is that even though $\bar\mX$ is uncountable, but those completed limit points become either  subsets (and thus absorbed) by the countably many sets, or is not a subset of any of the original language and thus could be safely ignored by the algorithm by proper design. 
\end{remark}

\subsubsection{Review of the algorithm that is accurate infinitely often}
\label{subsubsec:inf-acc}

We will start by reviewing our algorithm in \cite{kleinberg2025density, kleinberg2026density} which generates in the limit and at the same time, guaranteeing that the guess  is the correct language infinitely often. We say an algorithm is {\it index-based} if at every time $t$, it guesses the index of a language and generate a string from that language. We say  the algorithm is {\it accurate at time $t$} if the index guessed is that of the true language at  time $t$. 

\begin{thm}\label{thm:acc}\cite{kleinberg2025density}
There is an algorithm that can guarantee 
index-based generation in the limit and also be accurate
in an infinite sequence of time steps.
\label{stmt:inf-accuracy}
\end{thm}

Below is a review of this algorithm \cite{kleinberg2025density} and some of its key properties.

Let the countable collection of languages be $\coll = \{\lang{1}, \lang{2}, \dots, \}$. 
Define $\res{\coll}{n} = \{\lang{1}, \lang{2}, \dots, \lang{n}\}$, i.e., the first $n$ languages in $\coll$.

\begin{dfn}\label{dfn:sc}
A language $\lang{n}$ is {\em consistent} at step $t$ if it contains all the elements input by the adversary up to time $t$. 

We say a language $\lang{n}$ is {\em strictly critical} 
at step $t$ if
$\lang{n}$ is consistent at time $t$, and for every language
$\lang{i} \in \res{\coll}{n}$ that is consistent at time $t$,
we have $\lang{n} \subsetneq \lang{i}$.
\end{dfn}

We now describe the algorithm $\Acc$ that achieves Theorem
\ref{stmt:inf-accuracy}.

At time $t$, consider the sequence of all strictly critical languages, which is a strictly descending chain of languages (under set inclusion) 
\[ \mC_t = \lang{{c_t(1)}}, \lang{{c_t(2)}}, \lang{{c_t(3)}}, \ldots.\] 
Note that the sequence could be either finite or infinite. It satisfies
\[ \lang{{c_t(1)}} \supset \lang{{c_t(2)}} \supset  \lang{{c_t(3)}} \supset \ldots.\]

Let $h_t$ be maximum $j$ such that $c_t(j) \leq t$.

We define $i_1$ arbitrarily, and then for $t > 1$, we define $i_t$ as 
follows.
\begin{itemize}
\item[(a)] If $w_t$ belongs to all the strictly critical languages from 
time $t-1$ (that is, if $w_t \in \lang{{c_{t-1}(j)}}$ for all $j$)
then we define $i_t = c_{t-1}(h_{t-1})$.
\item[(b)] Otherwise, $w_t$ does not belong to all
the strictly critical languages from time $t-1$.
Because the strictly critical languages are nested by proper inclusion,
this means there is some index $k_{t-1}$ such that
$w_t \in \lang{{c_{t-1}(j)}}$ for all $j \leq k_{t-1}$ and 
$w_t \not\in \lang{{c_{t-1}(j)}}$ for all $j > k_{t-1}$.
We define $i_t = c_{t-1}(k_{t-1})$. If no such $k_{t-1}$ exists, then simply let $i_t = c_{t-1}(h_{t-1})$. This edge case that $k_{t-1}$ does not exist will never appear after some finite amount of time.
\end{itemize}
In each step, the algorithm guesses $i_t$ as its index for the language
$\lang{{i_t}}$. The crucial part of this algorithm is that at time $t$, the judgement is based on the sequence of strictly critical languages at time $t-1$ instead of time $t$.

There are some properties we will be using  for our new algorithm. 

\begin{claim}\cite{kleinberg2025density,kleinberg2026density}\label{claim:Accproperty}
    \begin{enumerate}
    \item The chain $\mC_t$ is always a strictly descending chain of languages (under set inclusion) $\lang{{c_t(1)}} \supset \lang{{c_t(2)}} \supset  \lang{{c_t(3)}} \supset \ldots$.
        \item There is always at least one strictly critical for every time step $t$;
        \item If a language $\lang{n}$ is strictly critical in step $t$, then $\lang{n}$
will continue be strictly critical in every step $t' > t$
for which it is consistent at step $t'$.
\item There exists a time step $\zt$ such that for all $t \geq \zt$,
the true language $\lang{z}$ is strictly critical at step $t$. (i.e., $L_z \in \mC_t$). 
\label{stmt:true-strictly-critical}
\item After some finite time, in $\mC_t$, all the languages that appear before $L_z$ will be strict supersets of $L_z$, and the chain up to $L_z$ will not change afterwards. 
    \end{enumerate}
\end{claim}

\subsubsection{Some difficulties not present when working with asymptotic lower density. }

We now explain why our lower-asymptotic-density algorithms from
\cite{kleinberg2025density, kleinberg2026density} cannot simply be reused. The issue is not validity but
locality. Prefix-density arguments may charge a missing late string to
generated strings earlier in the order. Banach density forbids this
charge: a sparse interval must be repaired inside that interval. Consequently, the algorithm must be sensitive to the spatial location of the revealed elements, rather than blindly prefer smaller elements. This is
why the algorithm needs both a static Cantor--Bendixson tree and local
pods around adversarial inputs.

Another major convenience with asymptotic density is that it depends only on the poset structure of the language, where the order is given by set inclusion. In particular, all previous results concerning asymptotic density rely almost solely on working with the descending chain $\mathcal{C}_t$. However, when moving upward along the chain $\mathcal{C}$, the adversary can, in principle, stack a long interval that consists almost entirely of adversarial inputs. This will make the Banach lower density (or in general, $f$-window density when $f$ is unbounded) of the algorithm output very low.
This phenomenon, however does not pose a problem for asymptotic lower density. The reason is that every interval used to measure the density is anchored at $1$, which prevents the adversary from taking advantage of such stacked intervals to significantly reduce the measured density. 
For example, when the interval is anchored at $1$, one can employ the charging method \cite{kleinberg2025density} or the pod method \cite{kleinberg2026density}. The key idea is that whenever the algorithm appears to ``climb'' up the chain, this ascent must be preceded by a descent along the chain. During such a descent in the past, the algorithm has already charged a large number of algorithm-generated inputs with smaller elements.
Consequently, even if we observe a long interval consisting primarily of adversarial inputs, there must exist many algorithm-generated elements prior to this interval. As a result, when measuring the ratio of elements occupied by the adversary versus those generated by the algorithm in an interval starting from the origin, the presence of a long adversarial interval does not significantly affect the overall density bound.

These techniques no longer work in the case when working with $f$-window density for non-constant functions $f$, that is, when we move away from the asymptotic density setting. As shown in Theorem~\ref{thm:lowerdensity0f}, we can construct examples in which the $f$-window density can be made arbitrarily close to zero. The long intervals consisting predominantly of adversarial inputs arise precisely from climbing up a descending chain of languages, allowing the adversary to exploit windows that moving away from the origin.

\paragraph{Avoiding zero $f$-window density bounds.}
A natural question is how one might avoid the $f$-window density being close to zero, so that the $f$-window density could be strictly positive or even close to the optimal bound $1/2$?

\textbf{A first naive approach} would be to restrict attention to collections of languages that do not contain arbitrarily long inclusion chains. While such an assumption would indeed preclude the construction underlying Theorem~\ref{thm:lowerdensity0f}, it would also make the problem largely uninteresting. As observed already in the original paper~\cite{kleinberg2024limit}, the question of whether any nontrivial generation breadth can be guaranteed arises precisely because arbitrarily long descending chains exist, forcing the algorithm to continue descending in order to avoid overshooting.

On the other hand, whenever there exists an infinite inclusion chain of languages with a well-defined limit $K$, one almost inevitably encounters arbitrarily long chains of languages $\mathcal{C}_t$. Consequently, in all genuinely interesting settings, structural features akin to arbitrarily long ``chains'' of languages cannot be avoided. Any attempt to control the $f$-window density must therefore accommodate these long chain-like structures rather than seek to eliminate them.

We overcome the ``long chain'' issue while still allowing such chains to persist by exploiting the underlying topology. In particular, the algorithm must be designed to leverage not only the fact that $\mathcal{C}_t$ forms a descending chain—which, as we have seen, may be arbitrarily long in all interesting examples—but also additional structural features.

To illustrate the subtlety of this difficulty, we consider the following set of examples. We first introduced these examples  in~\cite{kleinberg2025density} to demonstrate the necessity of the charging argument and the ``pullback'' technique. They also serve, in a sense, as worst-case examples that provide a useful testbed for our new algorithm, highlighting the delicacy involved in designing methods suitable for studying lower Banach density, and in general lower $f$-window density.

\begin{example}\label{example:1}
All the strings are subsets of $\mathbb{N}$. We start with some markers, at integers $a_0 = 0, a_1 = 3, a_2 = 3^2, a_3 = 3^3, \dots$. The markers get sparser and sparser. Let $\lang{n}$ be the language $\bigcup_{i = 0}^\infty [a_i, a_i + n]$. So $L_n$ is attaching an interval of length $n$ to each of the marker $a_i = 3^i$. 
The pictures below illustrates $\lang{1}$ and $\lang{4}$. 

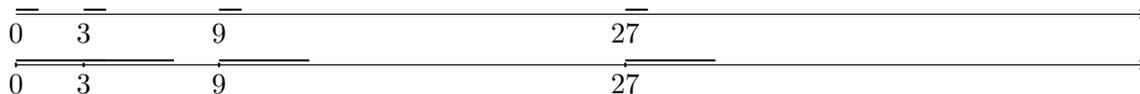
\vspace{0.2in}
\begin{tikzpicture}[scale=0.3]
    \draw[->] (0,0) -- (50,0) node[right] 
    {};
    
    \foreach \x in {0,3,9,27,81} {
        \draw[thick] (\x,0.2) -- ({\x+1},0.2);
        \node[below] at (\x,0) {\x};
    }
\end{tikzpicture}

\begin{tikzpicture}[scale=0.3]
    \draw[->] (0,0) -- (50,0) node[right] {};
    
    \foreach \x in {0,3,9,27,81} {
        \draw[thick] (\x,0.2) -- ({\x+4},0.2);
        \draw[thick] (\x,0.1) -- (\x,-0.1);
        \node[below] at (\x,0) {\x};
    }
\end{tikzpicture}
So we have an infinite sequence of nested languages $\lang{1} \subsetneq \lang{2} \subsetneq \lang{3} \subsetneq \cdots$. Let the true language be $K = \mathbb{N}$. 
Clearly $L_1 \subset L_2 \subset \dots$ and this is an increasing chain whose union is $K = \mathbb{N}$. 

We will create a family of languages from these building blocks. 

For each positive integer \( i \), let \(\lang{i}^{(0)}\) be the language $\lang{i}$ defined above. Define the strictly increasing function \( f_i^{(1)}(k): \mathbb{N}_+ \to \mathbb{N}_+ \) such that its image consists of the strings of \(\lang{i}^{(0)}\). Such an \( f_i^{(1)} \) is thus unique. For each positive integer \( j \), we construct the language \(\lang{i:j}^{(1)}\) so that its \( \ell \)-th string is given by \( f_i^{(1)}(\lang{j}^{(0)}(\ell)) \), where \(\lang{j}^{(0)}(\ell)\) denotes the \( \ell \)-th string of \(\lang{j}^{(0)}\), ordered from smallest to largest.
This construction ensures that
\[
\bigcup_{j=1}^\infty \lang{i:j}^{(1)} = \lang{i}^{(0)}.
\]

It is easy to see that the set of languages with the same superscript form an infinite perfect tower towards $K = \mathbb{N}$. Among this subset, we  can thus pick infinitely many subsequences which is an infinite perfect tower towards ${K}$. However, if in each step, the algorithm does not overshoot, then the adversary can create arbitrarily long intervals purely consisting of adversary elements. 
For example, if the true language is $K$ but the adversary made the algorithm to guess language in each step is $L_1^{(0)}, L_2^{(0)}, \dots$, then the adversary will make the lower Banach density zero as again we have an infinite  chain  
\[L_1^{(0)} \subset  L_2^{(0)} \subset L_3^{(0)}\subset  \dots \]
So we have to make sure that the algorithm will overshoot occasionally to block these long adversary intervals, but also be careful to still guarantee validity. 

\textbf{One seemingly natural fix} is to enforce that whenever the process transitions from $L_i^{(0)}$ to a different $L_j^{(0)}$, the algorithm first pulls back to $\mathbb{N}$; that is, it generates several elements from $\mathbb{N}$ before proceeding to $L_j^{(0)}$. Validity is preserved, since the algorithm overshoots only finitely many times. This modification resolves the issue for the specific sequence $L_1^{(0)}, L_2^{(0)}, \dots$.
Motivated by this observation, one might attempt a more general remedy: pull back to the \emph{smallest} language that is a limit point of the current sequence (possibly together with a tail). Requiring minimality is natural, as it ensures that validity is maintained.

However, this approach turns out to be insufficient. To see this, suppose the adversary forces the algorithm to produce the sequence
\[
\lang{1:1}^{(1)}, \lang{1:2}^{(1)}, \dots, \lang{2:1}^{(1)}, \lang{2:2}^{(1)}, \dots, \lang{3:1}^{(1)}, \lang{3:2}^{(1)}, \dots.
\]
(As $\Acc$ must be correct infinitely often, we may insert $K = \mathbb{N}$ sparsely into the sequence without affecting the argument.)
Following the above strategy, the algorithm would pull back to the sequence $L_1^{(0)}, L_2^{(0)}, \dots$. However, this pullback sequence exhibits the same difficulty as before: by enforcing this ordering at each time step, the adversary can still drive the Banach density arbitrarily close to zero.
Thus, the seemingly natural strategy of always pulling back to the minimal limit point \textbf{does not suffice}. \end{example}

It turns out that the appropriate way to tackle this problem is to leverage the Cantor--Bendixson rank on the topological space defined earlier in the paper, together with multiple layers of pullback. The key property of the Cantor--Bendixson derivatives is that the limit of any sequence of points of the same rank must have strictly smaller rank. In the example above, $L_{i:j}^{(a)}$ has Cantor--Bendixson rank $a+1$.

A natural way to reason about the number of pullback layers required is via a tree structure, where one pulls back to the least common ancestor. However, the Cantor--Bendixson topology does not canonically induce such a tree (see Figure~\ref{fig:ex1setinclusion}). One obstacle is that being a limit point is inherently a limiting notion: for instance, it is not meaningful to assert that $L_2^{(0)}$ is the limit of a single point such as $L_{2:1}^{(1)}$, and hence no edge should be drawn between them in a tree representation.
Nevertheless, in Example~\ref{example:1}, the languages admit a convenient indexing that allows us to impose a tree structure by connecting each $L_{i:j}^{(a)}$ to $L_i^{(a-1)}$ (see Figure~\ref{fig:ex1}). For fixed $i$ and $a$, the sequence $\{L_{i:j}^{(a)}\}_j$ then converges to $L_i^{(a-1)}$. This perspective resolves the earlier density-zero issue: Although the sequence of hypothesized languages could take the form
$\lang{1:1}^{(1)}, \lang{1:2}^{(1)}, \dots, \lang{2:1}^{(1)}, \lang{2:2}^{(1)}, \dots, \lang{3:1}^{(1)}, \lang{3:2}^{(1)}, \dots,$
once the sequence transitions from the subtree rooted at $L_1^{(0)}$ to that rooted at $L_2^{(0)}$, we are forced to pull back to their least common ancestor, namely $\mathbb{N}$.

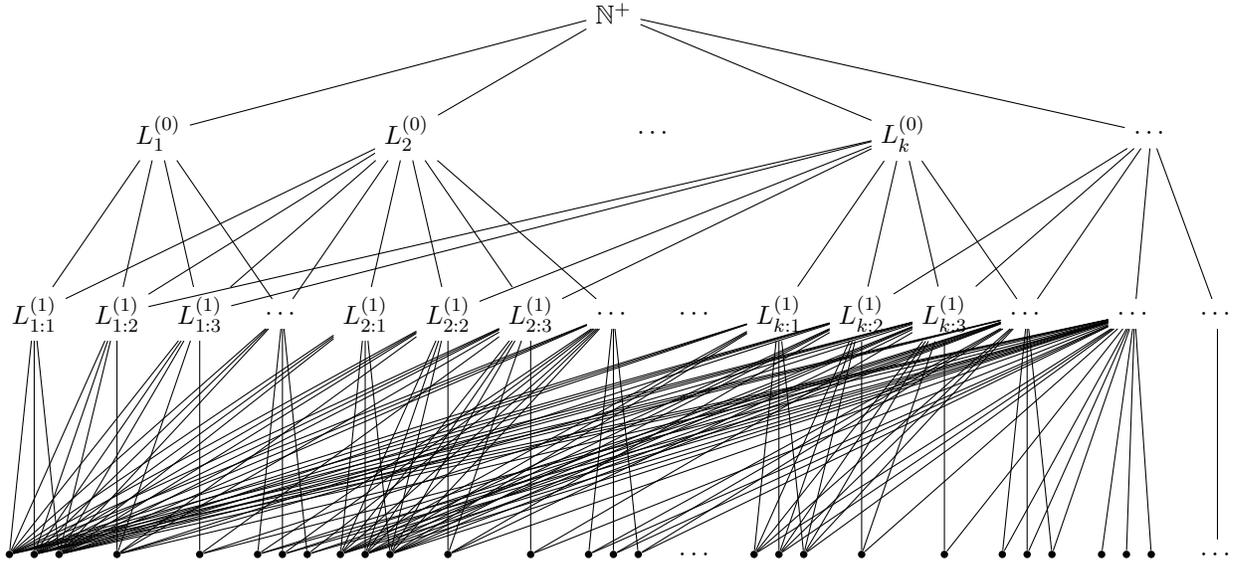
\begin{figure}[h]
\begin{tikzpicture}[x=1.1cm, y=1.6cm, >=stealth, every node/.style={scale=0.9}]

    \node (N) at (1.5, 0) {$\mathbb{N}^+$};

    \node (L10) at (-4, -1) {$L_1^{(0)}$};
    \node (L20) at (-1, -1) {$L_2^{(0)}$};
    \node (dots0) at (2, -1) {$\cdots$};
    \node (Lk0) at (5, -1) {$L_k^{(0)}$};
    \node (LdotsRight) at (8, -1) {$\cdots$};

    \draw (N) -- (L10);
    \draw (N) -- (L20);
    \draw (N) -- (Lk0);
    \draw (N) -- (LdotsRight);

    \node (L11) at (-5.5, -2.5) {$L_{1:1}^{(1)}$};
    \node (L12) at (-4.5, -2.5) {$L_{1:2}^{(1)}$};
    \node (L13) at (-3.5, -2.5) {$L_{1:3}^{(1)}$};
    \node (L1d) at (-2.5, -2.5) {$\cdots$};
    
    \node (L21) at (-1.5, -2.5) {$L_{2:1}^{(1)}$};
    \node (L22) at (-0.5, -2.5) {$L_{2:2}^{(1)}$};
    \node (L23) at (0.5, -2.5) {$L_{2:3}^{(1)}$};
    \node (L2d) at (1.5, -2.5) {$\cdots$};

    \node (Lcd) at (2.5, -2.5) {$\cdots$};

    \node (Lk1) at (3.5, -2.5) {$L_{k:1}^{(1)}$};
    \node (Lk2) at (4.5, -2.5) {$L_{k:2}^{(1)}$};
    \node (Lk3) at (5.5, -2.5) {$L_{k:3}^{(1)}$};
    \node (Lkd) at (6.5, -2.5) {$\cdots$};

    \node (Lend1) at (7.8, -2.5) {$\cdots$};
    \node (Lend2) at (8.8, -2.5) {$\cdots$};

    \foreach \n in {L11, L12, L13, L1d} \draw (L10) -- (\n);
    \foreach \n in {L11, L12, L13, L1d} \draw (L20) -- (\n);
    \foreach \n in {L21, L22, L23, L2d} \draw (L20) -- (\n);
    \draw (Lk0) -- (L12); \draw (Lk0) -- (L13);
    \draw (Lk0) -- (L22); \draw (Lk0) -- (L23);
    \foreach \n in {Lk1, Lk2, Lk3, Lkd} \draw (Lk0) -- (\n);
    \draw (LdotsRight) -- (Lk2); 
    \draw (LdotsRight) -- (Lk3);
    \draw (LdotsRight) -- (Lkd);
    \draw (LdotsRight) -- (Lend1);
    \draw (LdotsRight) -- (Lend2);

    \foreach \pos/\id in {
        -5.8/11a, -5.5/11b, -5.2/11c,
        -2.8/1da, -2.5/1db, -2.2/1dc,
        -1.8/21a, -1.5/21b, -1.2/21c,
         1.2/2da,  1.5/2db,  1.8/2dc,
         3.2/k1a,  3.5/k1b,  3.8/k1c,
         6.2/kda,  6.5/kdb,  6.8/kdc,
         7.4/enda, 7.7/endb, 8.0/endc} 
    {
        \node[fill, circle, inner sep=1.1pt] (Leaf\id) at (\pos, -4.5) {};
    }
    
    \node[fill, circle, inner sep=1.1pt] (Leaf12) at (-4.5, -4.5) {};
    \node[fill, circle, inner sep=1.1pt] (Leaf13) at (-3.5, -4.5) {};
    \node[fill, circle, inner sep=1.1pt] (Leaf22) at (-0.5, -4.5) {};
    \node[fill, circle, inner sep=1.1pt] (Leaf23) at (0.5, -4.5) {};
    \node[fill, circle, inner sep=1.1pt] (Leafk2) at (4.5, -4.5) {};
    \node[fill, circle, inner sep=1.1pt] (Leafk3) at (5.5, -4.5) {};

    \node (Leafcd) at (2.5, -4.5) {$\cdots$};
    \node (Leafendd) at (8.8, -4.5) {$\cdots$};

    \tikzset{
        solid/.style={black}
    }

    \foreach \src in {L11, L12, L13, L1d} {
        \foreach \tgt in {Leaf11a, Leaf11b, Leaf11c} \draw[solid] (\src) -- (\tgt);
    }
    \draw[solid] (L12) -- (Leaf12);
    \draw[solid] (L13) -- (Leaf12); \draw[solid] (L13) -- (Leaf13);
    \foreach \tgt in {Leaf1da, Leaf1db, Leaf1dc} \draw[solid] (L1d) -- (\tgt);

    \foreach \src in {L21, L22, L23, L2d} {
        \foreach \tgt in {Leaf11a, Leaf11b, Leaf11c} \draw[solid] (\src) -- (\tgt);
        \foreach \tgt in {Leaf21a, Leaf21b, Leaf21c} \draw[solid] (\src) -- (\tgt);
    }
    \draw[solid] (L22) -- (Leaf12);
    \draw[solid] (L23) -- (Leaf12); \draw[solid] (L23) -- (Leaf13);
    \foreach \tgt in {Leaf1da, Leaf1db, Leaf1dc} \draw[solid] (L2d) -- (\tgt);
    \draw[solid] (L22) -- (Leaf22);
    \draw[solid] (L23) -- (Leaf22); \draw[solid] (L23) -- (Leaf23);
    \foreach \tgt in {Leaf2da, Leaf2db, Leaf2dc} \draw[solid] (L2d) -- (\tgt);

    \foreach \src in {Lk1, Lk2, Lk3, Lkd} {
        \foreach \tgt in {Leaf11a, Leaf11b, Leaf11c} \draw[solid] (\src) -- (\tgt);
        \foreach \tgt in {Leaf21a, Leaf21b, Leaf21c} \draw[solid] (\src) -- (\tgt);
        \foreach \tgt in {Leafk1a, Leafk1b, Leafk1c} \draw[solid] (\src) -- (\tgt);
    }
    \draw[solid] (Lk2) -- (Leaf12);
    \draw[solid] (Lk3) -- (Leaf12); \draw[solid] (Lk3) -- (Leaf13);
    \foreach \tgt in {Leaf1da, Leaf1db, Leaf1dc} \draw[solid] (Lkd) -- (\tgt);
    \draw[solid] (Lk2) -- (Leaf22);
    \draw[solid] (Lk3) -- (Leaf22); \draw[solid] (Lk3) -- (Leaf23);
    \foreach \tgt in {Leaf2da, Leaf2db, Leaf2dc} \draw[solid] (Lkd) -- (\tgt);
    \draw[solid] (Lk2) -- (Leafk2);
    \draw[solid] (Lk3) -- (Leafk2); \draw[solid] (Lk3) -- (Leafk3);
    \foreach \tgt in {Leafkda, Leafkdb, Leafkdc} \draw[solid] (Lkd) -- (\tgt);

    \foreach \tgt in {Leafenda, Leafendb, Leafendc} \draw[solid] (Lend1) -- (\tgt);
    \draw[solid] (Lend2) -- (Leafendd);
    
    \foreach \tgt in {Leaf11a,Leaf11b,Leaf11c,Leaf12,Leaf13,Leaf1da,Leaf1db,Leaf1dc,Leaf21a,Leaf21b,Leaf21c,Leaf22,Leaf23,Leaf2da,Leaf2db,Leaf2dc,Leafk1a,Leafk1b,Leafk1c,Leafk2,Leafk3,Leafkda,Leafkdb,Leafkdc} \draw[solid] (Lend1) -- (\tgt);

\end{tikzpicture}
\caption{Graph representing Example \ref{example:1}, where edges between adjacent CB ranks indicate set inclusion}
\label{fig:ex1setinclusion}
\end{figure}

\begin{figure}[h]
\begin{tikzpicture}[
 level distance=1.5cm,
 level 1/.style={sibling distance=3cm},
 level 2/.style={sibling distance=0.7cm},
 level 3/.style={sibling distance=0.5cm} 
]
\node (N) {$\mathbb{N}_+$}
    child {node (K1) {$\lang{1}^{(0)}$}
        child {node (L1) {$\lang{1:1}^{(1)}$}
            child {node {}}
            child {node {}}
        }
        child {node (L2) {$\lang{1:2}^{(1)}$}
            child {node {}}
            child {node {}}
        }
        child {node (Ld) {$\lang{1:3}^{(1)}$}
            child {node {}}
            child {node {}}
        }
        child {node (La) {$\dots$}
            child {node {}}
            child {node {}}
        }
    }
    child {node (K2) {$\lang{2}^{(0)}$}
        child {node (Lb1) {$\lang{2:1}^{(1)}$}
            child {node {}}
            child {node {}}
        }
        child {node (Lb2) {$\lang{2:2}^{(1)}$}
            child {node {}}
            child {node {}}
        }
        child {node (Lbd) {$\lang{2:3}^{(1)}$}
            child {node {}}
            child {node {}}
        }
        child {node (Lb) {$\dots$}
            child {node {}}
            child {node {}}
        }
    }
    child {node (Kd) {$\dots$}}
    child {node (Ki) {$\lang{i}^{(0)}$}
        child {node (Li1) {$\lang{i:1}^{(1)}$}
            child {node {}}
            child {node {}}
        }
        child {node (Li2) {$\lang{i:2}^{(1)}$}
            child {node {}}
            child {node {}}
        }
        child {node (Lid) {$\lang{i:3}^{(1)}$}
            child {node {}}
            child {node {}}
        }
        child {node (Li) {$\dots$}
            child {node {}}
            child {node {}}
        }
    }
    child {node (Kd) {$\dots$}}
;
\end{tikzpicture}
\caption{One possibility of a spanning tree in the the Cantor-Bendixson directed graph of Example \ref{example:1}}
\label{fig:ex1}
\end{figure}
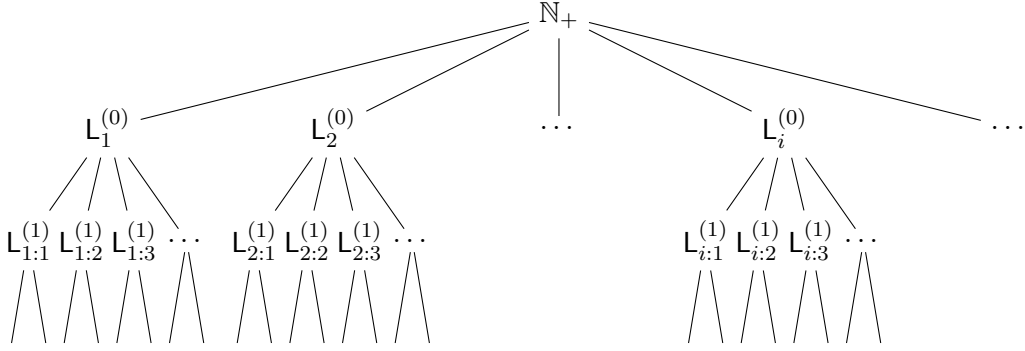

However, in more general settings, it is highly ambiguous how to extract an appropriate tree structure. Even in Example~\ref{example:1}, the choice is not canonical: for instance, $L_{1:1}^{(1)}$ may also lie in a sequence converging to $L_2^{(0)}$. Why, then, should we omit an edge from $L_{1:1}^{(1)}$ to $L_2^{(0)}$ while retaining only the edge to $L_1^{(0)}$?

More broadly, the collection of languages can be intricate and highly intertwined, making it infeasible to partition them into a clean tree structure. Consequently, determining an appropriate spanning tree is far from straightforward. The topology guarantees only a ranked (and potentially very dense) directed graph, where edges may connect nodes across different layers. If one were to include all edges corresponding to set inclusion (i.e., from a subset to a superset), the resulting graph would be overwhelmingly dense. This raises a fundamental question: which edges should be retained in order to obtain a meaningful spanning tree?
There are additional subtleties as well. For example, in some cases it may not even be possible to explicitly describe an infinite sequence, either because minimal or maximal elements fail to exist, or because the sequence itself contains an overabundance of limit points  in which case it is even unclear how we should index the languages in a ``sequence" (because the space $\mX$ might not be complete with respect to the topology).

A more fundamental difficulty lies in selecting a tree that ensures the resulting algorithm satisfies validity. Specifically, we must guarantee that, after some finite time, the algorithm only outputs languages contained in the subtree rooted at the true language $K$. At the same time, this requirement must be balanced with the density constraint: it is essential that any lateral move forces a pullback of at least one level.
Furthermore, we require the tree to be \emph{static}, rather than dynamically adjusted throughout the process, in order to obtain a uniform upper bound on the total number of upward traversals. This imposes additional constraints, as the tree must remain suitable regardless of the adversarial enumeration.
In the algorithm described below, we construct the spanning tree $\mT$ carefully so as to address all of these challenges. 

\subsubsection{Algorithm}
\label{subsubsec:define-alg}

The new algorithm builds on $\Acc$, with the key innovation being the incorporation of the Cantor–Bendixson rank of each element. The construction involves several delicate choices of parameters and setup; we highlight these decisions at the points where they arise.

Since the Cantor--Bendixson rank of $\mX$ is finite, say $r$, we assign a rank to each element $L \in \mX$ as follows.
By definition, $r$ is the smallest integer such that
$ \mX^{(r)} = \emptyset. $
We define elements of $\mX^{(r-1)}$ to have rank $r$. More generally, for each $1 \le i \le r$, elements in
$ \mX^{(i-1)} \setminus \mX^{(i)} $
are assigned rank $i$. In particular, the elements in
$ \mX^{(0)} \setminus \mX^{(1)} = \mX \setminus d(\mX) $
are exactly those of rank $0$.
By the definition of the derived sets, each rank $0,1,\dots,r$ is attained by at least one element of $\mX$. For $L \in \mX$, we write
\[ \mathrm{CB}(L) \]
to denote its Cantor--Bendixson rank.


We define a {\it CB Structural tree} $\mT$ as follows. 

\paragraph{CB Structural Tree $\mT$.}
For each language $L = L_n$, where $n$ denotes its index in the ordering of $\mX$, we define a directed edge from $L$ as follows. The edge points to a language $L'$ among the first $n-1$ languages of $\mX$ satisfying:
(i) $CB(L') > CB(L)$, and
(ii) there exists a sequence in $\mX$ that includes $L$ and has $L'$ as a limit point.
Among all such candidates, we choose $L'$ with the smallest index in $\mX$. If no such $L'$ exists, no edge is drawn from $L$.

In this way, we obtain a static directed tree that depends only on $\mX$ and the underlying topology. \footnote{This is in contrast to the dynamic tree considered in \cite{kleinberg2025density}, which depends on $\mC_t$. Such a dynamic construction would not suffice in our setting: if the tree evolves over time, the relationships between vertices may change, making it impossible to bound the number of times we climb up the chain in $\mC_t$. Furthermore, as discussed earlier, working with a general directed graph is insufficient; a tree structure is essential for our analysis.}

\begin{remark}
It turns out that if $L$ is the $n$-th language in the original ordering of $\mX$, then defining the parent of $L$ to be the minimal superset of $L$ with strictly larger Cantor--Bendixson rank among the first $n-1$ languages is the correct choice. A natural alternative is to select the ``minimum'' superset of $L$ among all languages of higher Cantor--Bendixson rank; while one can address the issue that such a minimum need not exist, this approach still fails. (Please see the discussion in the previous subsection). 
The restriction to the first $n-1$ elements is the crucial ingredient that makes Claim~\ref{claim:parent}  hold. In particular, it ensures that after some finite stage, the identified languages are always descendants of the true language $K$, which in turn guarantees the correctness of the algorithm.
\end{remark}

\vspace{-0.5in}
\begin{remark}
In our definition, the tree is ranked in the usual sense; however, we do not require this rank to coincide with the Cantor--Bendixson rank of the vertices in $\mX$. While it would be convenient to align these two notions of rank, this appears difficult to guarantee from the outset, and especially after removing inconsistent languages from the tree during the process of the algorithm. Instead, we only require that the parent of any vertex has strictly larger Cantor--Bendixson rank than the vertex itself, without imposing that their ranks differ by exactly one. This relaxation turns out to be sufficient for our subsequent analysis.
\end{remark}

\begin{defn}
Given a subset $A \subset \mathbb{N}$ and two elements $x,y \in A$, the \emph{relative distance} between $x$ and $y$, denoted \[d_A(x,y),\] is the difference between their positions in the increasing ordering of $A$.
\end{defn}

For example, if $A = \{1,4,5,7,8\}$ and $x=4$, $y=7$, then $x$ is the second smallest element of $A$ while $y$ is the fourth smallest element; hence $d_A(x,y)=2$.

It is essential to work with relative distance rather than the usual distance in $\mathbb{N}$. There is no way for the algorithm to know the true language, and therefore it cannot measure distances between elements relative to the ground truth, nor can it determine how dense the true language is in $\mathbb{N}$. Instead, all measurements must be defined with respect to the current hypothesis, which motivates the notion of relative distance.

After the setup, we describe the algorithm as follows.

For the notation, we say that an element is \emph{available} at time $t$ if it has not been used by either the adversary or the algorithm up to time $t$.
\paragraph{Algorithm.}
Let $(s_t)_{t \in \mathbb{N}}$ be a sequence of increasing positive integers and $(z_n)_{n \in \mathbb{N}}$ an increasing sequence of positive integers. These will be related to the ``pod diameter" in the algorithm. 

Initialize pod $P = \emptyset$. Pod will be the set of candidate set of strings that the algorithm will prioritize to generate. 

Suppose we have just completed time $t-1$ and are now at time $t$. The adversary presents the input $w_t$. Recall that $L_{i_{t-1}}$ denotes the language identified by $\Acc$ at time $t-1$.

\begin{enumerate}
\item If $L_{i_t} = L_{i_{t-1}}$, define $PB(t)= L_{i_t}$. 

\item If $L_{i_t} \neq L_{i_{t-1}}$, define the {\it pullback} $PB(t)$ as the common ancestor of $L_{i_{t-1}}$ and $L_{i_{t}}$ in $\mT$ after removing inconsistent languages if exists. If it does not exist, then $PB(t) = L_{i_{t}}$. 
\end{enumerate}
In either situation, 
add to $P$ the following elements: for each $w_i$, $i \leq t$, define the diameter parameter
\[
\gamma_t = \max(\gamma_{t-1}+1, \max(s_t, z_{w_i})).
\] (the first $\gamma_{t-1}+1$ is to make sure that $\gamma_t$ is increasing)
from $PB(t)$ the $\gamma_t$ available elements immediately preceding $w_i$ in $PB(t)$, and the $\gamma_t$ available elements in $PB(t)$ that are immediately following $w_i$, and then plus $w_i$. (For convenience of the presentation of the proof later, we also include in $P$ all the elements in $PB(t)$ that are not available).

Write the current pod $P = P_t$. 

Let the algorithm output $o_t$ be the smallest available integer in $L_{i_{t}} \cup P_t$ that minimizes $d_{L_{i_t} \cup P_t \cup \{w_t\}}(o_t, w_t)$. 
\begin{remark}
It is crucial for the later analysis that $\gamma_t$ is increasing, regardless of the elements used at time $t$. However, it is not sufficient to define $\gamma_t$ solely as a function of $t$, as is done when working with asymptotic lower density. In the case of lower Banach density, the relevant window size is sensitive to the location in $\mathbb{N}$. Consequently, if the algorithm or the adversary produces elements that are large in $\mathbb{N}$ too quickly, the growth rate of $\gamma_t$ must increase accordingly to keep pace with this progression. Therefore, $\gamma_t$ must depend not only on $t$, but also on the inputs provided by the adversary.
\end{remark}

\vspace{-0.5in}
\begin{remark}
It is crucial that the algorithm makes decisions based on the relative density 
$d_{L_{i_t} \cup P_t \cup \{w_t\}}(o_t, w_t)$, rather than relative to the true language $K$ or to $\mathbb{N}$. The former is impossible since the algorithm has no access to $K$, while the latter is inadequate because the distribution of $K$ within $\mathbb{N}$ may be highly non-uniform and unknown over the unrevealed elements. 
However, this reliance on relative density is a primary source of difficulty in the analysis, as it complicates the task of establishing guaranteed density bounds.
\end{remark}

\subsubsection{Validity guarantee}
\label{subsubsec:validity}

We first show that the algorithm converges in the limit.
To this end, we use the following simple but important property of the Cantor--Bendixson rank.
\begin{claim}
Suppose $K$ is the limit point of a sequence of languages $J_1, J_2, \dots$. Then there exists an index $n$ such that for all $j > n$, the Cantor--Bendixson rank of $J_j$ is strictly smaller than that of $K$.
\end{claim}
\begin{proof}
Suppose, for contradiction, that the claim does not hold. Then there exists a subsequence of $J_1, J_2, \dots$ that still converges to $K$ and such that each element in the subsequence has Cantor--Bendixson rank at least that of $K$. With a slight abuse of notation, denote this subsequence again by $J_1, J_2, \dots$.
Let $\alpha = CB(K)$. Since $CB(J_n) \ge \alpha$ for all $n$, it follows from the definition of the Cantor--Bendixson rank that $J_n \in \mX^{(\alpha)}$ for every $n$. Hence, $\{J_n \mid n \in \mathbb{N}\} \subseteq \mX^{(\alpha)}$.
Since the sequence $(J_n)$ converges to $K$, this implies that $K$ is a limit point of $\mX^{(\alpha)}$, and therefore $K \in \mX^{(\alpha+1)}$. This implies that $CB(K) \ge \alpha + 1$, contradicting the assumption that $CB(K) = \alpha$.
\end{proof}

By the property of $\Acc$ in Claim~\ref{claim:Accproperty}, after some finite time, $K$ is always strictly critical, and the identified language $L_{i_t}$ is a subset of the true language $K = L_z$. 
Moreover, we may assume that we are beyond a finite stage at which all remaining consistent languages $L_i$ with $i < z$ are supersets of $K$. Indeed, among these first $z-1$ languages, any language that is not a superset of $K$ will eventually be inconsistent when the adversary presents a string that belongs to $K$ but not to that language.

It suffices to show that $PB(t)$ is also a subset of $K$ after some finite time.

\begin{lem}\label{lem:PBsubsetK}
After some finite time, all identified languages $L_{i_t}$ lie in the subtree rooted at $K$. Consequently, $PB(t)$ also lies in the subtree rooted at $K$, which immediately implies that $PB(t) \subseteq K$.
\end{lem}
\begin{proof}
It is clear that if a language $L$ becomes inconsistent, then all languages in the subtree rooted at $L$ also become inconsistent.

Moreover, if both $L_{i_t}$ and $L_{i_{t-1}}$ lie in the subtree rooted at $K$, then by the definition of $\mT$, it follows that $PB(t) \subseteq K$.

Otherwise, since $L_{i_t} \subseteq K$, if $L_{i_t}$ is not in the subtree rooted at $K$, then one of the following must hold: 
(I) $L_{i_t}$ lies in a tree disjoint from the one containing $K$, or 
(II) $L_{i_t}$ lies in the same tree as $K$ but is not a descendant of $K$.

We will prove the following claim.
\begin{claim}\label{claim:parent}
Suppose we are beyond some finite time $T$ at which all languages with index smaller than that of the true language are supersets of $K$ (this is guaranteed by Claim~\ref{claim:Accproperty}). Let $L \subset K$ be a language that remains consistent at some time $t > T$. 

If there exists a sequence containing $L$ that converges to $K$ and $CB(L) < CB(K)$, then $L$ is not a root, and its parent in $\mT$ is a subset of $K$.
\end{claim}
\begin{proof}
    Since  $L \subset K$ and $t > T$, the index of $L$ in the original language ordering will be larger than that of $K$. Since there  is a sequence including $L$ whose limit is $K$ and $CB(L) < CB(K)$, $L$ will not be a root of a tree since $K$ is one option to be its  parent. 

    Let the parent of $L$ be $J \neq K$ as otherwise we are done. If the index of $J$ is smaller that of $K$, by our definition of $T$, $J$ will be a strict superset of $K$, contradicts our construction of $\mT$ the minimality of the parent. Now suppose the index of $J$ is larger that of $K$. If $J$ and $K$ are incomparable in set-inclusion sense, then by our priority when building $\mT$, the parent of $L$ should have been $K$. If $J$ is a superset of $K$, again we should have picked $K$ as its parent instead of $J$. Thus $J$ is a subset of $K$. 
\end{proof}
If case (II) occurs infinitely often, then since the adversary eventually enumerates all strings in $K$ and $L_{i_t} \subseteq K$, these languages form a sequence whose limit is $K$. However, by the definition of the Cantor--Bendixson rank and derived sets, the tail of this sequence must consist of languages whose ranks are strictly smaller than that of $K$.
By Claim~\ref{claim:parent}, for each such language with rank strictly smaller than that of $K$, its lowest ancestor of rank at least that of $K$ is again a strict subset of $K$. Note that such an ancestor must exist: since the language is not a descendant of $K$, its least common ancestor with $K$ is a strict ancestor of $K$, and therefore has Cantor--Bendixson rank at least that of $K$.
These ancestors thus form a sequence whose limit is $K$. However, all of them have Cantor--Bendixson rank equal to that of $K$, contradicting the definition of the Cantor--Bendixson rank.

If case (I) occurs infinitely often, then the tail of the resulting sequence again has $K$ as its limit point, and hence consists of languages whose Cantor--Bendixson ranks are strictly smaller than that of $K$. Let $A_{t_1}, A_{t_2}, \dots$ denote the parents of these identified languages, listed in chronological order.
By Claim~\ref{claim:parent}, the tail of this parent sequence consists of strict subsets of $K$, and thus also forms a sequence whose limit is $K$. It follows that the tail must again have Cantor--Bendixson ranks strictly smaller than that of $K$.
Repeating this argument, we consider the sequence obtained by taking parents iteratively. Each time, the tail of the resulting sequence has limit $K$, and therefore consists of languages with strictly smaller Cantor--Bendixson rank than $K$. After at most $r$ iterations—since the Cantor--Bendixson rank of $\mX$ is at most $r$ and the depth of $\mT$ is at most $r$—we obtain a sequence consisting of roots whose limit is $K$, and hence whose ranks are strictly smaller than that of $K$.
However, each of these roots is also a strict subset of $K$, and by Claim~\ref{claim:parent}, such languages cannot be roots. This is a contradiction.
\end{proof}

\subsubsection{Properties of the Algorithm}
\label{subsubsec:alg-properties}

We now show that we can achieve the $f$-window lower density to the optimal $1/2$. By the monotonicity lemma (Claim~\ref{claim:monotone}), it suffices to establish this for the lower Banach density.

We first introduce the notion of a relative interval.

\begin{defn}[Relative interval]
Given a subset $A \subset \mathbb{N}$ and integers $1 \le a \le b$, the \emph{interval} $I = [a,b]_A$ is defined to be the set of elements of $A$ from the $a$-th smallest to the $b$-th smallest (inclusive) in the natural ordering of $A$.
\end{defn}

The following lemma is the main lemma of the proof.
\begin{lem}\label{lem:ell}
    At the end of the algorithm, for any $\ell > 0$, there exists some $n_0(\ell)$ such that for all $a > n_0$, the set of strings only occupied by the adversary in the interval $[a, a+\ell']_K$ is $1/2 + o_{\ell}(1)$ for all $\ell \leq \ell' \leq 2\ell$. 
\end{lem}

To see that this lemma implies the original theorem statement, we use a  averaging argument.
\begin{proof}[Proof of Lemma \ref{lem:ell} implying Theorem \ref{thm:finiterank}]
We argue by contradiction. Suppose there exists an increasing sequence of intervals $I_1, I_2, \dots$ in $K$, whose lengths tend to infinity, such that in each $I_i$, the fraction of elements occupied solely by the adversary is at least $\tfrac{1}{2} + \epsilon$ for some $\epsilon > 0$ (equivalently, the output density in $I_i$ is at most $\tfrac{1}{2} - \epsilon$).

Fix $\ell$ sufficiently large so that the bound $o_{\ell}(1)$ from the lemma is smaller than $\epsilon/2$. By discarding finitely many intervals, we may assume that each $I_i$ has length at least $4\ell$. 

For each $I_i$, select a subinterval $J_i \subseteq I_i$ whose length lies between $\ell$ and $2\ell$, and whose output density is at most $\tfrac{1}{2} - \epsilon$. Such a subinterval exists by a simple averaging argument over a partition of $I_i$ into subintervals of length between $\ell$ and $2\ell$. Among all such subintervals, choose one whose left endpoint is maximal. Among all such subintervals, choose one whose left endpoint is maximal.

We claim that the left endpoints of the intervals $J_i$ are unbounded. Suppose, for contradiction, that they are bounded above by some constant $C$; that is, all left endpoints lie within the first $C$ positions of $K$. 
Since the lengths of the intervals $I_i$ tend to infinity, for sufficiently large $i$ we may consider a subinterval of $I_i$ of length $m$ starting at position $C' = \max(C, n_0(\ell))$, where $m$ is arbitrarily large. The contribution from the first $C'$ positions is negligible, and thus the output density in $I_i$ is at least
\[
\frac{0 \cdot C' + (m - C' - \ell)\bigl(\tfrac{1}{2} - \epsilon/2\bigr)}{m + C'}
> \tfrac{1}{2} - \epsilon
\]
when $m$ is large (as $\ell$ is now fixed), which contradicts the assumption that each $I_i$ has output density at most $\tfrac{1}{2} - \epsilon$.

Now consider only the subsequence of intervals $J_1, J_2, \dots$ whose left endpoints are greater than $n_0(\ell)$. By construction, each such interval $J_j$ has length $\ell_j \in [\ell, 2\ell]$. 
By the claim, each interval $J_j$ has algorithmic output density at least 
\[
\frac{1}{2} - o_{\ell}(1) > \frac{1}{2} - \frac{\epsilon}{2},
\]
which contradicts our construction that the algorithm output density in $J_j$ is at most $\frac{1}{2} - \epsilon$.
\end{proof}

For the remainder of the proof, we focus on establishing Lemma~\ref{lem:ell}.

We henceforth assume that we are beyond some finite time $T$ after which $K$ remains strictly consistent and $\mC_t$ has stabilized up to $K$.

\vspace{0.1in}
Suppose we are given an interval $I = [a, a+\ell]_K$ relative to $K$, where the left endpoint satisfies $a \ge a_0$ for some sufficiently large $a_0$ such that the pod parameter satisfies $z_{a_0 - 2\ell} \gg 4\ell$.

The pod diameter parameter $z$ and the threshold $a_0$ are chosen to ensure the following property holds. Recall that $P_t$ denotes the pod updated at time $t$.
\begin{claim} 
    If $z_{a_0 - 2\ell} \gg 4\ell$, the the following holds. 
    
Define the extended interval $I_t$ as the union of $I = [a, a + \ell]_K$ with the $\ell$ elements of $P_t$ immediately preceding the smallest element in $I$ and the $\ell$ elements of $P_t$ immediately following the largest element in $I$. 
    
    Whenever an adversary inputs a string $w_t \in I$ at time $t$, if $P_t \cap I$ contains available integers, then the algorithm generates an integer $o_t$ satisfying  $d_{P_t  \cup L_{i_t}}(w_t, o_t) \leq \ell$.  
\end{claim}
\begin{proof}
Since $z_{a_0 - 2\ell} \gg 4\ell$ and $a \ge a_0$, whenever the adversary inputs $w_t \in I$ at time $t$, its pod in $PB(t)$ has diameter at least $z_{a_0 - 2\ell} \ge 4\ell$. Recall $P_t$ denote the pod at time $t$. Then $P_t$ covers all available elements of $PB(t)$ within $I$, since $PB(t) \subseteq K$ by Lemma~\ref{lem:PBsubsetK}.
As the algorithm outputs a string closest to $w_t$ within $P_t \cup L_{i_t}$, if there exists an unused element $x \in P_t \cap I$, then by the length of $I$ we have
$d_{P_t \cup L_{i_t}}(x, w_t) \le \ell - 1.$
Consequently, the algorithm cannot output a string $o_t$ such that
$d_{P_t \cup L_{i_t}}(w_t, o_t) > \ell.$
\end{proof}

Under the same assumption that $I = [a, a+\ell]_K$ where $z_{a - 2\ell} \gg 4\ell$, the following holds. 
\begin{lem}\label{lem:rchange}
Let $t_1 < t_2 < \cdots < t_M$ be the time steps at which the adversary reveals an element in $I$.

Let $Z_1$ be the identified language $L_{i_{t_1}}$ at time $t_1$. For each $1 \le i \le M$, define $Z_i$ to be the lowest common ancestor in $\mT$ of all identified languages over the closed time interval $[t_1, t_i]$.

Then:
\begin{enumerate}
    \item $Z_1 \subseteq Z_2 \subseteq \cdots \subseteq Z_M$.
    \item The sequence $(Z_i)$ changes value at most $r$ times, where $r$ is the Cantor--Bendixson rank of $\mX$.
\end{enumerate}
\end{lem}
\begin{proof}

    We first prove the following simple property of a tree. In a tree, we say for any two vertices $x, y \in V$: $x \preceq y$ if and only if  $x$ lies on the unique simple path from the root to $y$. 

\begin{claim}\label{claim:rtimes}
Let $T$ be a tree with depth $r$. 
Let $v_1, v_2, \dots, v_m$ be an arbitrary sequence of a subset of $m$ vertices in $T$ (they might not be distinct). For each $i \in \{1, 2, \dots, m\}$, let $V_i = \{v_1, v_2, \dots, v_i\}$, and define $u_i$ be the least common ancestor of vertices in $V_i$, denoted as $\LCA(V_i)$. Then, the sequence of vertices $(u_i)_{i=1}^m$ satisfies the monotonic ancestral chain:
$$u_m \preceq u_{m-1} \preceq \dots \preceq u_2 \preceq u_1$$
Furthermore, the number of strict changes in this sequence, defined as the cardinality of the set $C = \{ i \in \{2, 3, \dots, m\} \mid u_i \neq u_{i-1} \}$, is strictly bounded above such that $|C| \leq r$.
\end{claim}
\begin{proof}
For any $i \geq2$, since $V_{i-1} \subseteq V_i$, any common ancestor of $V_i$ must also be a common ancestor of $V_{i-1}$. In particular, $u_i = \LCA(V_i)$ is a common ancestor of $V_{i-1}$, which implies it must be an ancestor of $u_{i-1} = \LCA(V_{i-1})$. Thus, $u_i \preceq u_{i-1}$, establishing the monotonic chain $u_m \preceq u_{m-1} \preceq \dots \preceq u_1$.

To bound $|C|$, note that if $u_i \neq u_{i-1}$, the ancestral relation is strict ($u_i \prec u_{i-1}$). This implies that their depths must satisfy $d(u_i) < d(u_{i-1})$. Because the maximum depth of any vertex in $T$ is $r$, and depths are bounded below by $0$, the depth can strictly decrease at most $r$ times. Therefore, there can be at most $r$ strict changes in the sequence, so $|C| \leq r$.
\end{proof}

In our case, the tree is the static tree $\mT$ we have built, and we know its height is at most $r$ by our construction of $\mT$ and the fact $\mX$ has Cantor-Bendixson rank $r$. The vertex $v_t$ in the claim will be the identified language $L_{t}$ for $t_1 \leq t \leq t_M$. The claim implies 
\begin{align*}
    \LCA_{\mT}(L_{t_1}, L_{t_1+1}, \dots, L_{t_M}) \preceq \LCA_{\mT}(L_{t_1}, L_{t_1+1}, \dots, L_{t_M - 1}) \preceq \\ 
    \dots \preceq \LCA_{\mT}(L_{t_1}, L_{t_1+1}) \preceq \LCA_{\mT}(L_{t_1}).
\end{align*}
Since in $\mT$, an ancestor is always a superset of its descendant, we have that $Z_1 \subset Z_2 \subset \dots \subset Z_M$. 

Furthermore, the claim implies that the whole chain with $t_M - t_1+1$ sets \[\LCA_{\mT}(L_{t_1}, L_{t_1+1}, \dots, L_{t_M}), \LCA_{\mT}(L_{t_1}, L_{t_1+1}, \dots, L_{t_M - 1}),  \dots, \LCA_{\mT}(L_{t_1}, L_{t_1+1}), \LCA_{\mT}(L_{t_1})\] has at most $r$ changes, and thus also holds for $Z_1, Z_2,\dots, Z_M$.  
\end{proof}

With the same notation of $I = [a, a+\ell]_K$ and the $t_i$'s as above. 
Let \(K\{a\}\) denote the \(i\)-th element of \(K\) (ordered by its value in \(\mathbb{N}\)). For now $a$ and $\ell$ are fixed. 

For any infinite set $A \subset \mathbb{N}$, define a set $A^{\leq b}$ in the following way. 
Let \(x\) be the minimum among the $b$ largest elements of $A$ that are strictly smaller than \(K\{a\}\). If there are less than $r$ such elements write $x= 0$ (Note that \(K\{a\}\) need not belong to $A$.) Analogously let $y$ be the maximum among the $b$ smallest elements of $A$ that are strictly larger than \(K\{a+\ell\}\). (Note that  again \(K\{a+\ell\}\) need not belong to $A$). 
Define the truncated set with respect to $a$ and $b$ as:
\[
A^{\leq b} := A \cap [x,y]_{\mathbb{N}} \subset A.
\]
Implicitly from now on, we are always fixing $I = [a, a+\ell]_K$. The definition $A^{\leq \ell}$ is also anchored with this $a$.

\begin{lem}\label{lem:PtZt}
    For any time $t_i$ for $1 \leq i \leq M$, 
    \[P_{t_i}\supset Z_i^{\leq 4\ell}. \]
    Furthermore, 
      \[ \left(P_{t_1} \cup  Z_i^{\leq 4\ell}\right)  \subset P_{t_i} \subset (P_{t_1} \cup Z_i),  \ \ \ \left(P_{t_1} \cup  Z_i^{\leq 4\ell}\right)  \subset \left(P_{t_i} \cup L_{t_i}\right)\subset (P_{t_1} \cup Z_i).\]
\end{lem}

We will need the following fact about tree. 
\begin{claim}\label{claim:tree}
Let $T = (V, E)$ be a rooted tree. Let $S = \{Q_1, Q_2, \dots, Q_s\} \subseteq V$ be a finite sequence of vertices with $s \geq2$, and let $Q$ be their least common ancestor denoted as $\LCA(S)$. For any arbitrary linear ordering of the elements of $S$, there exists at least one index $i \in \{1, \dots, s-1\}$ such that $\LCA(Q_i, Q_{i+1}) = Q$.
\end{claim}
\begin{proof}
If $Q \in S$, say $Q = Q_k$, then trivially $\LCA(Q_k, Q_{k+1}) = Q$ (or $\LCA(Q_{k-1}, Q_k) = Q$ if $k=s$).
Assume $Q \notin S$. Since $Q = \LCA(S)$, the elements of $S$ must reside in the subtrees of at least two distinct children of $Q$. Because the sequence $Q_1, \dots, Q_s$ contains elements from at least two such subtrees, there must be some index $i$ where $Q_i$ and $Q_{i+1}$ belong to subtrees rooted at different children of $Q$. For this adjacent pair, their paths to the root first intersect at $Q$, so $\LCA(Q_i, Q_{i+1}) = Q$.\end{proof}

\begin{proof}[Proof of Lemma \ref{lem:PtZt}]
Note that by definition, $Z_i$ is the lowest common ancestor of all $L_{i_t}$ for $t_1 \le t \le t_i$. By Claim~\ref{claim:tree}, there exists some time $t'$ with $t_1 \le t' < t_i$ such that the lowest common ancestor of $L_{i_{t'}}$ and $L_{i_{t'+1}}$ is $Z_i$. This implies that
\begin{equation}
    PB(t'+1) = Z_i. \label{eq:pbZi}
\end{equation}
Let $j$ be the largest index such that $t_j \le t'$. Then $w_j$ lies in $I$ and is revealed before time $t'+1$, and hence before $t_i$. 

By the update rule of the algorithm, the choice of the pod diameter parameter $\gamma_t$, and in particular the facts that $\gamma_t$ is monotone increasing in $t$ and $\gamma_t \ge 9\ell$ for all $t \ge t_1$, it follows that all elements of $Z_i$ within relative distance $9\ell$ from $w_j$ (with respect to $Z_i$) are added to the pod.

Moreover, $P_t$ is monotone in $t$, in the sense that elements are only added and never removed; that is, $P_t \subseteq P_{t'}$ for all $t \le t'$. Since $w_j \in I$, we conclude that $P_{t_i}$ contains all elements of $Z_i$ within relative distance $4\ell$ of $I$, i.e., $Z_i^{\le 4\ell}$.

We now prove the second claim. By a similar argument, using the update rule of the algorithm and the choice of the pod parameter $\gamma_t$, we obtain that
\begin{equation}
      \left( P_{t_1} \cup \bigcup_{t_1 \leq t \leq t_i} PB(t)^{\leq 4\ell}\right) \subset P_{t_i} \subset  \left(P_{t_1} \cup \bigcup_{t_1 \leq t \leq t_i} PB(t)\right). \label{eq:2}
   \end{equation}

  The first inclusion together with (\ref{eq:pbZi}) implies 
  \[  P_{t_i} \supset \left( P_{t_1} \cup \bigcup_{t_1 \leq t \leq t_i} PB(t)^{\leq 4\ell} \right) \supset \left(P_{t_1} \cup  PB(t'+1)^{\leq 4\ell} \right) = (P_{t_1} \cup  Z_i^{\leq 4\ell}).\]
   
   Note (\ref{eq:pbZi}) also implies \[\bigcup_{t_1 \leq t \leq t_i} PB(t) \supset Z_i.\] 
On the other hand, $Z_i$ is the lowest common ancestor of all identified languages $L_{i_t}$ for $t_1 \le t \le t_i$, and hence also the lowest common ancestor of all pullbacks $PB(t_1+1), \dots, PB(t_i)$. Therefore,
\[
\bigcup_{t_1 \le t \le t_i} PB(t) \subseteq Z_i.
\]
Combining this with the previous inclusion, we obtain
\[
\bigcup_{t_1 \le t \le t_i} PB(t) = Z_i.
\]
Thus, the second inclusion in~(\ref{eq:2}) implies that
$P_{t_i} \subseteq P_{t_1} \cup Z_i,$
as desired. Furthermore, since $L_{i_t} \subseteq PB(t_i) \subseteq Z_i$, we also have
\[
(P_{t_i} \cup L_{i_t}) \subseteq (P_{t_1} \cup Z_i).
\]
\end{proof}

\begin{lem}\label{lem:important}
    At time $t_i$, if $\left(I \cap (Z_i^{\leq \ell} \cup P_{t_1}^{\leq \ell}) \right) \setminus \{w_{t_i}\}$ has available elements, then $o_{t_i} \in (Z_i^{\leq \ell} \cup P_{t_1}^{\leq \ell})$. 
And 
    $I \cap (Z_i^{\leq \ell} \cup P_{t_1}^{\leq \ell}) = I \cap (P_{t_i} \cup L_{t_i})$.

Furthermore, suppose that the set $\left(I \cap (Z_i^{\leq \ell} \cup P_{t_1}^{\leq \ell}) \right) \setminus \{w_{t_i}\}$ contains at least one element that is available at time $t_i$. Then $o_{t_i}$ is defined to be 
\[
o_{t_i} \in \arg\min \left\{ d_{(Z_i^{\leq \ell} \cup P_{t_1}^{\leq \ell})}(x, w_{t_i}) \;\middle|\; x \in (Z_i^{\leq \ell} \cup P_{t_1}^{\leq \ell} )\setminus \{w_{t_i}\},\ x \text{ available at time } t_i \right\}.
\]
\end{lem}
\begin{proof}

By the second claim in Lemma \ref{lem:PtZt} that  \begin{equation}
    \left(P_{t_1} \cup  Z_i^{\leq 4\ell}\right)  \subset \left(P_{t_i} \cup L_{t_i}\right)\subset (P_{t_1} \cup Z_i), \label{eq:sandwich}
\end{equation} 
we have
\begin{equation}
   I \cap  \left(P_{t_1} \cup  Z_i^{\leq 4\ell}\right)  \subset  I \cap \left(P_{t_i} \cup L_{t_i}\right)\subset  I \cap (P_{t_1} \cup Z_i).\label{eq:sandwich2}
\end{equation} 
Notice that $I \cap P_{t_1} \subset P_{t_1}^{\leq 0}$ and similarly for $Z_i$, so  $I \cap (Z_i^{\leq \ell} \cup P_{t_1}^{\leq \ell}) = I \cap (P_{t_i} \cup L_{t_i})$. 

If 
$\left(I \cap (Z_i^{\leq \ell} \cup P_{t_1}^{\leq \ell}) \right) \setminus \{w_{t_i}\}$ contains at least one element $x$ that is available at time $t_i$, since both $x, w_{t_i} \in I$, then 
\[d_{P_{t_i} \cup L_{t_i}}(o_{t_i}, w_{t_i}) \leq d_{P_{t_i} \cup L_{t_i}}(x, w_{t_i}) \leq d_{I}(x, w_{t_i}) \leq  \ell. \]

Therefore, by our distance-preference algorithm, 
\begin{equation}
    o_{t_i} \in (P_{t_i}\cup L_{t_i})^{\leq \ell} \subset (P_{t_1}^{\leq \ell} \cup L_{t_i}^{\leq \ell}). \label{eq:4}
\end{equation}

If $o_{t_i} \in I$. Then by (\ref{eq:sandwich}), we have $o_{t_i} \in (P_{t_1} \cup Z_i)^{\leq 0} \subset (P_{t_1} \cup Z_i)^{\leq \ell}$ and the claim holds. 

From now on, we assume $o_{t_i} \notin I$.
As $P_{t_i} \cup L_{t_i} \subset P_{t_1} \cup Z_i$, suppose 
$o_{t_i} = y \in Z_i$. Suppose $\ell'$ is the smallest integer such that $y \in Z_i^{\leq \ell'}$. If $\ell' > \ell$, assume $y$ is to the left of $I$, then let $y'$ be the smallest integer in $Z_i^{\leq \ell}$ that is to the left of $I$. Clearly $y < y'$ as integers in $\mathbb{N}$. As $y < y'$ and both are to the left of $I$ and $w_{t_i} \in I \cap Z_i^{\leq \ell}$, for any subset $A$ that contains both $y, y', w_{t_i}$,  $d_A(y, w_{t_i}) > d_A(y', w_{t_i})$. So 
\[ d_{P_{t_i} \cup L_{t_i}}(y, w_{t_i}) > d_{P_{t_i} \cup L_{t_i}}(y', w_{t_i})
\geq d_{P_{t_1} \cup Z_i^{\leq \ell}}(y', w_{t_i}) \geq d_{ Z_i^{\leq \ell}}(y', w_{t_i})  = \ell.\] The second inequality is by (\ref{eq:sandwich}). This contradicts (\ref{eq:4}). The case when $y$ is to the right of $I$ is similar. Therefore we have shown that if $o_{t_i} \in Z_i$, then $o_{t_i} \in Z_i^{\leq \ell}$. 
The same argument also implies that if $o_{t_i} \notin Z_i$, $o_{t_i}$ should be larger than the smallest integer in $Z_i^{\leq \ell}$ and smaller than the largest integer in $Z_i^{\leq \ell}$. 

Now assume $o_{t_i} \notin I$ and $o_{t_i} \in P_{t_1}$. The we can show $o_{t_i} \in P_{t_1}^{\leq \ell}$ by the same argument as above with $Z_i$ replaced with $P_{t_1}$. 

Therefore we have shown that if $\left(I \cap (Z_i \cup P_{t_1})^{\leq \ell} \right) \setminus \{w_{t_i}\}$ contains at least one element $x$ that is available at time $t_i$, then 
\[ o_{t_i} \in Z_i^{\leq \ell} \cup P_{t_1}^{\leq \ell}.\] 

Furthermore, $o_{t_i}$ must be larger than both the smallest element in $Z_i^{\le \ell}$ and the smallest element in $P_{t_1}^{\le \ell}$, and smaller than both the largest element in $Z_i^{\le \ell}$ and the largest element in $P_{t_1}^{\le \ell}$. Any valid choice of $o_{t_i}$ must satisfy all of these conditions; we refer to this collection of requirements as Property~(*).

Recall that, by the definition of the algorithm, $o_{t_i}$ must be chosen to be
\[\arg\min \left\{ d_{(P_{t_i} \cup L_{t_i})}(x, w_{t_i}) \;\middle|\; x \in (P_{t_i} \cup L_{t_i}) \setminus \{w_{t_i}\},\ x \text{ available at time } t_i \right\}.\]
So far, we have shown that under the conditions,
\[
o_{t_i} \in \arg\min \left\{ d_{(P_{t_i} \cup L_{t_i})}(x, w_{t_i}) \;\middle|\; x \in (P_{t_1}^{\leq\ell} \cup L_{t_i}^{\leq \ell}) \setminus \{w_{t_i}\},\ x \text{ available at time } t_i \right\}.
\]
It remains to show that whenever $x$ satisfies Property~(*), then
$d_{P_{t_i} \cup L_{t_i}}(x, w_{t_i}) = d_{Z_i^{\leq \ell} \cup P_{t_1}^{\leq \ell}}(x, w_{t_i})$. As $x, w_{t_i}\in  (P_{t_i}\cup L_{t_i})^{\leq \ell}$ by (\ref{eq:4}),
for these $x$, 
\begin{equation}d_{P_{t_i}\cup L_{t_i}}(x, w_{t_i})  = d_{(P_{t_i}\cup L_{t_i})^{\leq \ell}}(x, w_{t_i}). \label{eq:temp1}
\end{equation} 

By Lemma \ref{lem:PtZt} the  statement $(P_{t_i}\cup L_{t_i}) \supset \left( P_{t_1} \cup  Z_i^{\leq \ell} \right) \supset \left( P_{t_1}^{\leq \ell} \cup  Z_i^{\leq \ell} \right) $, we have 
    \begin{equation}
    \ell \geq    d_{P_{t_i}\cup L_{t_i}}(x, w_{t_i}) \geq d_{P_{t_1}^{\leq \ell} \cup  Z_i^{\leq \ell}}(x, w_{t_i}). \label{eq:temp2}
    \end{equation}

By Lemma \ref{lem:PtZt} the  statement  
    $P_{t_i}\cup L_{t_i} \subset \left( P_{t_1} \cup Z_i \right),$
we have  
    \begin{equation}
      d_{P_{t_i}\cup L_{t_i}}(x, w_{t_i})  \leq d_{P_{t_1}\cup Z_i}(x, w_{t_i}) = d_{P_{t_1}^{\leq \ell} \cup  Z_i^{\leq \ell}}(x, w_{t_i}). \label{eq:temp3}
    \end{equation}
The last equality follows from Property~(*). Indeed, $d_{P_{t_1} \cup Z_i}(x, w_{t_i})$ counts the number of elements in $P_{t_1} \cup Z_i$ lying between $x$ and $w_{t_i}$. We already know this quantity is at most $\ell$.
By Property~(*), any path from $x$ to $w_{t_i}$ passes through at most $\ell$ elements of $P_{t_1}$ outside $I$, and at most $\ell$ elements of $Z_i$ outside $I$. Since these are  the only contributions to $d_{P_{t_1} \cup Z_i}(x, w_{t_i})$, the last claimed equality follows.

Combining (\ref{eq:temp1}), 
 (\ref{eq:temp2}) 
and  (\ref{eq:temp3}), we have 
$d_{P_{t_i}\cup L_{t_i}}(x, w_{t_i}) = d_{Z_i^{\leq \ell} \cup P_{t_1}^{\leq \ell}}(x, w_{t_i})$. Therefore the conclusion is reached. 

\end{proof}

After this setup, we abstract the analysis into a combinatorial game and prove a key lemma, which we later relate back to complete the proof of the theorem.

\subsubsection{A combinatorial game lemma}
\label{subsubsec:comb-game}

A key ingredient of our approach is the following standalone combinatorial lemma.

Throughout this subsubsection, we use notation that is independent of the rest of the paper.

\paragraph{Formal Game Description:}

We analyze a discrete adversarial game played on a finite interior board $I = [1, \ell]$ embedded within the extended board $\mathbb{Z}$. 
This game is a discrete, two-player sequential contest. Visibility of available spots is governed by a filtration of finite sets \[Y_1 \subset Y_2 \subset \dots \subset Y_r \subset \mathbb{Z}\]
with \[I \subset Y_r.\]

The Turns: The adversary moves first. In each turn $t$:
\begin{enumerate} 
\item Nature Help: At the beginning of each turn, Nature may place zero or more white stones on positions in $I$ that have not been occupied by the Algorithm, the Adversary, or Nature in previous rounds.

\item Adversary Move: The Adversary chooses an  integer $y_t \in I$ that has not been used by the adversary and put down a black stone. The adversary can choose to put down a stone at a spot where there is already a white stone. 

Let \[B_t = \{y_1, \dots, y_t\}\] be the set of Adversary stones up to turn $t$. (Note that the adversary could put down a stone at a spot already occupied by the algorithm, but the adversary cannot repeat itself). 

\item The Algorithm observes the set $Y^{(t)} = Y_{i_t}$, where $i_t = \min \{ i : B_t \subseteq Y_i \}$. 
\item Algorithm move: (Distance-based Greedy algorithm)

Whenever there exists an available position in $I \cap Y^{(t)}$, the Algorithm must choose $w_t$ from elements of $Y^{(t)}$ that are not reserved (see the next item) and have not been occupied by either a white or a black stone, so as to minimize the absolute relative distance $d_{Y^{(t)}}(w_t, y_t)$. If multiple such points exist, a fixed tie-breaking rule is applied. In other words,
 \[ y_t = \arg \min \{ d_{Y^{(t)}}(x, y_t) \mid x \in Y^{(t)} \setminus \{y_t\}, \ x \text{ available at time } t   \}. \]

If there is no available position in $I \cap Y^{(t)}$, the Algorithm skips its move. In this case, Nature may either place a white stone in $Y^{(t)} \setminus I$ or choose not to place any stone.

\item Nature Reserve: After the Algorithm’s move in round $t$, Nature may mark some unused elements in $Y^{(t)} \setminus I$ as \emph{reserved}; denote this set by $R_t$. Once an element is marked as reserved, it cannot be selected by the Algorithm in any future round. (Note that the Adversary never places stones outside $I$, so $R_t$ does not interfere with the Adversary.) Nature may also choose not to reserve any elements between rounds $t$ and $t+1$.

Let $B_t$ be the set of black stones and $W_t$ the set of white stones at the end of turn $t$. 
\end{enumerate}

The game terminates at turn $T$ when $B_T = I$.
Let  $W_T$ be the white stones in $I$ when the game terminates. The density we care about is: $\frac{|W_T|}{|I|}.$

It turns out that, for the purposes of the analysis below, Nature’s actions do not affect the outcome. Thus, Step~4 (the algorithm’s greedy decision rule) is the central component of the game. In particular, the reader may, for now, ignore these aspects and focus on the simpler version of the game.

\begin{lem}\label{lem:comb}
Regardless of the adversary’s moves, the filtration, or Nature’s actions, the adversarial ratio always satisfies:\[\frac{|W_T|}{|I|} \geq 1/2 - O_r\left(\frac{\log \ell}{\ell}\right).\] 
\end{lem}
\begin{proof}
It suffices to upper bound the total number of ``leaking moves" and the total number of skipped rounds of the algorithm. A leaking move is a specific instance where the Algorithm's greedy choice is forced to occupy an integer outside the Interior Board $I$, despite the Adversary's input being inside $I$. More formally, an algorithm's move is defined as {\it left-leaking} if the adversary input $a_t$ at time $t$ satisfies $a_t \notin W_T$ and the Algorithm's output $o_t$ is such that $o_t \leq 0$. Conversely, it is {\it right-leaking} if $o_t > \ell$ and $a_t \notin W_T$.  Leaking moves represent ``lost capacity" for the Algorithm. Since the game only ends when $I$ is full, every stone placed outside $I$ is a turn where the Algorithm failed to claim a spot in $I$, potentially allowing the Adversary to claim more than half of the interior board and thus increase $S$.

Note that $i_t = \min\{i: H_t \subset Y_i\}$ is monotone increasing in $t$. So we can split the whole game into at most $r$ phases, corresponding to when $i_t = 1, 2, \dots, r$. 
    \begin{claim}[Distance Monotonicity]\label{claim:1}
Fix a phase $P$ where $Y^{(t)} = Y$ is constant. Let $\{a_1, \dots, a_k\} \subseteq H_T \setminus W_T$ be the chronological subsequence of Adversary moves that triggered Algorithm moves $o_i \in Y \setminus I$ (leaking moves), and the time where the adversary did not place a stone at a spot originally occupied by the algorithm. For left-leaking moves ($o_i \leq 0$), the sequences satisfy $a_1 < a_2 < \dots < a_k$ and $o_1 > o_2 > \dots > o_k$.
\end{claim}

\begin{proof}[Proof of Claim]
Consider turns $t_i < t_j$ within phase $P$. At turn $t_i$, the Algorithm chose $o_i$ while $a_j$ was unoccupied (since $a_j$ is not played until $t_j$ as $a_j \in H_T \setminus W_T$ and $t_j > t_i$). Because $a_j \in I$ and $a_j \in Y$, it was available at $t_i$. The greedy choice implies $d_Y(o_i, a_i) \leq d_Y(a_j , a_i)$. Since $o_i \leq 0$ and $a_i, a_j \geq1$, if $ a_j < a_i$, it implies $o_i < a_j < a_i$. Thus  the relative distance $d_Y(a_i , a_j)< d_Y (a_i, o_i)$, a contradiction. Thus $a_j > a_i$. 

Similarly, when $j > i$, the string $o_j$ was unoccupied at $t_i$. As $a_i > 0, o_i, o_j \leq 0$, greedy choice implies $d_Y(a_i , o_i) < d_Y (o_j , a_i)$ which implies $o_j < o_i$. 
\end{proof}

\begin{claim}[Bound on Leaking Moves]\label{claim:leaking}
The number of left-leaking moves $k$ in any single phase $P$ is strictly bounded by $\log_2 \ell$. Specifically, $k \leq \log_2 \ell$.
\end{claim}

\begin{proof}
Let phase $P$ have a visibility set $Y$. Let $\{(a_i, o_i)\}_{i=1}^k$ be the sequence (in chronological order) of pairs where $a_i \in (I \cap Y)$ is the Adversary's move and $o_i \in (Y \setminus I)$ is the Algorithm's left-leaking response ($o_i \leq 0$). We retrospectively analyze the realized distances.

Consider any index $i < k$. At turn $t_i$, the Algorithm chose $o_i$. However, the integer $a_{i+1}$ (the Adversary's move at a future turn $t_{i+1}$) was necessarily unoccupied at turn $t_i$ as $a_{i+1} \notin W_T$. Since $a_{i+1} \in I \cap Y$, it was a valid candidate for the Algorithm. The greedy rule dictates that the distance to the chosen point $o_i$ must be less than or equal to the distance to any other available point in $Y$:
\begin{equation} d_Y(o_i, a_i) \leq d_Y(a_{i+1}, a_i). \label{eq:greedy_constraint} \end{equation}

Define $x_i = d_Y(a_i, o_i)$ as the relative distance in $Y$ between the $i$-th pair. We observe the growth of $x_{i+1}$ relative to $x_i$:
\[ x_{i+1} = d_Y(a_{i+1} , o_{i+1}). \]
Expanding this term using the previous coordinates and the monotonicity claim we just proved:
\[ x_{i+1} = d_Y(a_{i+1} , a_i) + d_Y(a_i , o_i) + d_Y(o_i, o_{i+1}). \]
From \eqref{eq:greedy_constraint}, we have $d_Y(a_{i+1}, a_i) \geq x_i$. Additionally, by the spatial monotonicity established in Claim \ref{claim:1}, $o_i > o_{i+1}$. Since these are distinct integers, $o_i - o_{i+1} \geq1$. Substituting these into the expansion:
\begin{equation} x_{i+1} \geq x_i + x_i + 1 = 2x_i + 1  \label{eq:1} \end{equation}

The base case $x_1 = d_Y( a_1 , o_1)$. Since $a_1 \in [1, \ell]$ and $o_1 \in [-\ell, 0]$, the smallest possible value for $x_1$ is $1 - 0 = 1$. Solving the recurrence (\ref{eq:1}), we obtain \[x_i \geq2^i - 1 \implies x_k \geq2^k - 1.\] 
Combining with (\ref{eq:greedy_constraint}), we have 
\[ d_Y(a_{i+1},  a_i) \geq x_i \geq 2^{i}-1. \]

Together with Claim \ref{claim:1}, we have 
\[ d_Y(a_k, a_1) = d_Y(a_k, a_{k-1}) + d_Y(a_{k-1}, a_{k-2})+ \dots + d_Y(a_2, a_1) \geq \sum_{i=1}^{k-1} (2^i-1) = 2^k - k-1. \]

However, the relative distance $d_Y(a_k,a_1)$ is physically bounded by the dimensions of the board $I = [a, a+\ell]$. Specifically, $d_Y(a_k, a_1) \leq \ell$. 
Equating the recurrence lower bound with the board's upper bound:
\[ 2^k - k-1 \leq \ell \implies  k \leq \log_2 \ell. \]
\end{proof}

Since there are at most $r$ phases, and in each phase there could be either left or right leaks, the total leaking moves \[W_{out} \leq 2r \log_2 \ell.\] 

We now bound the total number of skipped moves of the algorithm. An  {\it true skip} at turn $t$ is when $a_t \in B_T \notin W_T$ and the algorithm skips at turn $t$. 
\begin{claim}
The total number of true skips $K$ is at most $r$.
\end{claim}

\begin{proof}
The visibility index $i_t = \min \{ i : H_t \subseteq Y_i \}$ is monotonically non-decreasing because $H_t$ is a cumulative set. A phase $P_k$ is defined as a maximal interval of turns where $i_t = k$. 

A phase transition ($i_{t+1} > i_t$) is triggered if the Adversary plays $y_{t+1} \notin Y_{i_t}$. 
Now consider a skip at turn $t$. By definition, a skip occurs only if $I \cap Y_{i_t}$ is completely occupied. 
There are two scenarios. 

The first case is when the adversary does not occupy a spot that was already taken by a white stone. 
Since the game has not terminated, there must exist at least one unoccupied spot in $I$. Because $I \cap Y_{i_t}$ is full the unoccupied spots must be in $I \setminus Y_{i_t}$. As the adversary cannot repeat itself, the subsequent play by the adversary, which is $y_{t+1}$, must be in some unoccupied spot, and thus must be in $I \setminus Y_{i_t}$. 
When the Adversary plays this $y_{t+1}$ on the subsequent turn, the history $H_{t+1}$ is no longer contained in $Y_{i_t}$, forcing $i_{t+1} \geq i_t + 1$. 
Thus, in this case, the skip (and each full occupation of a visibility set) forces an immediate phase transition. Since there are only $r$ sets in the filtration, the total number of skips in this scenario is at most $r$. 

The second scenario is that the adversary occupies a spot that was already taken by a white stone. However in such a case the adversary input $a_{i_t} \in W_T$, and thus will not be counted towards a true skip. 

Therefore the total number of true skips is at most $r$. 
\end{proof}

Since $|B_T\setminus W_T| - |W_T|$ is at most the total number of leaks plus the total number of true skips, we have that 
\[ |B_T\setminus W_T| - |W_T| \leq 2r \log_2 \ell + r.  \]
Since $|B_T\setminus W_T| + |W_T|  = \ell$, we have $|W_T| \geq \frac{\ell - (2r \log_2 \ell + r)}{2}$.
Then:
\[ \frac{|W_T|}{|I|} \geq  \frac{\ell - (2r \log_2 \ell + r)}{2\ell}. \]
\end{proof}
It can also be seen from the proof that Nature Reserve and Nature’s interventions outside $I$ do not affect the bound on $\frac{|W_T|}{|I|}$, as long as the reservation process is monotone (i.e., once a position is reserved, it cannot be unreserved) and reservations occur only outside $I$. On the other hand, Nature Help can only improve the bound on $\frac{|W_T|}{|I|}$.

\subsubsection{Finish up proof of Lemma \ref{lem:ell}}
\label{subsubsec:finish-analysis-one-dim}

\begin{proof}[Completion of the Proof of Lemma \ref{lem:ell}]
Recall $I = [a, a+\ell]_K$. 

Recall the conclusion of Lemma \ref{lem:important}: 
    At time $t_i$, if $Z_i^{\leq \ell} \cup P_{t_1}^{\leq \ell} \setminus \{w_{t_i}\}$ has available elements, then $o_{t_i} \in Z_i^{\leq \ell} \cup P_{t_1}^{\leq \ell}$. 
Furthermore, suppose that the set $Z_i^{\leq \ell} \cup P_{t_1}^{\leq \ell} \setminus \{w_{t_i}\}$ contains at least one element that is available at time $t_i$. Then $o_{t_i}$ is defined to be an element
\[
o_{t_i} \in \arg\min \left\{ d_{Z_i^{\leq \ell} \cup P_{t_1}^{\leq \ell}}(x, w_{t_i}) \;\middle|\; x \in Z_i^{\leq \ell} \cup P_{t_1}^{\leq \ell} \setminus \{w_{t_i}\},\ x \text{ available at time } t_i \right\}.
\]

For each $1 \leq i \leq M$, 
define \[G_i = Z_i^{\leq \ell} \cup P_{t_1}^{\leq \ell}.\]
By Lemma \ref{lem:important} that $I \cap (Z_i^{\leq \ell} \cup P_{t_1}^{\leq \ell}) = I \cap (P_{t_i} \cup L_{t_i})$, and the fact that $w_{t_i} \in P_{t_i}\cap I$ implies that $I \subset G_M$. 

By Lemma~\ref{lem:rchange}, we have
\[
G_1 \subseteq G_2 \subseteq \cdots \subseteq G_M,
\]
and the sequence changes value at most $r$ times. Let $r' \le r$ be the number of distinct values in this chain, and denote them by
\[
\mG_1 \subsetneq \mG_2 \subsetneq \cdots \subsetneq \mG_{r'}.
\]
As shown above, we have $I \subseteq \mG_{r'}$.

Define a monotone non-decreasing function $f : [M] \to [r']$ such that
\[
G_i = \mG_{f(i)} \quad \text{for all } i \in [M].
\]

We apply Lemma~\ref{lem:comb} with $\mG_i$ playing the role of $Y_i$. The interval $I = [1,\ell]$ in Lemma~\ref{lem:comb} corresponds to our interval $I = [a, a+\ell]_K$. 

A black stone placed by the adversary in a turn of Lemma~\ref{lem:comb} corresponds to a string revealed by the adversary at some time step, while a white stone placed by the algorithm corresponds to a string generated by our algorithm at that time.

Note that $t_1 < t_2 < \cdots < t_M$ are precisely the time steps at which the adversary reveals a string in $I = [a, a+\ell]_K$. Prior to time $t_1$, no element in $I$ has been occupied by the adversary.

When the adversary inputs an element $w_{t_i} \in I$, by definition we have $w_{t_i} \in PB(t_i)$ and hence $w_{t_i} \in Z_i$. Since $w_{t_i} \in I$, it follows that $w_{t_i} \in G_i = \mG_{f(i)}$. Let $o_{t_i}$ denote the output of the algorithm at time $t_i$. By Lemma~\ref{lem:important}, if $I \cap G_i$ contains an available element, then
\[
o_{t_i} = \arg\min \left\{ d_{\mG_{f(i)}}(x, w_{t_i}) \;\middle|\; x \in \mG_{f(i)} \setminus \{w_{t_i}\},\ x \text{ available at time } t_i \right\}.
\]

We are ultimately counting the number of strings placed only by the adversary (black stones) and those generated by the algorithm (white stones) within $I = [a, a+\ell]_K$. If there is no available element in $I \cap G_i$ for the algorithm, then $o_{t_i}$ is placed outside $I$. This corresponds to Nature intervening in the game.

So far, there is a clear correspondence between our algorithm and the setup in Lemma~\ref{lem:comb}. However, there is one subtlety. In our analysis, we only control the adversary’s inputs within $I$, namely at times $t_1, t_2, \dots, t_M$, which are precisely the time steps at which the adversary reveals a string in $I = [a, a+\ell]_K$. We do not control the adversary’s behavior outside $I$, and more importantly, we do not control the behavior of the algorithm itself.
In particular, the algorithm may generate elements inside $I$ before time $t_1$ or between times $t_i$ and $t_{i+1}$. These actions are modeled in Lemma~\ref{lem:comb} via Nature Help and Nature Reserve. Specifically, the elements generated by the algorithm $\mA$ strictly between times $t_{i-1}$ and $t_i$ (with $t_0 := -\infty$) that lie in $I$ correspond to Nature Help in turn $i$. Similarly, the elements generated by $\mA$ and those used by the adversary strictly between times $t_i$ and $t_{i+1}$ that lie outside $I$ correspond to Nature Reserve in turn $i$.

Therefore, by Lemma~\ref{lem:comb}, the output density within $I = [a, a+\ell]_K$ is at least $1/2 - o_\ell(1)$, as desired. This completes the proof of Lemma~\ref{lem:ell}.
\end{proof}

Recall from the discussion following Lemma~\ref{lem:ell} that Lemma~\ref{lem:ell} implies Theorem~\ref{thm:finiterank}. This completes the proof of Theorem~\ref{thm:finiterank}.

\vspace{0.1in}
We say that an algorithm has speed $C$ if, at each time $t$, when the adversary inputs one string, the algorithm may generate up to $C$ strings.

An essentially identical argument yields the following corollary.

\begin{cor}
      Suppose the countable collection of languages $\mX$ has Cantor-Bendixson rank at most $r < \infty$. Then for any $f: \mathbb{N} \to \mathbb{N}$ that is increasing and not uniformly bounded above, there is an algorithm  with speed $C$ which generates in the limit whose lower $f$-window density in the true language is at least $C/(C+1)$.  
\end{cor}
Therefore, as $C$ increases, the guaranteed lower Banach density can approach $1$.

\section{Higher Dimensional Model: topology plus geometric discrepancy}

We now turn from intervals to boxes. The one-dimensional results show
that Banach density exposes a topological layer of the breadth problem.
In dimensions \(d\ge2\), a second layer appears: even when the target
language is known in advance, rectangular windows cannot always be
uniformly balanced. \textbf{Thus higher dimensions introduce a genuinely
geometric barrier, separate from the learning/topological one.}

We now formalize the higher-dimensional setting. Suppose the strings
under consideration are embedded as a countable set of points in
\(\mathbb R^d\). We can think of this $d$-dimensional space as 
corresponding to an embedding of the strings
in an underlying space, in the way that
language models standardly start from 
a $d$-dimensional embedding of the set of strings \cite{mikolov2013word2vec,pennington2014glove,pereira1993distributional}. 

Since our density notions count strings inside windows,
rather than measuring ambient volume, as discussed in Subsection \ref{subsec:introbanach}, the relevant structure of this
embedding is the coordinate-window structure: for each axis-parallel
rectangle, which strings lie inside it?

We work with the discrete version of this structure. In each coordinate,
assume that these coordinate value sets are discrete, this maps the
embedded strings to a subset of \(\mathbb Z^d\), and it preserves exactly
the incidence relation with axis-parallel windows. Hence the lower
Banach-density question over coordinate windows in the original embedding
is exactly the lower Banach-density question over axis-parallel
rectangles in \(\mathbb Z^d\).

\

We can again study the same 
 model of Kleinberg and Mullainathan \cite{kleinberg2024limit}.
As before, the adversary chooses a target language $K$ known only
to come from a countable set of candidates 
$\coll = \{\lang{1}, \lang{2}, \dots\}$, and
it enumerates the strings of $K$.
At time $t$, having seen the strings $\seen_t$ generated so far, 
the algorithm's goal is to generate a string $\out_t$ that comes 
from $\trueL - \seen_t$: that is, from 
true language $\trueL$ but not among the set $\seen_t$ already
enumerated. 
We say the algorithm generates in the
limit from $\trueL$ if after some finite time $\fint$, all of its
output strings $\out_t$ come from $K - \seen_t$.
As Theorem \ref{thm:acc} does not care about the underlying set of strings as long it is countable, it is also clear that it is always possible to generate in the limit. 

The question is how to measure the density (breadth) of the generation? 

When $d = 1$, the choice is constrained in the way we've seen so far in the paper: we are ordering strings  as $1,2,\dots = \mathbb{N}$, and computing the output density in some interval in the true language $K$. In other words, for any interval $I_\mathbb{N}$ in $\mathbb{N}$, we are looking at the density of output in $I_\mathbb{N} \cap K$. In other words, assume $O$ is the output at the end of the algorithm; we are then measuring
\[  \frac{|O \cap I_\mathbb{N} \cap K|}{|I_\mathbb{N} \cap K|}
\]
and the condition is always under $|I_\mathbb{N} \cap K| \to \infty$ (plus maybe some other conditions of the positions of $I$). 
In particular, the Banach lower density (defined in (\ref{def:banach})) precisely can be rewritten  as 
\[\delta_B(O, K) = \liminf_{|I_\mathbb{N} \cap K| \to \infty}   \frac{|O \cap I_\mathbb{N} \cap K|}{|I_\mathbb{N} \cap K|}.
\] 
and the asymptotic lower density can be 
(defined in (\ref{def:asymptotic})) precisely can be rewritten  as 
\[\underline{d}(O, K) = \liminf_{|I_\mathbb{N} \cap K| \to \infty, I \text{ starts at }1}   \frac{|O \cap I_\mathbb{N} \cap K|}{|I_\mathbb{N} \cap K|}.
\] 

Note that intervals in $\mathbb{N}$ are the one-dimensional analogue of rectangles. Thus, the densities above can equivalently be interpreted in terms of axis-parallel rectangles $R$, by considering
\[
\liminf_{|R \cap K| \to \infty} \frac{|O \cap K \cap R|}{|R \cap K|}.
\]
where $\liminf$ is among all axis-parallel rectangles $R$. 
Indeed, in 1-dimensional space, every rectangle is an interval $I_{\mathbb{N}}$.
This explains why windows are the right higher-dimensional analogue of
intervals. A window is obtained by choosing two coordinate cutoffs in
each dimension. A local failure of breadth is witnessed by a large
window \(R\) for which \(|R\cap K|\) is large but \(|R\cap K\cap O|\) is
small. After the order-index reduction, these witnesses are exactly
axis-parallel rectangles in \(\mathbb Z^d\).

\subsection{Ordinary Banach density detects a geometric barrier -- density 0}
This interpretation would allow us to generalize the density to higher dimensions. It is tempting to define the lower Banach density for higher dimensions as follows.
Consider the family of all axis-parallel rectangles $R$ in $\mathbb{Z}^d$, and compute the density as 
\[d(A,B) = \liminf_{|R \cap B| \to \infty}   \frac{|A \cap B \cap R|}{|B \cap R|}.
\] 

However, for our particular purposes, this definition turns out to be \textbf{overly restrictive}. In particular, we will show next that even when there is only a single target language, the quantity
$\liminf_{|R \cap (O \cap K)| \to \infty}   \frac{|O \cap K \cap R|}{|K \cap R|}$
may vanish regardless of the algorithm strategy. Here again, $O$ denotes the output of the algorithm, $K$ is the true language, and $R$ ranges over all axis-parallel rectangles in $\mathbb{Z}^d$ for $d\geq 2$.

\begin{thm}\label{prop:strictbanachHD}
  Suppose languages are infinite subset of $\mathbb{Z}^d$ where $d\geq 2$.  There is a language $K$ such that the adversary always has a way to enumerate the elements in the language such that 
   \[ \liminf_{|R \cap (O \cap K)| \to \infty, R \text{ axis-parallel rectangle}}   \frac{|O \cap K \cap R|}{|K \cap R|} = 0.\]
\end{thm}

This contrast highlights a fundamental limitation of the above definition of density. In one dimension, when there is a single true language, the algorithm can trivially guarantee the optimal fraction of correct output, i.e., $\liminf |K\cap O \cap R|/ |K \cap R| \geq 1/2$. However, this behavior does not persist in higher dimensions. For $d \geq 2$, even when the true language is fixed and known (so the algorithm faces no uncertainty whatsoever) there inevitably exist arbitrarily large axis-parallel rectangles on which the algorithm produces no output. Consequently, the density defined via \[\liminf_{|R \cap (O \cap K)| \to \infty, R \text{ axis-parallel rectangle}}   \frac{|O \cap K \cap R|}{|K \cap R|}\]
is equal to zero. This pathology shows that the definition is overly sensitive to sparse but structured regions of the domain,  and thus fails to capture the algorithm’s behavior in higher dimensions, motivating a revised notion of density that reflects global performance, which we will develop in Subsection \ref{subsec:otherBanach}.

\begin{proof}[Proof of Theorem \ref{prop:strictbanachHD}]
    We utilize a result by Chen, Pach, Szegedy, Tardos \cite{Chen2008DelaunayGO}.
    \begin{thm}[Chen, Pach, Szegedy, Tardos, \cite{Chen2008DelaunayGO}]\label{thm:CPST}
      Fix a positive integer $d\geq 2$. For any $c \in \mathbb{N}$, there is a finite set $P$ in $\mathbb{R}^d$  such that for any two coloring of the points in $P$, there always exist an axis-parallel rectangle with at least $c$ points from $P$ and these $c$ points are monochromatic.   
    \end{thm}
Notice that by slightly perturbing these points in $\mathbb{R}^d$
  without changing their relative order along any coordinate axis, we may assume that all points lie in $\mathbb{Z}^d$. We call such a set $P$ having a {\it monochromatic property with parameter $c$}.

A natural way to apply the above construction is to color elements of $O$ with color~1 and elements of $K \setminus O$ with color~2. At first glance, this appears plausible: due to the speed requirement, the set $K \setminus O$ is at least ``larger'' than $O$, suggesting that any monochromatic rectangle should consist entirely of points from $K \setminus O$. However, this intuition is not technically sound. It is also quite possible that a large rectangle contains only points from $O$, since the discrepancy guarantees apply only to a fixed, prescribed rectangle and provide no control over the relative distribution of points from $O$ and $K \setminus O$ across arbitrary rectangles. Consequently, an additional argument is required to ensure the existence of a rectangle consisting entirely of points from $K \setminus O$, while not ruling out the presence of large monochromatic rectangles composed solely of points from $O$.

 The trick to use is the following copying trick.

    We consider the following construction. Fix a strictly increasing sequence of natural numbers $c_1, c_2, \dots$. First placing two identical copies (after translation) of sets $P_1, P_1'$ with the monochromatic property with parameter $c_1$ and these two  copies are placed so far apart that any point in $P_1$ is smaller to any point in $P_1'$ in all the coordinates. Then placing two identical copies of sets $P_2, P_2'$ with Monochromatic Property with parameter $c_2$ and these two  copies are placed far apart between themselves as well as far apart from $P_1, P_1'$. So every point in $P_1 \cup P_1'$ is smaller than any point in $P_2 \cup P_2'$, and every point in $P_2$ is smaller than any point in $P_2'$. Continue this process. 

   We claim that the adversary has a strategy so that at the end,  in $P_k, P_k'$, there is always a rectangle with at least $c_k$ points in $P_k$ or $P_k'$ that is completely occupied by the adversary and not used by the algorithm. 

   The adversary strategy is as follows. We call it ``Copy strategy". We pair up elements in $P_k$ and $P_k'$ by the respective copies. The adversary tries to achieve the goal that no $x \in P_k$ and its copy $x' \in P_k'$ are both occupied by the algorithm. Suppose the algorithm outputs $o_t$ at time $t$ in $P_k$ for some $k$, then the adversary will reveal the copy of $o_t$ in $P_k'$ at time $t+1$; and similarly, if the algorithm outputs $o_t$ at time $t$ in $P_k'$ for some $k$, then the algorithm will reveal the copy of $o_t$ in $P_k$ at time $t+1$. 

There is a minor technical issue: the above algorithm appears to assume that the algorithm moves first. This assumption can be removed with a  modification.

The adversary may begin by repeatedly selecting elements from $P_1$, attempting to form a monochromatic rectangle, until the algorithm first outputs an element in $P_{k_1} \cup P_{k_1}'$ for some $k_1 > 1$. At that point, the adversary switches to a copy strategy within $P_{k_1} \cup P_{k_1}'$.

If the algorithm never reaches a new set $P_{k_1} \cup P_{k_1}'$ before the adversary does, then the adversary applies the following modified copy strategy within $P_1 \cup P_1'$.
The adversary first selects an element $w_1 \in P_1 \cup P_1'$. If the algorithm does not immediately output the corresponding copy of $w_1$ within $P_1 \cup P_1'$, the adversary proceeds by simulating a copy strategy against the algorithm. In particular, whenever the algorithm outputs an element, the adversary responds by outputting its paired copy, except that if the algorithm eventually outputs the copy of $w_1$, the adversary treats this as a reset and continues the same procedure.

At termination, since the algorithm never reaches a new set $P_{k_1} \cup P_{k_1}'$ before the adversary, both players will have exhausted all elements in $P_1 \cup P_1'$. By considering the first placement at each spot (since the adversary could use a string the algorithm already output), we obtain a configuration in which each element is assigned to exactly one of the two players. This induces a coloring of $P_1$ that is the opposite of the coloring of $P_1'$.
Consequently, for any monochromatic rectangle in $P_1$ of size $c_1$, its corresponding copy in $P_1'$ has the opposite color. Therefore, at least one of the two rectangles must be entirely occupied by the adversary.

At this point, the adversary may begin selecting elements in $P_1 \cup P_1'$ that have already been used by the algorithm. However, this forces the algorithm to move to a new set $P_k \cup P_k'$, when the adversary can in turn employ the copy strategy.

If the algorithm reaches a new set $P_k \cup P_k'$ before all elements in $P_1 \cup P_1'$ have been used by either player, then the adversary immediately switches to the copy strategy in $P_k \cup P_k'$, where the algorithm effectively moves first.

By repeatedly selecting strings that have already been used by the algorithm, the adversary can delay progress and wait until the algorithm initiates a new set $P_k \cup P_k'$.

If the adversary can consistently employ the copy strategy, then—given that it must eventually enumerate all strings—it may occasionally forgo applying the copy strategy on certain sets $P_k \cup P_k'$. During these steps, the adversary instead selects strings that were previously used only by the algorithm, thereby replenishing such elements.

In the end, there exists an infinite sequence $\{k_i\}$ such that, within each $P_{k_i} \cup P_{k_i}'$, every element in $O \cap (P_{k_i} \cup P_{k_i}')$ has its paired copy in $(P_{k_i} \cup P_{k_i}')$ lying in $K \setminus O$, and vice versa. This follows from the fact that the algorithm is not allowed to output any string previously revealed by the adversary.
Thus, $K$ is partitioned into $O$ and $K \setminus O$, which we view as colors $1$ and $2$, respectively. By Property $c_{k_i}$, it follows that in either $P_{k_i}$ or $P_{k_i}'$ there exists an axis-parallel rectangle containing at least $c_{k_i}$ points that lies entirely in $K \setminus O$.

   Therefore we have shown that there is a language $K$ such that the adversary always has a strategy such that 
   \[ \liminf_{|R \cap (O \cap K)| \to \infty, R \text{ axis-parallel rectangle}}   \frac{|O \cap K \cap R|}{|K \cap R|} = 0.\]
\end{proof}

Theorem~\ref{prop:strictbanachHD} should be interpreted as a statement
about the \emph{worst-case singleton benchmark} for ordinary Banach
density. By this benchmark, we mean the density guarantee that can be
promised uniformly over all singleton classes \(\mathcal X=\{K\}\), where
the target language is known from the start and there is no learning or
identification problem.

In one dimension, this benchmark is \(1/2\). If \(K\subseteq\mathbb Z\)
is known, the generator can alternate with the adversary along the order
of \(K\), and the adversary can prevent any universal guarantee above
\(1/2\). Thus \(1/2\) is the natural learning-free baseline in one
dimension.

In dimensions \(d\ge2\), Theorem~\ref{prop:strictbanachHD} shows  the corresponding worst-case singleton
benchmark collapses to zero. This does not mean that every known language
\(K\subseteq\mathbb Z^d\) has benchmark zero. Rather, it means that there
exist known languages \(K\) for which every generator can be forced to
have ordinary lower Banach density zero. The obstruction is therefore not
topological or learning-theoretic: it is a purely geometric obstruction
coming from rectangular Ramsey/discrepancy.

Thus ordinary Banach density in higher dimensions mixes two phenomena:
the learning/topological difficulty of identifying where to generate,
and a geometric difficulty that already appears when the target is known.
Filtered Banach density is introduced to factor out this worst-case
singleton geometric obstruction. After filtering out geometrically
degenerate boxes, the meaningful singleton benchmark returns to 
\(1/2\), and the main question becomes whether the
learning/topological obstruction can be controlled. Our finite-rank
theorem answers this positively, giving \(1/2-\varepsilon\) filtered
Banach density.

\subsection{Filtered Banach lower density}\label{subsec:otherBanach}

What is the optimal benchmark to compare the output density in higher dimension? 
In one dimension ($d=1$), this problem is trivial: the elements of $K$ possess a strict linear ordering, allowing for a perfectly alternating pattern between $O$ and $K \setminus O$ sequence. Under such an alternation, any interval that contains at least two points will inherently contain both points from $O$ and $K\setminus O$ (in fact, $O$ consists of almost a half of the interval). In particular, trivially in one-dimension, it completely precludes the existence of monochromatic intervals of size $c \geq2$. 

However, extending this logic to higher dimensions introduces profound structural rigidities. Because point set in $\mathbb{R}^d$ for $d \geq2$ exhibit distinct orderings along each coordinate axis, an arrangement that alternates between $K\setminus O$ and $O$ along one axis could  appear highly clustered and unstructured along another. This is why Theorem \ref{prop:strictbanachHD} shows that simply considering rectangles with more and more points in $K$ is too restrictive and inevitably leads to collapses of output density, even when there is only one language in the collection. 

This strong negative result demonstrates the impossibility of guaranteeing  mixing when the threshold depends solely on the discrete cardinality $|K \cap R|$ for axis-parallel rectangles $R$. It implies that purely combinatorial thresholds are fundamentally too restrictive, as they fail to account for how points are spatially distributed within the ambient space.

To bridge this gap, we must recognize that ignoring the spatial embedding of the points is precisely what renders the problem intractable. Once we accept that bounding the absolute point count $|P \cap R|$ is insufficient to enforce the mixing of $O$ versus $K \setminus O$ (or coloring in the discrepancy language), it becomes necessary to incorporate geometric information about $R$. The most basic such measure is its Lebesgue measure (volume), denoted $\Vol(R)$. Indeed, the colorability of a point set is governed not merely by cardinality, but by its spatial density within the enclosing rectangle. A rectangle that captures $c$ points by spanning a vast but sparse region of $\mathbb{R}^d$ exhibits fundamentally different behavior from one that captures $c$ points within a small, dense region. Consequently, replacing some density threshold relating $\Vol(R)$, or equivalently, with the side lengths of $R$, is unavoidable if we are to obtain meaningful positive results. By relating $|P \cap R|$ to some quantities incorporating $\Vol(R)$, we arrive at conditions that more accurately reflect both the combinatorial structure and the underlying geometry of the point set.

Similar to the spirit of the impossibility result Theorem \ref{prop:strictbanachHD} which utilizes Theorem \ref{thm:CPST}, which has connection to discrepancy theory. When there is only one language $K$ in the collection of languages, let $\chi: K \to \{+1, -1\}$, (considering $+1$ for the elements in $K \setminus O$, and $-1$ for $O$). For any set (such as a rectangle) $S$ in a set system $\mS$, let $\disc(\chi, S) = \sum_{x \in S} \chi(x)$ and $\disc(\chi, \mS) = \max_{S \in \mS} \disc(\chi, S)$.

The quantity that is most relevant for our density application is for any member $S$ in the set system (such as the set of axis-parallel rectangles), 
\[\disc(\chi, S)/|S \cap K|.
\] 
Such quantity is called size-sensitive discrepancy. We would like for the sets $S$ where $|S \cap K|$ is growing in some metric, the quantity $\disc(\chi, S)$ is $o(|S \cap K|)$.  Recall in the 1-dimensional case it is clear that the alternating coloring $\chi$ satisfies $\disc(\chi, I)= O(1)$ for all rectangles $I$ (as rectangles in $\mathbb{Z}$ are simply intervals), so $\disc(\chi, S)/|S \cap K| = o(1)$. 

When  dimension $d \geq 2$, because of the geometric tension between different dimensions, even the best coloring could have large discrepancy that does not simply depend on $|S \cap K|$. Most of the classical results in discrepancy theory has the discrepancy depending on $\text{poly}\log n$, where $n$ is the size of the full point set, and potentially be tremendously larger than $|S \cap K|$, but we would like $\disc(\chi, S)  \leq \epsilon |S \cap K|$.  The existence of the $\text{poly}\log (n)$ term in the best known discrepancy bounds is problematic for us, because our application is more similar to apply to an infinite set $K$ (or a larger and larger sets in $K$), and thus the extra ``$\log |K|$" term is detrimental to bound $\disc(\chi, S)/|S \cap K| = o(1)$. 

The result of Chen et al (Theorem \ref{thm:CPST}) can be strengthened to $c = \log^*(n)$, and the upper bound would be $\text{poly}\log n$ by the result of \cite{Ezra2016} or \cite{Matouek2014FactorizationNA}. Therefore a good guess for a better generalization of the lower Banach density in higher dimension would be requiring the rectangle $R$ to be non-degenerate with respect to $K$, i.e., $|R \cap K|$ to grow faster than $\text{poly}\log(\Vol(R))$. (Since $\Vol(R)$ is at most the $d$-th power of its largest side length and $d$ is considered as fixed  constant). 
\begin{defn}\label{def:filBanach}
Let $g:[0,\infty) \to [0, \infty)$ be nondecreasing with $g(x) \to \infty$. For $A \subset A' \subset \mathbb{Z}^d$, 
 define the {\it $g$-filtered lower Banach density}  by
    \begin{equation}
    \delta^g(A, A') = \liminf_{|S \cap A'|/ (g(\Vol(S))  \log \Vol(S) \to \infty}   \frac{|S \cap A \cap A'|}{|S \cap A'|}.\label{eq:banachdimd}
\end{equation}
where $S$ runs through all axis-parallel rectangles in $\mathbb{Z}^d$. 
\end{defn}

Again the requirement $|S \cap A'| \gg  \text{poly}\log \Vol(S)$ is to exclude the unavoidable irregular rectangle. The nondegeneracy condition excludes rectangles whose point count is too small relative to their geometric scale; Theorem \ref{prop:strictbanachHD} shows that such rectangles cannot be uniformly controlled even for singleton language classes.

Even with this density-filtered definition of lower Banach density, the collection of languages in Theorem \ref{thm:lowerdensity0f} can be easily generalized to the higher dimensional case such that all algorithms will still have lower Banach density zero.  
\begin{cor}
   Let $d \geq 2$.  There is a countable collection of languages $\mX$ with underlying set of strings being $\mU_1 \times \dots \times \mU_d$ where each $\mU_i$ is again a countable set of strings, such that for any algorithm that generates in the limit, there is a way for the adversary so that the filtered lower density of the output in $K$ with respect to the metric (\ref{eq:banachdimd}) is always zero. 
\end{cor}
\begin{proof}
In fact, we could even restrict our attention to squares (every side lengths are the same) and the same negative result still holds.

    Consider the construction in one-dimensional from Theorem \ref{thm:lowerdensity0f}. For each $L$ in that construction, extend that language to $\mathbb{N}^d$ by appending zeros to the second to the $d$-th dimension. Therefore for any axis-parallel square $B$ that intersects $K$, $|B \cap K| = \ell$ where $\ell$ is the side length of $B$. Therefore $|B \cap K| = \Vol(B)^{1/d} \gg \text{poly} \log \Vol(B)$, satisfying the condition in (\ref{eq:banachdimd}). Therefore the same conclusion in Theorem \ref{thm:lowerdensity0f} still applies here. 
\end{proof}

Again the construction above is about a countable  collection of languages with infinite Cantor-Bendixson rank. 

When the Cantor-Bendixson rank of the collection is finite,  we can prove a dichotomy using the adapted notion of lower density.

\begin{thm}\label{thm:banachhighdimRecs2}
 Let $f: [0, \infty) \to [0, \infty)$ be a non-decreasing function such that $\lim_{x \to \infty} f(x) = \infty$.
   Let $\mathcal{X}$ be a countable collection of subsets of $\mathcal{U}_1 \times \cdots \times \mathcal{U}_d$, where each $\mathcal{U}_i$ is a countable set of strings. Suppose that $\mathcal{X}$ has finite Cantor--Bendixson rank.
Then for every $\epsilon > 0$, there exists an algorithm that generates in the limit and the algorithm output $O$ in the true language $K$ has filtered Banach lower density at least $1/2 - \epsilon$, where the filtered Banach lower density is defined as
    \[\liminf_{|R \cap K| / \left( f(\Vol(R)) \log \Vol(R) \right)\to \infty}   \frac{|R \cap K \cap O|}{|R \cap K|}\]
where $R$ ranges over all axis-parallel rectangles in $\mathbb{Z}^d$.
\end{thm}

Since the set of boxes is a subset of axis-parallel rectangles, 
a corollary of the result is the following. 
\begin{cor}\label{thm:banachhighdim}
Let $f: [0, \infty) \to [0, \infty)$ be a non-decreasing function such that $\lim_{x \to \infty} f(x) = \infty$. 
   Let $\mathcal{X}$ be a countable collection of subsets of $\mathcal{U}_1 \times \cdots \times \mathcal{U}_d$, where each $\mathcal{U}_i$ is a countable set of strings. Suppose that $\mathcal{X}$ has finite Cantor--Bendixson rank. Then for every $\epsilon > 0$, there exists an algorithm that generates in the limit and the algorithm output $O$  in the true language $K$ has filtered lower Banach density least $1/2 - \epsilon$. Here the (filtered) Banach lower box density with respect to axis-parallel boxes is defined as
\[
\liminf_{\substack{x \in \mathbb{Z}^d,\; r \to \infty \\
\frac{|\mathbb{B}(x,r)\cap K|}{f(r^d)\,\log r^d} \to \infty}}
\frac{|\mathbb{B}(x,r)\cap K \cap O|}{|\mathbb{B}(x,r)\cap K|},
\]
where $\mathbb{B}(x,r) \subset \mathbb{Z}^d$ denotes the axis-parallel box centered at $x$ with side length $r$. 
\end{cor}

Theorem~\ref{thm:banachhighdim} does not imply the one-dimensional Theorem~\ref{thm:finiterank}. Indeed, condition~\eqref{eq:banachdimd} only considers boxes $\mathbb{B}(x,r)$ satisfying
$\frac{|\mathbb{B}(x,r)\cap K|}{\mathrm{poly}\log r} \to \infty,$
whereas Theorem~\ref{thm:finiterank} merely requires that $|\mathbb{B}(x,r)\cap K| \to \infty$, allowing $K$ to be arbitrarily sparse within $\mathbb{B}(x,r)$. In this sense, Theorem~\ref{thm:finiterank} appears formally stronger.
On the other hand, due to inherent limitations in higher dimensions, it seems necessary to impose a growth condition of the form
\[
\frac{|\mathbb{B}(x,r)\cap K|}{\mathrm{poly}\log r} \to \infty.
\]
Moreover, the techniques underlying the proofs of the two theorems are substantially different, and neither proof appears to extend readily to the other setting.
We discuss in Subsection~\ref{subsec:firstattempt} that a direct generalization of the one-dimensional techniques yields a significantly weaker version of Theorem~\ref{thm:banachhighdim}.

In fact, the result extends to more general axis-parallel rectangles whose side lengths have fixed ratios. 
Fix an axis-parallel rectangle $R \subset \mathbb{Z}^d$. For $x \in \mathbb{Z}^d$ and $r > 0$, let $R(x,r)$ denote the translate of $R$ by $x$ and scaled by a factor of $r$. The proof is identical to that for boxes.

We define
\begin{equation}
    \delta_{R}(A,B)
    =
    \liminf_{\substack{x \in \mathbb{Z}^d,\; r \to \infty \\
    \frac{|R(x,r)\cap B|}{\log^2 r} \to \infty}}
    \frac{|R(x,r)\cap B \cap A|}{|R(x,r)\cap B|}.
    \label{eq:banachdimd2}
\end{equation}

\begin{cor}\label{cor:banachhighdim}
Let $R \subset \mathbb{Z}^d$ be a fixed axis-parallel rectangle. Let $\mathcal{X}$ be a countable collection of subsets of $\mathcal{U}_1 \times \cdots \times \mathcal{U}_d$, where each $\mathcal{U}_i$ is a countable set of strings, and suppose that $\mathcal{X}$ has finite Cantor--Bendixson rank. Then for every $\epsilon > 0$, there exists an algorithm that identifies any target set $K \in \mathcal{X}$ in the limit and outputs a set $O$ such that the filtered lower Banach density with respect to $R$ satisfies
$\delta_R(O,K) \;\ge\; \frac{1}{2} - \epsilon.$

\end{cor}

\subsection{Filtered Banach density in higher dimensions}
The proof requires some new ingredients as well as some heavy lifting from the one-dimensional case. 

\subsubsection{Some results and first attempts}\label{subsec:firstattempt}
We begin with a simple observation. Suppose we define the asymptotic lower density by restricting to boxes (i.e., axis-parallel rectangles with equal side lengths) anchored at the origin $(0,\dots,0)$. Then, analogous to our optimal one-dimensional asymptotic lower density bound from \cite{kleinberg2026density}, we again obtain the optimal density in this setting.
\begin{thm}\label{thm:asympHD}
    For any countable collection of languages $\mX$ with underlying set of strings being $\mathbb{Z}^d$. Then there exists an algorithm that generates in the limit so that the asymptotic lower density  of the algorithm output in the true language is always $1/2$.
\end{thm}
For this asymptotic lower density notion, we do not require the boxes to be degenerate, i.e, we do not need to require 
$\frac{|\mathbb{B}(0,r)\cap B|}{\mathrm{poly}\log r} \to \infty,$
since the discrepancy over boxes $\mathbb{B}(0,r)$ anchored at the origin is $o\big(|\mathbb{B}(0,r)\cap B|\big)$ as $r \to \infty$ (for example similar to chessboard coloring). Consequently, the upper bound of $1/2$ is always achievable.

The matching lower bound follows from our one-dimensional lower asymptotic density theorem of \cite{kleinberg2026density}. Optimality already follows from the usual singleton speed-one benchmark. We also explain why this approach does not extend to higher dimensions in the Banach density setting.
\begin{proof}[Proof of Theorem \ref{thm:asympHD}]
We first create a bijective mapping $\rho: \mathbb{Z}^d \to \mathbb{N}$. Let $S_k = \{x \in \mathbb{Z}^d \mid \|x\|_\infty = k\}$.  Thus 
\[| S_0 \cup S_1 \cup \dots \cup S_{k-1}| = (2k-1)^d.
\] 
Let $N(x)$ be the number of points $z$ on the same shell $S_k$ that come strictly before $x$ in lexicographical order:\[N(x) = |\{ z \in \mathbb{Z}^d \mid \|z\|_\infty = k \text{ and } z \prec x \}|.\]
Then define \[\rho(x) = (2k-1)^d + N(x).\]
Clearly $\rho(x)$ is bijective, and most importantly, any box $\mathbb{B}(0,r)$ is mapped by $\rho$ to a consecutive interval  in $\mathbb{N}$ starting at $0$. Therefore, the algorithm strategy should be first transforming $\mathbb{Z}^d$ to $\mathbb{N}$ by $\rho$, and follow our one-dimensional algorithm from \cite{kleinberg2026density}. Then the lower asymptotic density in the $d$-dimensional case is the same as the lower asymptotic density in the one-dimensional case, which is $1/2$.    
\end{proof}

When working with Banach lower density, we restrict attention, for simplicity, to axis-parallel boxes $\mathbb{B}(x,r)$, i.e., rectangles whose side lengths are all equal to $r$. A natural first attempt is to mimic the proof of Theorem~\ref{thm:finiterank}, the one-dimensional case.
The topological component carries over unchanged, as the dimension does not affect the underlying topology. For the algorithmic part, an attempt would be to  similarly follow Theorem~\ref{thm:finiterank} by outputting, at each step $t$, an element $o_t \in P_t$ that is closest (in Euclidean distance) to the adversarial input $w_t$.
However, this approach encounters a fundamental obstacle: it leads to too many leaking elements, as characterized in Claim~\ref{claim:leaking}.

\begin{example}\label{example:leaking}
To illustrate the issue, consider the case $d=2$. Suppose $\mathcal{X}$ contains a single language $K$, and that $K$ contains the points $(a, a+2\mathbb{N})$ and $(a-1, a+2\mathbb{N})$, but no points with $x$-coordinate in $\{a+3, a+4, \dots, a+\ell+1\}$.

Suppose the adversary first reveals $(a,0)$. The closest available element is then $(a-1,0)$, which the algorithm outputs. Next, the adversary reveals $(a,2)$, and the algorithm will output $(a-1,2)$. Continuing in this manner, the adversary occupies all points in $(a,2\mathbb{N})$ before the algorithm does, while the algorithm outputs all points in $(a-1,2\mathbb{N})$.

Consequently, if a box of side length $\ell$ has its left boundary aligned with the line $x=a$, then all points of $K$ inside this box lie in $K \setminus O$. Therefore, the algorithm's output is extremely sparse within this box.
\end{example} 

Therefore, already in dimension $2$, handling such boxes requires that a box of side length $\ell$ contains significantly more than $\ell$ points from $K$. More generally, if one follows the same algorithm as in the one-dimensional case, one is led to require that boxes $\mathbb{B}(x,\ell)$ satisfy
\[
|\mathbb{B}(x,\ell)\cap K| \gg \ell^{d-1}\log \ell,
\]
which is a substantially stronger condition than what is required in Theorem~\ref{thm:banachhighdimRecs2}. However, such a requirement appears necessary in order to extend the proof of Theorem~\ref{thm:finiterank} to higher dimensions using the same approach.

The failure of the proof of Theorem~\ref{thm:finiterank} under weaker conditions appears to stem from the fact that Euclidean distance is not the appropriate metric for prioritizing outputs. In the example above, in order to avoid excessive leakage, the algorithm would do better if it prioritizes elements lying in the same row or column as the adversarial input $w_t$, rather than relying solely on Euclidean proximity.
However, the construction in Example~\ref{example:leaking} is not rigid: a slight perturbation of the points ensures that the adversarial input no longer lies exactly on the same row or column. In such cases, it is unclear how to prioritize nearby elements—for instance, whether to favor points to the left or to the right—since either choice may lead to poor behavior depending on the placement of the box.

A natural approach is to design a more suitable metric. This idea resonates with the construction in Theorem~\ref{thm:asympHD}, where one defines a distance of the form $d(x,y) = |\rho(x)-\rho(y)|$ for a bijection $\rho : \mathbb{Z}^d \to \mathbb{N}$. However, in the present setting, where all translations of boxes are allowed, it appears that no choice of $\rho$ can circumvent the need for a density condition of the form
\[
|\mathbb{B}(x,\ell)\cap K| \gg (\log \ell)\,\ell^{d-1}
\]
in order to guarantee an asymptotic output density of $1/2$.

Thus, overcoming this polynomial threshold requires fundamentally new ideas.

\

How can we handle Example~\ref{example:leaking} without explicitly designing a metric? One approach is to use a pairing technique. Specifically, we pair the elements $(a,6k)$ with $(a-1,6k)$ for $k \in \mathbb{Z}$, pair $(a,6k+2)$ with $(a,6k+4)$, and pair $(a-1,6k+2)$ with $(a-1,6k+4)$ for all $k \in \mathbb{Z}$.
The algorithm then proceeds as follows: whenever the adversary reveals an input $w_t$ that has not been previously selected by the algorithm, the algorithm outputs the other element in the pair containing $w_t$.
With this strategy, for any sufficiently tall and thin axis-parallel rectangle, the output density of the algorithm is at least $1/2$ in the limit in Example~\ref{example:leaking}. Intuitively, this pairing mechanism avoids the pathological behavior in Example~\ref{example:leaking} without relying on a global metric.

To extend this idea, suppose first that there is a single language $L \subset \mathbb{Z}^d$. We aim to generalize the pairing approach by assigning to each element of $L$ a value in $\{-1,+1\}$ so that every axis-parallel rectangle contains nearly balanced assignments of $+1$ and $-1$.

More precisely, we seek a coloring $\chi : L \to \{-1,+1\}$ such that for every axis-parallel rectangle $S \subset \mathbb{Z}^d$,
\[
\disc(\chi,S) = \left| \sum_{x \in S} \chi(x) \right|  \;\le\; \epsilon\,|S \cap L|.
\]
If such a coloring exists, then one can achieve high output density by applying a generalized pairing strategy as described above.

We then turn to the more general setting where there are infinitely many languages in $\mX$ and  the identified language varies over time. This requires combining the above coloring ideas with techniques from the one-dimensional case, as the learning process must accommodate a continually changing identified language. This necessitates a different algorithmic construction, and carrying out this step requires additional care.

In the next subsection, we develop the missing ingredient: a ``size-sensitive discrepancy lemma". We also discuss several barriers to establishing such results.
 
\subsubsection{A size-sensitive discrepancy lemma for infinite point sets}

The main new ingredient for the higher-dimensional positive theorem is a
size-sensitive discrepancy lemma for arbitrary, possibly infinite, point
sets \(P\subseteq\mathbb Z^d\). Classical discrepancy bounds for finite
point sets usually depend on a global ambient parameter \(n\). Such
bounds are not suitable here, since the target language is infinite and
the density guarantee must hold locally, over boxes of many sizes and
locations. We need a coloring whose error in any not-so-sparse box $R$ is negligible compared
to \(|P\cap R|\) whenever the box is nondegenerate.

Classical discrepancy theory for a finite ground set of size $n$ is primarily concerned with minimizing the absolute global discrepancy, typically yielding bounds of the form $O(\log^k n)$. However, such bounds are mostly negligible at the largest scale and can dominate the signal on smaller subsets. For our purposes, bounds expressed in terms of $n$ are insufficient, as we require the local discrepancy to be independent of $n$, since in our case $n$ can be considered as $|K|$, which is infinite. In contrast, our objective is relative discrepancy: for every sufficiently large rectangle $S$, we seek colorings satisfying \[|\chi(S \cap P)| \leq \varepsilon |S \cap P|,\]
so that the error is always small compared to the main term $|S \cap P|$, uniformly across scales and locations. Existing discrepancy techniques, based on VC dimension, Beck–Fiala–type arguments, or $\gamma_2$ factorizations, are inherently global: they optimize worst-case error over all ranges simultaneously and therefore accumulate discrepancy proportional to an ambient parameter $n$, making them unsuitable for enforcing relative guarantees at every scale (as $n$ in our case will essentially be $|L|$, which is infinite). 
In the infinite setting, this limitation becomes fundamental, since translation invariance (as we work with Banach lower density) rules out any uniform polylogarithmic bound that applies simultaneously to all rectangles.

To address this, we introduce a localized framework based on a density filtration, restricting attention to rectangles whose point count is not too sparse relative to their geometric scale. Within this regime, we obtain desired discrepancy bounds that depend only on local density and geometry, not on any global parameter.

\begin{thm}[Multi-level Size-Sensitive Discrepancy]\label{thm:multi-discrepancy}
    Let $P \subseteq \mathbb{Z}^d$ be an arbitrary set of points (possibly infinite). Fix $\epsilon \in (0, 1)$ and $\alpha > 0$. Let the baseline density constant $c_0$ be \[c_0 = \left( \frac{2\pi^2}{3 \epsilon^2} \right)^{\frac{\alpha}{2d}}.\]
    For any axis-parallel rectangle $R \subset \mathbb{R}^d$ such that $|R \cap P| = m \geq1$, we define its sparsity tier $j \in \mathbb{N}_{\geq 0}$ as the unique non-negative integer satisfying:
    \[\frac{c_0}{2^j} L(R)^\alpha \leq m < \frac{c_0}{2^{j-1}} L(R)^\alpha\]
    where $L(R)$ denotes the maximum geometric side length of the rectangle $R$.
    Then there exists a global 2-coloring $\chi: P \to \{-1, +1\}$ such that for every axis-parallel rectangle $R$, the discrepancy satisfies:
    \[|\chi(R \cap P)| \leq \epsilon m + \beta \sqrt{m j} + \gamma(m) \sqrt{m}\]
    where the constants depend only on the dimension $d$ and the filtration exponent $\alpha$, explictly: $\beta = \sqrt{ 2 \left( \frac{2d}{\alpha} + 1 \right) \ln 2 }$ and $\gamma(m) = \sqrt{ 2 \left( \frac{2d}{\alpha} + 2 \right) \ln m + 4 \ln 2 }$.
\end{thm}
\begin{proof}
    Let $V_N = [-N, N]^d \cap \mathbb{Z}^d$ be a finite bounding window, and define the finite point set $P_N = P \cap V_N$. Let $\mathcal{R}_N$ denote the family of all axis-parallel rectangles $R$ such that $R \cap P_N \neq \emptyset$.
    
    Consider a rectangle $R' \in \mathcal{R}_N$ containing exactly $k$ points from $P_N$ and belonging to tier $j$. By the strict lower bound of its sparsity tier definition, its maximum geometric side length is bounded by:
    \begin{equation}
        L(R') \leq \left( \frac{k \cdot 2^j}{c_0} \right)^{1/\alpha}. \label{eq:L}
    \end{equation}

Fix $x$ a target point $x = (x_1, \dots, x_d) \in P_N$. Fix non-negative integers $k, j$. We now upper bound $N_x(k,j)$, defined as the maximum number of distinct tier-$j$, size-$k$ rectangles in $\mathcal{R}_N$ that contain the point $x$. Let $R'$ be such a rectangle. By (\ref{eq:L}), its maximum geometric side length $L(R')$ is bounded above by $L := \left\lfloor \left( \frac{k \cdot 2^j}{c_0} \right)^{1/\alpha} \right\rfloor$.  

Consider the geometric projection of $R'$ onto the $i$-th coordinate axis. This projection forms a 1-dimensional integer interval $[a, b]$ of some length $l \in \{1, 2, \dots, L\}$ that strictly contains the coordinate $x_i$. For any fixed length $l$, the left endpoint $a$ must satisfy $x_i - l + 1 \leq a \leq x_i$. Consequently, there are exactly $l$ valid choices for the interval's position along this axis.
To find the total number of valid 1-dimensional intervals containing $x_i$, we sum these choices over all possible interval lengths up to $L$: \[\sum_{l=1}^{L} l = \frac{L(L+1)}{2} \leq L^2.\]
Because an axis-parallel rectangle in $\mathbb{R}^d$ is uniquely determined by the Cartesian product of its $d$ one-dimensional orthogonal projections, the total number of distinct $d$-dimensional rectangles containing $x$ is bounded by the product of the 1-dimensional choices, which is 
\[N_x(k, j) \leq \left( L^2 \right)^d = L^{2d}.\]
Substituting our upper bound for $L$ yields:
\[N_x(k, j) \leq \left( \left( \frac{k \cdot 2^j}{c_0} \right)^{1/\alpha} \right)^{2d} = c_0^{-2d/\alpha} 2^{2dj/\alpha} k^{2d/\alpha}.\]
Letting $\eta = c_0^{-2d/\alpha}$, we obtain the final point-wise capacity bound:
\[N_x(k, j) \leq \eta \cdot 2^{2dj/\alpha} k^{2d/\alpha}.\]

Now, fix a target rectangle $R \in \mathcal{R}_N$ containing exactly $m$ points from $P_N$, belonging to an arbitrary tier $i$. A rectangle $R'$ depends on $R$ if and only if they share at least one point in $P_N$. By taking the union bound over the exactly $m$ points in $R$ (each acting as $x$ earlier), the total number of tier-$j$, size-$k$ rectangles intersecting $R$ satisfies:
\begin{equation}
    N(m \to k, j) \leq \sum_{x \in R \cap P_N} N_x(k, j) \leq m \eta \cdot 2^{2dj/\alpha} k^{2d/\alpha}. \label{eq:Nmkj}
\end{equation}
    
   We define a probability space by assigning independent, uniformly random  variables $\chi_N(x) \in \{-1, +1\}$ to each point in $P_N$, with equal probability to be either $1$ or $-1$. Let $A_{m,i}$ denote the ``bad event" that a rectangle of size $m$ and tier $i$ violates the target discrepancy threshold $\lambda_{m,i}$:
   \[A_{m,i} = \left\{ \left| \sum_{x \in R \cap P_N} \chi_N(x) \right| > \lambda_{m,i} \right\}.\]
   By Chernoff bound, the probability of this event is bounded by:
   \begin{equation}
       \Pr(A_{m,i}) \leq 2 \exp\left( - \frac{\lambda_{m,i}^2}{2m} \right). \label{eq:badprob}
   \end{equation}

We apply the asymmetric version of Lov\'asz Local Lemma (Theorem \ref{thm:LLL}), whose full statement is in appendix. We assign a deterministic weight $x_{k,j} \in (0, 1/2]$ to every possible rectangle of size $k \geq1$ and tier $j \geq0$. We define the double-indexed polynomial-exponential weight function:
\begin{equation}
  x_{k,j} = \frac{1}{2} k^{-p} 2^{-qj},  
\end{equation}
where $p = \frac{2d}{\alpha} + 2$ and $q = \frac{2d}{\alpha} + 1$.
The Lov\'asz Local Lemma guarantees the existence of an assignment avoiding all bad events $A_{m,i}$ if for all target sizes $m \geq1$ and all target tiers $i \geq0$: 
\begin{equation}
    \Pr(A_{m,i}) \leq x_{m,i} \prod_{k=1}^\infty \prod_{j=0}^\infty (1 - x_{k,j})^{N(m \to k, j)}. \label{eq:LLLcondition2}
\end{equation}

Because $x_{k,j} \leq 1/2$ for all $k \geq1$ and $j \geq0$, we have $1 - z \geq\exp(-2z)$. We have
\[
\text{RHS of } (\ref{eq:LLLcondition2}) \geq  x_{m,j}\exp\left( -2  \sum_{k=1}^\infty \sum_{j=0}^\infty x_{k,j} N(m \to k, j)  \right).
\]  
We explicitly evaluate the infinite double sum in the exponent, by (\ref{eq:Nmkj}) we have 
\[2\sum_{k=1}^\infty \sum_{j=0}^\infty x_{k,j} N(m \to k, j) \leq 2 \sum_{k=1}^\infty \sum_{j=0}^\infty \left( \frac{1}{2} k^{-p} 2^{-qj} \right) \left( m \eta 2^{2dj/\alpha} k^{2d/\alpha} \right).\]
Factoring out the terms independent of the summation indices, we have 
\begin{align} \text{RHS of }(\ref{eq:LLLcondition2}) \geq & x_{m,j}\exp\left(-m \eta \left( \sum_{k=1}^\infty k^{\frac{2d}{\alpha} - p} \right) \left( \sum_{j=0}^\infty 2^{\frac{2d}{\alpha} j - qj} \right)\right)\\
\geq & x_{m,j}\exp\left( -m \eta \left( \sum_{k=1}^\infty k^{-2} \right) \left( \sum_{j=0}^\infty 2^{-j} \right)    \right) = x_{m,j}\exp\left( -m \eta \left( \frac{\pi^2}{3}   \right) \right).
\end{align}
Let $S = \eta \frac{\pi^2}{6}$. The Lov\'asz Local Lemma condition (\ref{eq:LLLcondition2}) is satisfied if 
\[\Pr(A_{m,i}) \leq 2 \exp\left( - \frac{\lambda_{m,i}^2}{2m} \right) \leq \frac{1}{2} m^{-p} 2^{-qi} \exp( - 2m S ).\]

Next is routine computation to make sure the desired inequalities hold. Taking the natural logarithm of both sides, we need $\ln 2 - \frac{\lambda_{m,i}^2}{2m} \leq -\ln 2 - p \ln m - qi \ln 2 - 2m S.$  This is equivalent to
\[\lambda_{m,i}^2 \geq4S m^2 + (2 q \ln 2) m i + (2 p \ln m + 4 \ln 2) m.\] 
We use the inequality $(\sqrt{a} + \sqrt{b} + \sqrt{c})^2  \geq a+b+c$ for $a,b,c \geq 0$. Thus, to make the lower bound of $\lambda_{m,i}^2$ hold, it is  sufficient to set $\lambda_{m,i}$ as the sum of the square roots of the three linear coefficients, which is \[\lambda_{m,i} = \sqrt{4S} m + \sqrt{2 q \ln 2} \sqrt{m i} + \sqrt{2 p \ln m + 4 \ln 2} \sqrt{m}.\]
Recalling $\eta = c_0^{-2d/\alpha}$, we have $\sqrt{4S} = \sqrt{ 4 \left( c_0^{-2d/\alpha} \frac{\pi^2}{6} \right) } = \sqrt{ \frac{2\pi^2}{3} c_0^{-2d/\alpha} } = \epsilon$ by our choice of $c_0$. Substituting $q = \frac{2d}{\alpha} + 1$, the coefficient $\sqrt{2 q \ln 2}$ exactly matches the definition of $\beta$. Substituting $p = \frac{2d}{\alpha} + 2$, the coefficient $\sqrt{2 p \ln m + 4 \ln 2}$ exactly matches the definition of $\gamma(m)$. 
Because this choice of $\lambda_{m,i}$ guarantees that (\ref{eq:LLLcondition2}) is satisfied for every rectangle $R \in \mathcal{R}_N$. Thus by Lov\'asz Local Lemma there exists at least one valid coloring $\chi_N : P_N \to \{-1, +1\}$ that avoids all the bad events simultaneously on the finite grid.

We extend this result to the case when $P$ is  infinite in $\mathbb{Z}^d$. Let $\mathcal{X} = \{-1, +1\}^P$ be the space of all possible global colorings of the point set $P$. We equip $\{-1, +1\}$ with the discrete topology, making it a compact Hausdorff space. By Tychonoff's Theorem, the product space $\mathcal{X}$ is compact under the product topology. For each integer $N \geq1$, let $\chi_N$ be the valid coloring on $V_N \cap P$ established earlier in the finite case. We extend $\chi_N$ to a global coloring $\tilde{\chi}_N \in \mathcal{X}$ by assigning $+1$ to all points in $P \setminus P_N$. Because $\mathcal{X}$ is compact (and indeed sequentially compact, as it is a countable product of metrizable spaces since $P$ is countable), the infinite sequence of colorings $(\tilde{\chi}_N)_{N=1}^\infty$ contains a convergent subsequence, converging to some limit coloring $\chi^* \in \mathcal{X}$. 

Finally, fix an arbitrary axis-parallel rectangle $R$ in the infinite grid. 
Its discrepancy depends only on the finite set of points $R \cap P$. 
Hence, the set of colorings in $\mathcal{X} = \{-1,+1\}^P$ for which $R$ satisfies its target discrepancy $\lambda_{m,i}$ is a closed (indeed clopen) cylinder set in the product topology. 
Since the limiting coloring $\chi^*$ is the limit of a sequence (or subnet) of colorings $\tilde{\chi}_N$ that satisfy this condition for all sufficiently large $N$ (as $R \cap P \subseteq V_N$ for sufficiently large $N$), it follows that $\chi^*$ itself satisfies the discrepancy bound for $R$. 
As $R$ was arbitrary, $\chi^*$ preserves the finite discrepancy bounds simultaneously for all axis-parallel rectangles.
\end{proof}

\subsubsection{The algorithm}
By the order-index reduction before, we work directly
in the \(\mathbb Z^d\) normalization of the discrete embedded model. Thus
strings are points of \(\mathbb Z^d\), and coordinate windows are
axis-parallel rectangles.

We begin by partitioning the space $\mathbb{Z}^d$ as follows.

Let $f$ be the function in Theorem \ref{thm:banachhighdimRecs2}.
\begin{lem}\label{lem:small-number-parts}
Let $f: \mathbb{R}_{\geq0} \to \mathbb{R}_{\geq0}$ be a non-decreasing function such that $\lim_{x \to \infty} f(x) = \infty$. There exists a strictly increasing sequence of real numbers $\{r_k\}_{k=0}^\infty$ with $r_0 = 0$ and $\lim_{k \to \infty} r_k = \infty$, such that if we define a partition of the lattice points $\mathbb{Z}^d$ as follows:
For each integer $k \geq 0$, the regions $\mR_k \subset \mathbb{Z}^d$ is defined as 
\begin{equation}
    \mR_k = \{x \in \mathbb{Z}^d \mid r_k \leq \|x\|_\infty < r_{k+1}\}. \label{eq:Rk}
\end{equation} 
Then for any integer $\ell \geq1$, any axis-parallel box $B \subset \mathbb{Z}^d$ of side length $\ell$ intersects at most $f(\ell) + 2$ regions.
\end{lem}
\begin{proof}
We first define a well-behaved lower bound for $f$ as below, whose proof is in Appendix. 
\begin{claim}\label{claim:g}
    Let $f: [0, \infty) \to [0, \infty)$ be a non-decreasing function such that $\lim_{x \to \infty} f(x) = \infty$. There exists a strictly increasing, strictly concave function $g: [0, \infty) \to [0, \infty)$ and a constant $x^* \geq0$ such that $\lim_{x \to \infty} g(x) = \infty$, and for all $x \geq x^*$, $g(x) \leq f(x)$.
\end{claim}

Let $G: \mathbb{R}_{\geq0} \to \mathbb{R}_{\geq0}$ be the inverse function of $g$. Because $g$ is strictly increasing and strictly concave, $G$ is strictly increasing and strictly convex. We define our region boundaries by $r_k = G(k)$ for all integers $k \geq0$. Note that $r_0 = G(0) = 0$. 

We may assume $B$ has the same side lengths along each axis, as otherwise we can enlarge $B$ by its maximum side length $\ell$. 
Let $B = \prod_{i=1}^d [a_i, a_i + \ell]$ with $a_i \in \mathbb{Z}$ be an arbitrary axis-parallel box of side length $\ell$ in $\mathbb{Z}^d$. Let $I$ be the real interval representing the range of the $\ell_\infty$ norm evaluated on points within $B$. Since the maximum variation of $\|x\|_\infty$ over the box $B$ is bounded by its side length $\ell$, $I$ is contained within some interval $[A, A + \ell]$ for a non-negative real number $A \geq0$.
The box $B$ intersects a region $R_k$ if and only if the interval $[A, A + \ell]$ contains some norm value $v$ such that $r_k \leq v < r_{k+1}$. Therefore, the number of distinct regions $B$ intersects is bounded above by $m + 1$, where $m$ is the number of boundaries $r_k$ (for $k \geq1$) that fall within $[A, A + \ell]$.

Let $\mathcal{K} = \{k \in \mathbb{Z}_{\geq 1} \mid r_k \in [A, A+\ell]\}$. If $\mathcal{K}$ is empty, $B$ intersects at most 1 region, which satisfies the claim since $f(\ell) \geq0$. If $\mathcal{K}$ is non-empty, let $k_{\min}$ and $k_{\max}$ be the minimum and maximum integers in $\mathcal{K}$, respectively. The number of boundaries in the interval is precisely $m = k_{\max} - k_{\min} + 1$. Because both $r_{k_{\max}}$ and $r_{k_{\min}}$ lie within an interval of length $\ell$, it must hold that: $r_{k_{\max}} - r_{k_{\min}} \leq \ell$. Substituting the definition $r_k = G(k)$ and expressing $k_{\max}$ as $k_{\min} + m - 1$, we obtain: $G(k_{\min} + m - 1) - G(k_{\min}) \leq \ell$.

Because $G$ is strictly convex on $\mathbb{R}_{\geq0}$ and $G(0) = 0$, $G$ is super-additive. Thus, for any $x, y \geq0$, $G(x+y) \geq G(x) + G(y)$. Setting $x = k_{\min}$ and $y = m - 1$, we have: $G(k_{\min} + m - 1) - G(k_{\min}) \geq G(m - 1)$. Combining these inequalities yields: $G(m - 1) \leq \ell.$
Because $g$ is strictly increasing, applying $g$ to both sides preserves the inequality: $m \leq g(\ell) + 1$.
The total number of intersected regions is at most $m + 1$. Therefore, the box $B$ intersects at most $g(\ell) + 2$ regions. Since $g(\ell) \leq f(\ell)$ by construction, $B$ intersects at most $f(\ell) + 2$ regions.
\end{proof}

Let $\mathcal{R} = \{\mR_k\}_{k=0}^\infty$ be the partition of $\mathbb{Z}^d$ where $\mR_0 = \{0\}$ and $\mR_k$ is defined in (\ref{eq:Rk}).

The algorithm is as follows, and can be viewed as an adaptation of the one in Theorem~\ref{thm:finiterank}. The structural tree $\mT$ remains unchanged, as the change in dimension does not affect the underlying topology.

Fix $\epsilon > 0$. Given a finite set of elements $A \subset \mathbb{Z}^d$, we define \[\chi(A) = \chi_{\epsilon}(A)\] to be the coloring $\chi: A \to \{-1, 1\}$ provided by Theorem \ref{thm:multi-discrepancy}.

\paragraph{Algorithm.}
Let $(s_t)_{t \in \mathbb{N}}$ be a sequence of increasing positive integers and $(z_n)_{n \in \mathbb{N}}$ an increasing sequence of positive integers. These will be related to the ``pod diameter" in the algorithm. 

Initialize the pod $P \gets \emptyset$. The pod $P$ represents the set of candidate strings that the algorithm prioritizes for output.

Suppose we have just completed time $t-1$ and are now at time $t$. The adversary reveals the input $a_t$. Recall that $L_{i_{t-1}}$ denotes the language identified by $\Acc$ at time $t-1$.

\begin{enumerate}
\item If $L_{i_t} = L_{i_{t-1}}$, define $PB(t)= L_{i_t}$. 

\item If $L_{i_t} \neq L_{i_{t-1}}$, define the pullback $PB(t)$ as  the common ancestor of $L_{i_{t-1}}$ and $L_{i_{t}}$ in $\mT$ after removing inconsistent languages if exists. If it does not exists, then $PB(t) = L_{i_{t}}$. 
\end{enumerate}
In either case, we update $P$ by adding the following elements. To this end, define the region parameter $k_t \in \mathbb{N}$ to be the smallest integer such that:
\begin{enumerate}
    \item $k_t > k_{t-1}+2 $
    \item $\log_2{k_t}$ is larger than the infinity norm of any input and output up to time $t$
    \item There are at least two regions $\mR_i, \mR_j$ with $i, j \leq k_t$ that contain no elements  used by either the adversary or the algorithm. 
\end{enumerate}
Add to $P$ all elements of $PB(t)$ that lie in the regions $\mathcal{R}_0, \dots, \mathcal{R}_{k_t}$.

Let the current pod be denoted by $P_t$.

Define $\chi_t : P_t \to \{+1,-1\}$ as follows. For $x \in P_{t-1}$, set $\chi_t(x) = \chi_{t-1}(x)$. For $x \in P_t \setminus P_{t-1}$ with $x \in \mathcal{R}_i$, define $\chi_t(x)$ according to the coloring $\chi$ restricted to $\mathcal{R}_i \cap (P_t \setminus P_{t-1})$.

Note that $w_t \in PB(t)$ and hence $w_t \in P_t$, so $\chi_t$ is well-defined at $w_t$. Suppose $w_t \in \mathcal{R}_k$ (in which case $k < k_t$).

If $w_t$ has not been previously selected by the algorithm, then the algorithm outputs an available element $o_t = x \in \mathcal{R}_k$ that is closest to $w_t$ in the $\ell_\infty$ norm and satisfies $\chi_t(x) = -\chi_t(w_t)$, if such an element exists. If no such $x$ exists, the algorithm outputs an arbitrary element in $PB(t) \cup L_{i_t}$.

If $w_t$ has already been selected by the algorithm, then the algorithm outputs an available element $o_t = x \in \mathcal{R}_k$ that is closest to $w_t$ in the $\ell_\infty$ norm. If no such element exists in $\mathcal{R}_k$, the algorithm outputs an arbitrary element in $PB(t) \cup L_{i_t}$.

 \subsection{Density guarantee}   

\paragraph{Validity guarantee.}
Lemma~\ref{lem:PBsubsetK} continues to hold, as its proof relies only on the pullbacks $PB(t)$ of $\Acc$. In particular, it implies that after some finite time, $PB(t) \subseteq K$. Consequently, the algorithm guarantees validity.

We now turn to proving the density guarantee, which builds on several lemmas and claims from the one-dimensional setting.

We first prove the following key lemma.

\begin{lem}\label{lem:changer}
There exists a finite time $T$ such that the following holds. For any region $\mathcal{R}_k$ that is first included in $P_t$ for some $t > T$, consider the sequence
\[
\mathcal{R}_k \cap P_t \;\subseteq\; \mathcal{R}_k \cap P_{t+1} \;\subseteq\; \cdots.
\]
Then this chain contains at most $r$ distinct sets, where $r$ is the Cantor--Bendixson rank of $\mathcal{X}$. Moreover, the sequence stabilizes at $\mathcal{R}_k \cap K$, where $K$ is the target language.
\end{lem}

\begin{proof}
    By the algorithm, $\mR_k \cap P_t = PB(t)\cap \mR_k$. For any $t' > t$, it satisfies
    \[ \mR_k \cap P_{t'} = (\mR_k \cap P_{t'-1}) \cup (PB(t') \cap \mR_k). \]

Therefore, $\mathcal{R}_k \cap P_{t'}$ is always a superset of $\mathcal{R}_k \cap P_{t'-1}$, and the two sets differ only if $P_{t'} \nsubseteq P_{t'-1}$. By the definition of the algorithm, if $L_{i_{t'-1}}$ remains consistent at time $t'$, then
\[
PB(t') \subseteq L_{i_{t'-1}} \subseteq PB(t'-1),
\]
which implies that $P_{t'} \subseteq P_{t'-1}$. Hence, $P_{t'} \nsubseteq P_{t'-1}$ can occur only when $L_{i_{t'-1}}$ becomes inconsistent at time $t'$.
As in the one-dimensional case, by Claim~\ref{claim:tree}, $\mathcal{R}_k \cap P_{t'}$ corresponds to the common ancestor (in $\mT$) of $PB(t), PB(t+1), \dots, PB(t')$, restricted to $\mathcal{R}_k$. It is also the common ancestor of $PB(t), L_{i_{t+1}}, \dots, L_{i_{t'}}$, restricted to $\mathcal{R}_k$.

By Claim~\ref{claim:rtimes}, the increasing chain
\[
\mathcal{R}_k \cap P_t \;\subseteq\; \mathcal{R}_k \cap P_{t+1} \;\subseteq\; \cdots
\]
can change at most $r$ times, where $r$ is the Cantor--Bendixson rank of $\mathcal{X}$.

The final claim follows from the fact that the adversary eventually enumerates all elements of $K$, and we are considering times beyond the point at which all pullbacks are contained in $K$.
\end{proof}

Let $T$ be the time given by Lemma~\ref{lem:changer}. Let $k_0$ be the smallest integer such that all elements revealed by either the algorithm or the adversary up to time $T$ lie in $\mathcal{R}_0 \cup \cdots \cup \mathcal{R}_{k_0-1}$.

As in the previous argument, when analyzing the lower Banach density, an averaging argument allows us to restrict attention to boxes that are sufficiently far from the origin. Accordingly, we consider only axis-parallel rectangles $B \subset \mathbb{Z}^d$ with side length $L(B) = \ell$ such that $B$ is disjoint from $\mathcal{R}_0 \cup \cdots \cup \mathcal{R}_{k_0-1}$, and $\ell$ is sufficiently large. This restriction is without loss of generality, since
\[
\ell \;\le\; \mathrm{Vol}(B) \;\le\; \ell^d.
\]

Suppose $B$ intersects $\mathcal{R}_{k'}$ but does not intersect $\mathcal{R}_{k'+1}$.

Fix any $k_0 \leq k \leq k'$, and consider the sequence $w_{t_1}, w_{t_2}, \dots \notin O$, where $t_1 < t_2 < \cdots$ are the time steps at which these elements are generated inside $\mathcal{R}_k$. If no such elements exist, we skip this value of $k$.

By Lemma~\ref{lem:changer}, let
\[
Z_1 \subsetneq Z_2 \subsetneq \cdots \subsetneq Z_{r'}
\]
denote the distinct nonempty sets attained by $\mathcal{R}_k \cap P_t$ over time, where $r' \leq r$.

Suppose $Z_1 = \mathcal{R}_k \cap P_t$ for all $t$ in some time interval $I$. Consider any $t_i \in I$. At time $t_i$, after the adversary reveals $w_{t_i}$, if there exists an available element in $Z_1$ whose coloring under $\chi_t$ is opposite to that of $w_{t_i}$, then the algorithm outputs such an element in $Z_1$.
By construction, whenever the algorithm outputs at time $t_i$ (i.e., it does not skip), the output lies in the same region as the adversarial input $w_{t_i}$. Moreover, since $\chi_t$ restricted to $Z_1$ remains fixed throughout the interval $I$, it follows that within $B \cap Z_1$, the number of adversarial inputs exceeds the number of algorithm outputs by at most
\begin{align}
   &  \left||B \cap Z_1 \cap \chi^{-1}(1)| - |B \cap Z_1 \cap\chi^{-1}(-1)|\right| \label{eq:discBox}, 
\end{align} 
as the algorithm pairs each adversarial input with an available element of opposite coloring (under $\chi$) whenever such an element exists.

\begin{claim}\label{claim:recpartition}
For any axis-parallel rectangle $B \subset \mathbb{Z}^d$, the intersection $B \cap \mR_k$ can be partitioned into at most $2d$ mutually disjoint axis-parallel rectangles.
\end{claim}
The proof of this claim is in Appendix.

As $B \cap \mR_k$ can be partitioned into $2d$ rectangles $B_1, B_2, \dots$, and each $B_i \cap Z_1$ satisfies the size-sensitive discrepancy bound (Theorem \ref{thm:multi-discrepancy}). Thus summand in (\ref{eq:discBox}) is at most 
Thus 
\begin{align}
 & \left||B \cap  Z_1 \cap \chi^{-1}(1)| - |B\cap  Z_1 \cap\chi^{-1}(-1)|\right| \nonumber \\
 \leq  &  \sum_{i}\left||B_i \cap Z_1 \cap \chi^{-1}(1)| - |B_i \cap Z_1 \cap\chi^{-1}(-1)|\right|. \label{eq:discBox2}
\end{align} 

For any of the rectangles $B_i$, by Theorem \ref{thm:multi-discrepancy}, if 
$\frac{c_0}{2^j} L(B_i)^\alpha \leq |B_i \cap Z_1| < \frac{c_0}{2^{j-1}} L(B_i)^\alpha$ and $j \geq \epsilon^2 |B_i \cap Z_1|$ or when $|B_i \cap Z_i| \leq C$, 
 we say these rectangles {\it sparse}. 
We choose $C = C(\epsilon,\alpha)$ such that, by Theorem~\ref{thm:multi-discrepancy}, any box $B_i$ that is not sparse has discrepancy at most $3\epsilon\,|B_i \cap Z_1|$. Consequently, $B_i$ is classified as sparse if either $|B_i \cap Z_1| \le C$ or
\[ |B_i \cap Z_1| < \frac{c_0}{2^{\epsilon^2 |B_i \cap Z_1|-1}} L(B_i)^\alpha,
\]
implying 
\begin{equation}
    |B_i \cap Z_1| \leq \max \left(C, \left(\alpha \log \ell + \log c_0 \right)/\epsilon^2 +1  \right) \label{eq:sparse}
\end{equation}
by the fact $|B_i \cap  Z_i| \geq1$ and $L(B_i) \leq \ell$. 
The total number of points in $Z_i$ in sparse rectangles is thus at most 
\[  2d  \max \left(C, \left(\alpha \log \ell + \log c_0 \right)/\epsilon^2 +1  \right)\]
where the first part $2d$ is by Claim \ref{claim:recpartition}, and the last part is by (\ref{eq:sparse}). For the rectangles which are not sparse, we have that the discrepancy is at most $3\epsilon |B_i \cap Z_1|$. Therefore (\ref{eq:discBox2}) is at most 
\begin{align*}
  & \sum_{i: B_i \text{ sparse}}\left||B_i \cap Z_1 \cap \chi^{-1}(1)| - |B_i \cap Z_1 \cap\chi^{-1}(-1)|\right|  \\ + & \sum_{i: B_i \text{ not sparse}}\left||B_i \cap Z_1 \cap \chi^{-1}(1)| - |B_i \cap Z_1 \cap\chi^{-1}(-1)|\right| \\
   \leq & \sum_{i: B_i \text{ sparse}} |B_i \cap Z_1|  + \sum_{i: B_i \text{ not sparse}}\disc_{\chi_{t_1}}(B_i \cap Z_1) \\
   \leq &  2d \max \left(C, \left(\alpha \log \ell + \log c_0 \right)/\epsilon^2 +1  \right) + \sum_{i: B_i \text{ not sparse}} 3\epsilon |B_i \cap Z_1|  \\
   \leq & 2d \max \left(C, \left(\alpha \log \ell + \log c_0 \right)/\epsilon^2 +1  \right) + 3\epsilon|B \cap Z_1|.
\end{align*}

Similar argument applied to the rest of the $Z_i$'s. Therefore we have that within $B \cap \mR_k$, the adversary input only elements is more than the algorithm generated elements by at most (\ref{eq:discBox}) adding up all $Z_i$ is at most 
\begin{align}
&  \sum_{1 \leq i \leq r'}  2d \max \left(C, \left(\alpha \log \ell + \log c_0 \right)/\epsilon^2 +1  \right) + \sum_{1 \leq i \leq r'}3\epsilon|B \cap Z_i| \nonumber \\
   = & 2d r \max \left(C, \left(\alpha \log \ell + \log c_0 \right)/\epsilon^2 +1  \right) +  3\epsilon|B \cap \mR_k|. \label{eq:key}
\end{align}

Summing over all $\mathcal{R}_k$ that intersect $B$, of which there are at most $f(\ell)+2$ by Lemma~\ref{lem:small-number-parts}, we obtain that within $B$, the number of adversarial-input-only elements exceeds the number of algorithm-generated elements by at most
\begin{align}
  & (f(\ell)+2) 2d \max \left(C, \left(\alpha \log \ell + \log c_0 \right)/\epsilon^2 +1  \right) + \sum_{k: \mR_k \cap B \neq \emptyset}3\epsilon|B \cap \mR_k | \nonumber \\
   \leq & (f(\ell)+2) 2d \max \left(C, \left(\alpha \log \ell + \log c_0 \right)/\epsilon^2 +1  \right) + 3\epsilon|B \cap K|. \label{eq:Z1}
\end{align} 

Since we have required $|B \cap K|  \gg \log (\ell) f(\ell)$ and when $\ell$ is sufficiently large, \[(f(\ell)+2) 2d r \max \left(C(\alpha, \epsilon), \left(\alpha \log \ell + \log c_0 \right)/\epsilon^2 +1  \right) = O( \log (\ell) f(\ell)),\]
so (\ref{eq:Z1}) is $(1+ o_\ell(1))3\epsilon|B \cap K|$. 
Therefore, for sufficiently large $\ell$, whenever
\[
|B \cap K| \;\gg\; \log(\mathrm{Vol}(B))\, f(\mathrm{Vol}(B)) \;\ge\; (\log \ell)\, f(\ell),
\]
where we use that $\mathrm{Vol}(B) \ge \ell$, the number of adversarial-input-only elements in $B$ exceeds the number of algorithm-generated elements by at most
$(1 + o_\ell(1)) \cdot 3\epsilon |B \cap K|.$
Consequently, the lower Banach density is at least $\tfrac{1}{2} - 3\epsilon$. This completes the proof of Theorem \ref{thm:banachhighdim} by starting with $\epsilon/3$ as parameter.

\section{Concluding Remarks}
This paper studies the tension between validity and breadth in language
generation in the limit. For averaged, prefix-based breadth, prior work gives
a universal positive answer: every countable language class admits the optimal
\(1/2\) lower-asymptotic-density guarantee. We show that the answer changes
when breadth is local. For some countable language classes, every eventually
valid generator must leave arbitrarily large sparse holes in the true language.

Furthermore, this paper shows that breadth in language generation has a hidden
structural layer. 
In one dimension, finite Cantor--Bendixson rank restores the optimal
\(1/2\) lower Banach-density guarantee, while our infinite-rank
construction shows that genuine zero-density examples can occur beyond
finite rank. In dimensions \(d\ge2\), ordinary Banach density reveals a
separate geometric barrier: zero density can be forced even for a
singleton language class. Filtered Banach density removes this purely
geometric barrier, and for finite-rank classes we obtain a
\(1/2-\varepsilon\) guarantee.

Thus Banach density is not merely a stronger density requirement. First, it is motivated by spatial embedding. Second, mathematically,  this is
the density notion under which the hidden topology and geometry of broad
valid generation become visible and the nature of the problem.

There are a number of interesting questions left open by this work, and as one line of such questions, we ask whether the insights arising from the analysis here might be able to help shed light on some of the related phenomena observed in practice with large language models, in which they fail to fully represent large parts of the space in which the language is embedded \cite{holtzman2020degeneration,shypula2025diversity,zhang2025verbalized}.
While it will require further formalization to begin fully exploring the possible connections here, the issues of breadth both in the language generation model and in the way language models are deployed in practice point to a promising direction of new questions as we try to bring these ideas together more directly.

\paragraph{Declaration of Generative AI and AI-assisted technologies in the writing process.} During the preparation of this work, the authors used Chat GPT in order to improve language and readability. After using this tool/service, the authors reviewed and edited the content as needed and take full responsibility for the content of the publication.

\bibliographystyle{plain}
\bibliography{AIbib}

\newpage
\section{Appendix}

\begin{claim}\label{claim:nicef}
An increasing function $f: \mathbb{N} \to \mathbb{N}$ is nice if and only if it is not universally bounded above.
\end{claim}

\begin{proof}

We show if $f$ is nice, then $f$ is unbounded: 
We argue by contrapositive. Assume $f$ is bounded above by some integer $M \in \mathbb{N}$, meaning $f(x) \leq M$ for all $x \in \mathbb{N}$. Let $T = M$. If $f$ were nice, there would exist sequences satisfying the stated conditions. By Condition 2, we must have $a_1 > M$. However, Condition 4 dictates that $a_1 \leq f(k_1) \leq M$. This yields the contradiction $M < a_1 \leq M$. Thus, $f$ cannot be nice.

Now suppose $f$ is unbounded, we show $f$ is nice. 
Because $f$ is unbounded, there exists an integer $k_1 \in \mathbb{N}$ such that
$f(k_1) > T.$
We define $a_1 = f(k_1)$. This choice immediately satisfies $a_1 > T$ (Condition 2) and $a_1 \leq f(k_1)$ (Condition 4).
Next, we define the second term of the $a$-sequence as
$a_2 = a_1 + k_1 + 1.$
This explicit construction satisfies Condition 3. Because $k_1 \geq1$, it also  satisfies Condition 2.
Since $f$ is unbounded, there exists an integer $k_2$ large enough such that $k_2 > k_1$ and
$f(k_2) \geq a_2$.
This choice satisfies $k_1 < k_2$ (Condition 1) and $a_2 \leq f(k_2)$ (Condition 4). 
 Therefore, $f$ is nice.
\end{proof}

\begin{claim}Let $f: [0, \infty) \to [0, \infty)$ be a non-decreasing function such that $\lim_{x \to \infty} f(x) = \infty$. There exists a strictly increasing, strictly concave function $g: [0, \infty) \to [0, \infty)$ and a constant $x^* \geq0$ such that $\lim_{x \to \infty} g(x) = \infty$, and for all $x \geq x^*$, $g(x) \leq f(x)$.
\end{claim}
\begin{proof} 
The proof follows a standard discretization and integration argument. We will construct $g$ by first explicitly constructing its inverse, a strictly increasing, strictly convex function $G: [0, \infty) \to [0, \infty)$. 
First, we define the generalized inverse of $f$. In other words, for any $y \geq 0$, let:
\[F(y) = \sup\{x \geq 0 \mid f(x) < y\}.\]
Because $f$ is non-decreasing and diverges to infinity, $F(y)$ is well-defined and non-decreasing. Furthermore by the definition of $F(y)$, clearly for any  $x \geq F(y)$, it must hold that $f(x) \geq y$.

We construct a sequence of values $G_n$ that will serve as the nodes of our convex function at integer points. We ensure strict convexity by defining a strictly increasing sequence of discrete derivatives $d_n$.
Let $G_0 = 0$ and $d_0 = 1$. For all integers $n \geq1$, we define $d_n$ recursively:
\[d_n = \max\{d_{n-1} + 1, \; 2(F(n+1) - G_{n-1}) - d_{n-1}\}.\]
We then define the nodes as $G_n = G_{n-1} + \frac{1}{2}(d_{n-1} + d_n)$.
By the first term in the maximum, $d_n \geq d_{n-1} + 1$,  guaranteeing the sequence $\{d_n\}$ is strictly increasing.
The second term guarantees $d_n \geq 2(F(n+1) - G_{n-1}) - d_{n-1}$, which rearranges to $\frac{1}{2}(d_{n-1} + d_n) +  G_{n-1} \geq F(n+1)$. Therefore, 
$G_n \geq F(n+1)$ for all $n \geq 1.$

Now, we extend these discrete nodes into a continuous, strictly convex function $G(y)$. Let $D(y)$ be the piecewise linear function interpolating the points $(n, d_n)$ for all integers $n \geq0$. Because $\{d_n\}$ is strictly increasing and $d_0 > 0$, $D(y)$ is continuous, strictly positive, and strictly increasing everywhere on $[0, \infty)$.
We define \[G(y) = \int_0^y D(t) dt.\] By the Fundamental Theorem of Calculus, $G'(y) = D(y)$. Since its derivative is strictly increasing, $G(y)$ is strictly convex. Furthermore, $\int_{n-1}^n D(t)  dt = \frac{1}{2}(d_{n-1} + d_n)$. Thus, $G(n) = G_n$, meaning our continuous function preserves the bound $G(n) \geq F(n+1)$.
Since $D(y) \geq d_0 = 1$, we have $\lim_{y \to \infty} G(y) = \infty$. We define $g: [0, \infty) \to [0, \infty)$ as the inverse function $g(x) = G^{-1}(x)$. Because $G$ is strictly increasing and strictly convex, its inverse $g$ is strictly increasing and strictly concave. Because $G \to \infty$, $g \to \infty$. It remains to prove the bound $g(x) \leq f(x)$ for sufficiently large $x$. For any $x \geq G(1)$, since $G(y)$ is continuous and strictly increasing to infinity, there exists a unique integer $n \geq1$ such that: $ G(n) \leq x < G(n+1).$
Because $g$ is strictly increasing, applying $g$ to the strict upper bound yields $g(x) < g(G(n+1)) = n+1$.
Applying the lower bound, we have $x \geq G(n)$. From our recursive construction, $G(n) \geq F(n+1)$. Therefore, $x \geq F(n+1)$. By the property of the generalized inverse $F$, this implies $f(x) \geq n+1$. Combining these inequalities, for all $x \geq G(1)$, we have $g(x) < n+1 \leq f(x)$. This completes the proof. 
\end{proof}

\begin{claim}\label{claim:recpartition}
For any axis-parallel rectangle $B \subset \mathbb{Z}^d$, the intersection $B \cap \mR_k$ can be partitioned into at most $2d$ mutually disjoint axis-parallel rectangles.
\end{claim}

\begin{proof}
In this proof because the inclusion of boundary is not important, so we will not pay attention to whether a set is closed or not.
Define the discrete $L_\infty$ balls $C_{in} = \{x \in \mathbb{Z}^d \mid \|x\|_\infty \leq r_k - 1\}$ and $C_{out} = \{x \in \mathbb{Z}^d \mid \|x\|_\infty \leq r_{k+1} - 1\}$. The offset by one accounts for the strict inequality in the definition of $R_k$ and the discrete nature of $\mathbb{Z}^d$. (When $r_k = 0$, $C_{in}$ is empty, and the argument simplifies accordingly.) Both $C_{in}$ and $C_{out}$ are axis-parallel hypercubes, and thus axis-parallel rectangles, in $\mathbb{Z}^d$. By definition, $\mR_k = C_{out} \setminus C_{in}$.

We rewrite the intersection as:
\[
B \cap \mR_k = B \cap (C_{out} \setminus C_{in}) = (B \cap C_{out}) \setminus C_{in}.
\]
Let $A = B \cap C_{out}$. Because the intersection of two axis-parallel rectangles is itself an axis-parallel rectangle, $A$ is an axis-parallel rectangle. The problem reduces to partitioning the set difference $A \setminus C_{in}$ into at most $2d$ mutually disjoint rectangles.

Let $\gamma = r_k - 1$, so the inner hypercube can be written as $C_{in} = \prod_{j=1}^d [-\gamma, \gamma] \cap \mathbb{Z}^d$. For any point $x \in A \setminus C_{in}$, since $x \notin C_{in}$, there exists at least one coordinate where $|x_j| > \gamma$. Let $i \in \{1, \dots, d\}$ be the smallest such index. We assign $x$ to a specific sub-rectangle based on this first violating coordinate and its sign.

Formally, for each dimension $i \in \{1, \dots, d\}$, define the sets $L_i$ and $U_i$ as:
\begin{align*}
    L_i &= A \cap \left( \prod_{j=1}^{i-1} [-\gamma, \gamma] \times (-\infty, -\gamma - 1] \times \prod_{j=i+1}^d \mathbb{Z} \right), \\
    U_i &= A \cap \left( \prod_{j=1}^{i-1} [-\gamma, \gamma] \times [\gamma + 1, \infty) \times \prod_{j=i+1}^d \mathbb{Z} \right).
\end{align*}
All intervals are understood to be in $\mathbb{Z}$. Because each $L_i$ and $U_i$ is defined as the intersection of $A$ with a Cartesian product of intervals in $\mathbb{Z}$, they are axis-parallel rectangles.

By construction, every point $x \in A \setminus C_{in}$ is assigned to exactly one $L_i$ or $U_i$ corresponding to its minimal violating coordinate, guaranteeing that the union covers $A \setminus C_{in}$. Disjointness follows because points in $L_i \cup U_i$ violate the $i$-th coordinate constraint, while points in $L_{i'} \cup U_{i'}$ for $i' > i$ satisfy all earlier coordinate constraints; moreover, $L_i$ and $U_i$ are disjoint by sign.

Thus, we have partitioned the set difference:
\[
A \setminus C_{in} = \bigcup_{i=1}^d (L_i \cup U_i).
\]
This union consists of exactly $2d$ mutually disjoint axis-parallel rectangles (some of which may be empty depending on the boundaries of $A$). Therefore, $B \cap R_k$ can be partitioned into at most $2d$ axis-parallel rectangles.
\end{proof}

Below is the statement of the Lov\'asz Local Lemma (see e.g. \cite{probbook}). 

\begin{thm}[Lov\'asz Local Lemma]\label{thm:LLL}
    Let $\mathcal{A} = \{A_1, A_2, \dots, A_n\}$ be a finite set of events in a probability space. Let $G = (V, E)$ be a dependency graph for the events, where $V = \{1, 2, \dots, n\}$. This means that each event $A_i$ is mutually independent of all the events $A_j$ such that $(i, j) \notin E$. Let $\Gamma(i)$ denote the neighborhood of $i$ in $G$ (the set of all $j$ such that $(i, j) \in E$). If there exist real numbers $x_1, x_2, \dots, x_n \in (0, 1)$ such that for all $i \in \{1, \dots, n\}$:\[\Pr(A_i) \leq x_i \prod_{j \in \Gamma(i)} (1 - x_j).\] Then the probability that none of the events in $\mathcal{A}$ occur is bounded below by:$$\Pr\left(\bigcap_{i=1}^n \overline{A_i}\right) \geq\prod_{i=1}^n (1 - x_i) > 0$$
\end{thm}

\end{document}